\documentclass[manuscript,screen]{acmart}

\AtBeginDocument{%
  }

\usepackage{amsthm}

\usepackage[linesnumbered, ruled]{algorithm2e}
\usepackage{graphicx}
\usepackage{adjustbox}
\usepackage{subcaption}
\usepackage{textcomp}
\usepackage{xcolor}
\usepackage{tabularx}
\usepackage{booktabs}
\usepackage{multirow}
\usepackage{textcomp}
\usepackage{mathtools}
 \usepackage{booktabs}

\usepackage{verbatim}
\usepackage{comment}

\usepackage{enumerate}
\usepackage{titlesec}
\usepackage{xcolor}
\usepackage{etoolbox}
\usepackage[normalem]{ulem}
\usepackage{xspace}
\newtheorem{theorem}{Theorem}%[section]

\newtheorem{lemma}[theorem]{Lemma}

\DeclarePairedDelimiter{\paren}{\lparen}{\rparen}
\newcommand{\sysname}{Baechi\xspace}
\newcommand{\pyname}{Baechi-PY\xspace}
\newcommand{\tfname}{Baechi-TF\xspace}
\newcommand{\pytorch}{PyTorch\xspace}
\newcommand{\tensorflow}{TensorFlow\xspace}

\newcommand{\wmsct}{\omega_{\textrm{m-sct}}}
\newcommand{\wsct}{\omega_{\textrm{sct}}}
\newcommand{\wopt}{\omega_{\textrm{opt}}}
\newcommand{\wmopt}{\omega_{\textrm{m-opt}}}

\newtheorem*{lemma*}{Lemma}

\newcommand{\ignore}[1]{\xspace}

\newbool{IsPrintComment}
\boolfalse{IsPrintComment}

\newbool{highlightNew}
\boolfalse{highlightNew}

\newbool{showOld}
\boolfalse{showOld}

\newbool{todoComment}
\boolfalse{todoComment}

\newbool{canAddComment}
\boolfalse{canAddComment}

\newcommand{\todo}[1]
{
    \ifbool{IsPrintComment}
    {%
        {\bf \color{red} TODO: #1}
    }{}
}

\newcommand{\canAdd}[1]
{
    \ifbool{canAddComment}
    {%
        {\bf \color{red} Can Add: #1}
    }{}
}

\newcommand{\igc}[1]
{
    \ifbool{IsPrintComment}
    {%
        {\bf \color{cyan} IG comment: #1}
    }{}
}

\newcommand{\igt}[1]
{
    \ifbool{IsPrintComment}
    {%
        {\color{magenta} IG Text: #1}
    }{#1}
}
\newcommand{\linda}[1]
{
 \ifbool{IsPrintComment}
    {%
        {\color{blue}#1}
    }{}
}
\newcommand{\lindac}[1]
{
 \ifbool{IsPrintComment}
    {%
        {\color{brown} linda comment: #1}
    }{}
}
\newcommand{\cx}[1]
{
    \ifbool{IsPrintComment}
    {%
        {\bf \color{teal} CX: #1}
    }{}
}

\newcommand{\bj}[1]
{
    \ifbool{IsPrintComment}
    {%
        {\bf \color{violet} BJ: #1}
    }{}
}

\newcommand{\cs}[1]
{
    \ifbool{IsPrintComment}
    {%
        {\bf \color{red} CS: #1}
    }{}
}

\newcommand{\cstodo}[1]
{
    \ifbool{todoComment}
    {%
        {\bf \color{red} TODO: #1}
    }{}
}

\newcommand{\xm}[1]
{
    \ifbool{IsPrintComment}
    {%
        {\bf \color{orange} XM comment: #1}
    }{}
}

\newcommand{\old}[1]
{
    \ifbool{showOld}
    {%
        {\bf \color{gray} #1}
    }{}
}

\newcommand{\newAdd}[1]
{%
    \ifbool{highlightNew}
    {%
        {\color{blue} #1}
    }{#1}
}

\newcommand{\before}[1]
{
    \ifbool{IsPrintComment}
    {%
        {\bf \color{green} BEFORE: #1}
    }{}
}

\newcommand{\bjc}[1]
{%
    {#1}%
}

\newcommand{\igr}[1]
{%
    {#1}%
}

\newcommand{\lilyYiTao}[1]
{
    \ifbool{IsPrintComment}
    {%
        {\color{orange} #1}
    }{}
}

\newcommand{\cxc}[1]
{
    \ifbool{IsPrintComment}
    {%
        {\color{teal} #1}
    }{}
}

\newcommand{\squishlist} %Indy
{
    \begin{list}{$\bullet$}
    {
        \setlength{\itemsep}{0pt}      \setlength{\parsep}{3pt}
        \setlength{\topsep}{3pt}       \setlength{\partopsep}{0pt}
        \setlength{\leftmargin}{1.5em} \setlength{\labelwidth}{1em}
        \setlength{\labelsep}{0.5em}
    }
}

\newcommand{\squishend}
{
    \end{list}
}

\makeatletter
\def\blfootnote{\xdef\@thefnmark{}\@footnotetext}
\makeatother

\newcommand{\sctThreshold}{0.1}

\definecolor{light-gray}{gray}{0.95}
\newcommand{\code}[1]{%\colorbox{light-gray}
{\texttt{#1}}}

\titleformat*{\section}{\LARGE\bfseries}
\titleformat*{\subsection}{\Large\bfseries}
\titleformat*{\subsubsection}{\large\bfseries}

\begin{document}

\title{\sysname: Fast Device Placement of Machine Learning Graphs}

%\email{webmaster@marysville-ohio.com}
%\affiliation{%
%  \institution{Institute for Clarity in Documentation}
%  \streetaddress{P.O. Box 1212}
%  \country{USA}

\author{Beomyeol~Jeon}
\affiliation{%
  \institution{University~of~Illinois~at~Urbana-Champaign}
  %\streetaddress{USA}
  \city{Urbana}
  \country{USA}}
\email{bj2@illinois.edu}

\author{Linda~Cai}
\affiliation{%
  \institution{Princeton~University}
  %\streetaddress{USA}
  %\city{Urbana}
  \country{USA}}
\email{tcai@princeton.edu}

\author{Chirag~Shetty}
\affiliation{%
  \institution{University~of~Illinois~at~Urbana-Champaign}
  %\streetaddress{USA}
  %\city{Urbana}
  \country{USA}}
\email{cshetty2@illinois.edu}

\author{Pallavi~Srivastava}
\affiliation{%
  \institution{University~of~Illinois~at~Urbana-Champaign*}
  %\streetaddress{USA}
  %\city{Urbana}
  \country{USA}}
%\email{email@uiuc.edu}

\author{Jintao~Jiang}
\affiliation{%
  \institution{University~of~Illinois~at~Urbana-Champaign*}
  %\streetaddress{USA}
  %\city{Urbana}
  \country{USA}}
%\email{email@uiuc.edu}

\author{Xiaolan~Ke}
\affiliation{%
  \institution{University~of~Illinois~at~Urbana-Champaign*}
  %\streetaddress{USA}
  %\city{Urbana}
  \country{USA}}
%\email{email@uiuc.edu}

\author{Yitao~Meng}
\affiliation{%
  \institution{University~of~Illinois~at~Urbana-Champaign*}
  %\streetaddress{USA}
  %\city{Urbana}
  \country{USA}}
%\email{email@uiuc.edu}

\author{Cong~Xie}
\affiliation{%
  \institution{University~of~Illinois~at~Urbana-Champaign*}
  \country{USA}}
%\email{email@uiuc.edu}

\author{Indranil~Gupta}
\affiliation{%
  \institution{University~of~Illinois~at~Urbana-Champaign}
  %\streetaddress{USA}
  %\city{Urbana}
  \country{USA}}
\email{indy@illinois.edu}

\blfootnote{This submission is an extended version of \textbf{"Baechi: Fast Device Placement of Machine Learning Graphs - Beomyeol Jeon, Linda Cai, Pallavi Srivastava, Jintao Jiang, Xiaolan Ke, Yitao Meng, Cong Xie, and Indranil Gupta. In Proceedings of the 11th ACM Symposium on Cloud Computing (Virtual Event, USA) (SoCC ’20). Association for Computing Machinery, New York, NY, USA, 416–430. \href{https://doi.org/10.1145/3419111.3421302}{https://doi.org/10.1145/3419111.3421302}"}. Document detailing the additional contributions has been attached as a supplementary material}

\blfootnote{*Work done while the authors were at University~of~Illinois~at~Urbana-Champaign, USA}
%#######################################################3

\renewcommand{\shortauthors}{Jeon et al.}

\begin{abstract}
Machine Learning graphs (or models) can be challenging or impossible to train when either devices have limited memory, or  models are large. %Splitting the model graph across  multiple devices, today,  largely relies on learning-based approaches to generate this placement. 
To split the model across devices,  learning-based approaches are still popular. 
While these result in model placements that train fast on data (i.e., low step times), learning-based model-parallelism  is time-consuming, taking many hours or days to create a placement plan of operators on devices.
We present the  \sysname system, the first to adopt an algorithmic approach to the placement problem for running machine learning training graphs on small clusters of memory-constrained devices. We integrate our implementation of \sysname into  two popular open-source learning  frameworks: TensorFlow and \pytorch. 
Our experimental results using GPUs 
show that: (i)   \sysname generates placement plans 
{$654\times$--$206$K $\times$} faster than  state-of-the-art learning-based approaches, and (ii)  \sysname-placed model's step (training) time is comparable to expert placements in PyTorch, and only up to 6.2\% worse  than expert placements in TensorFlow. 
We prove  mathematically that our two algorithms 
 are within a constant factor of the optimal. 
Our work shows that compared to learning-based approaches, algorithmic approaches can face  different challenges for adaptation to Machine learning systems, but also they 
offer proven bounds, and significant  performance benefits. 

\end{abstract}

%%
%% The code below is generated by the tool at https://dl.acm.org/ccs



\maketitle
% \begin{abstract}
% \input{sections/abstract}
% \end{abstract}

\section{Introduction} \label{introduction}

Distributed Machine Learning frameworks use more than one device in order to train large models and allow for larger training sets. 
This has led to multiple open-source systems, 
including   TensorFlow  \cite{tensorflow}, PyTorch \cite{pytorch}, MXNet \cite{distributed_ml_3}, Theano \cite{theano2016}, Caffe \cite{jia2014caffe}, CNTK \cite{seide2016cntk}, and  others~\cite{distributed_ml_1, distributed_ml_2, distributed_ml_4}. 
{Many of these systems use {\it data parallelism}, wherein each device (GPU) runs the  entire model, and multiple items are inputted and trained in parallel across devices.}

Yet, the increasing size of Machine Learning  (ML)  models and scale of training datasets is quickly outpacing available GPU memory. 
For instance the vanilla implementation of a 1000-layer deep residual network required  48 GB memory~\cite{chen2016training}, which is much larger than the amount of RAM available on a typical COTS (Commercial Off-the-Shelf) device. 
Even after further optimizations to reduce memory cost, the ML model still required 7 GB memory, making it impossible to run an entire model on a single device with limited memory, as well as  prohibitively expensive on  public clouds like 
{AWS~\cite{aws}, Google Cloud~\cite{google_cloud}, and Azure~\cite{azure}}.

At the same time, today, ML training is gravitating towards being run among small collections of {\it memory-constrained} devices. These include small groups of cheap devices like edge devices (for scenarios arising from Internet of Things and Cyberphysical systems), 
Unmanned Aerial Vehicles (UAVs or drones), and to some extent even mobile devices. For instance, real-time requirements~\cite{mahdavinejad2018machine,zeydan2016big}, privacy needs~\cite{bonawitz2017practical,bonawitz2019towards}, or budgetary constraints, necessitate training only using nearby or in-house devices with limited resources.

These two trends---increasing model graph sizes and growing prevalence of puny devices being used to train the model graph---together cause scenarios wherein a single device is insufficient and results in an Out of Memory (or OOM) exception. {For example, we found that the {Google Neural Machine Translation (GNMT)~\cite{gnmt}} model  OOMs on a 4 GB GPU even with 
conservative parameters:  batch size 128, 4 512-unit \igr{long short-term memory (LSTM)} layers, 30K vocabulary, and sequence length 50.} 

This problem is traditionally solved by adopting  {\it model parallelism}, wherein the ML model graph is split 
across multiple {devices.} 
Today, a popular way to accomplish  model parallelism in industry is to use {\it learning-based approaches} 
to generate the placement of operators on devices, most commonly by using Reinforcement Learning (RL) or variants. Significant in this space are works from Google~\cite{colocRL,hierarchicalRL} and the Placeto system~\cite{placeto}. 
A learning-based approach  learns iteratively (via RL) and adjusts the placement on the target cluster, with the goal of minimizing 
{execution time for each training step in the placed model, i.e., its {\it step time}.} 

While learning-based approaches achieve step times around those obtained by expert placements, they can unfortunately  take an inordinately long time to generate their  placement plans. 
{For instance, using one state-of-the-art learning-based approach~\cite{placeto}, 
NMT models require 94,000 steps during the learning-based placement, and even with a   conservatively low estimate of runtime per learning step of 2.63 seconds, the total placement time would come to 68.67 hours. } One might possibly apply parallelization techniques \cite{jia2018data, jia2018exploring, wang2019tofu} to the learning model being used to perform  placement, in order to speed it up, but the total incurred resource costs would stay just as high---hence, parallelization is orthogonal to our discussion. 

Such long waits are cumbersome and even  prohibitive at model development time, when the software developer needs to make many quick and ad-hoc deployments~\cite{placeto}. 
In fact, studies of analytics clusters reveal that most analytics job runs tend to be short~\cite{alibaba2019,cloudera_workload}.  %most users prefer small  jobs~\cite{alibaba2019,cloudera_workload}. 
For instance, 
the step time for a typical model graph (e.g., NMT or  Inception-V3), to train on a single data batch, is O(seconds) on a typical GPU. Overwhelming this time with learning-based placement times which span hours,  significantly inhibits the developer's agility.%} 

Additionally, a learning-based placement run works only for a target cluster and {a given model graph with fixed hyperparameters (e.g., batch size, learning rate, etc.).}
If the model graph were to be transitioned to a different cluster with different GPU specs, the learning has to be repeated all over again, incurring the high overhead. Consider a developer who is trying to find the right batch size for a target cluster. This process of exploration is iterative, and every hyperparameter value trial needs a new run of the learning-based technique, making the overall undertaking slow. 

For the model development process to be agile, nimble, and at the same time coherent with  future real deployments, what is needed is a new class of placement techniques for model parallelism, that: i) are significantly faster in placement than learning-based approaches, and yet ii) achieve fast step times in the placed model.

{
This paper is the first to adopt a traditional {\it algorithmic} approach for the placement  of ML 
models   on memory-constrained clusters. Subsequent to our initial work~\cite{baechiconf}, a few other authors have published algorithmic or dynamic programming-based ideas for model placement, however  these are either
: i) standalone and not integrated into open-source systems~\cite{amarpaper}, %or
ii) or they are aimed at only placing specific models like transformers~\cite{terapipe}, \cite{megatron}, or iii)  they are at best comparable %demonstrably inferior 
in performance to ours~\cite{pesto}. 
Orthogonal to model parallelism is pipeline parallelism~\cite{pipedream},\cite{gpipe},\cite{dapple},\cite{pipemare}---our paper does not explore the latter, in order to keep our discussion focused on the benefits of algorithmic approaches over learning approaches.
}

The contributions of this paper are:
\squishlist
\item We adapt
classical literature from parallel job scheduling to propose two memory-constrained algorithms, called {\it m-SCT (memory-constrained Small Communication Times)} and {\it m-ETF (memory-constrained Earliest Task First)}. We also present \igr{\it m-TOPO (memory-constrained TOPO-logical order)}, a strawman. \igr{We focus on the static version of the problem.} 
\item We prove that under certain  assumptions,  both m-ETF and m-SCT steps time is within a constant factor of the optimal. 
\item \igr{We present the {\it \sysname} system ({Korean  for {\it placement}, pronounced ``Bay-Chee''}). \sysname incorporates m-SCT/m-ETF into both TensorFlow as well as \pytorch. Our exposition focuses on the multiple design decisions that were needed in \sysname{} to derive performance out of the algorithmic underpinnings. 
}
\item 
We present tailored integration of \sysname{} with both  \tensorflow and \pytorch to address the different   programming abstractions and architecture in these two frameworks.

\item We present experimental results from a real deployment on a small cluster of GPUs, using both TensorFlow and PyTorch which show that \sysname generates placement plans in time 
{$654\times$--$206$K $\times$} faster than  today's learning-based approaches, and yet the placed model's step time (training time) is {either faster than or, at worst, only up to 6.2\% higher, compared to  expert-based placements. }

\squishend
\section{New Algorithms for Memory-Constrained Placement} \label{Methodology}

This section presents the problem formulation and our  three placement techniques. For each technique, we first discuss the classical approach (not memory-aware), and then  our adapted memory-constrained algorithm. Where applicable, we prove optimality.

Our three approaches are: 1) {a placer based on  topological sorting (TOPO)}
2) a placer based on Earliest Task First (ETF), and 3) a placer based on Small Communication Time (SCT).

\subsection{Problem Formulation}
\label{Problem formulation} 

Given $n$  devices (GPUs), each with memory size $M$, and a Machine Learning (ML) graph $G$ that is a DAG (Directed Acyclic Graph) of operators, the device placement problem is to place nodes of $G$ (operators)  on the devices so that the makespan is minimized. Makespan, equivalent to step time, is traditionally defined as the total execution time to train {on} one {input mini-batch (i.e., unit of training data).}  
Table~\ref{fig:terminology} summarizes key terms used throughout the paper. When discussing classical algorithms, we use the classical terms ``tasks'' instead of operators. 

If one assumes devices have infinite (sufficient) memory,  {\it scheduling with communication delay}  is a well-studied theoretical problem.
The problem is NP-hard even when under the simplest of assumptions \cite{complexity}, such as infinite number of {devices }%processors
and unit times for computation and communication (UET-UCT). 

\igr{Out of the three best-performing solutions to the infinite memory problem, we choose the two that map best 
to ML graphs:   1) {\it Earliest Task First or ETF}~ \cite{algo_lit_review, etf}, and 2)  
 {\it Small Communication Time  or SCT}~\cite{sct}. SCT is provably close to optimal when the ratio of maximum communication time between any two {tasks} %nodes
 to  minimum computation time for any {task }%node 
 is $\leq 1$. 

We excluded a third solution, UET-UCT~\cite{uet_uct_lp}, since it assumes unit computation and communication times, but ML graphs have heterogeneous operators. 

}

\begin{table}[tb!]
\caption{
{
\it {\bf Terms and Notations.} Used in the Paper.
}
}
% \scriptsize{
%\footnotesize{
\centering
\begin{tabularx}{\linewidth}{lX}
 \hline
 $G$ & Machine Learning graph to be placed \newline (Classical: Dependency graph of tasks to be placed) \\
 \hline
 $m$ & Number of operators (or tasks) in $G$\\
 \hline
 $n$ & Number of devices in a cluster \\
 \hline 
 $M$ & Memory available per device  \\
 \hline
 $d_i$ & Size of memory required by operator (task) $i$\\
 \hline
 $k_i$ & Computation time of operator (task) $T_i$\\
 \hline 
 $c_{ij}$ & Communication time of the output of operator $T_i$\\
 \hline
 $\rho$ & Ratio between maximum operator-to-operator (task-to-task) communication time and minimum per-operator (per-task) computation time\\
 \hline 
 SCT assumption & Small communication time assumption: Ratio between maximum operator-to-operator (task-to-task) communication time and minimum per-operator (per-task) computation time is $\leq 1$\\
 \hline 
 makespan & Training time for one data mini-batch, i.e., runtime for executing a ML graph on one input mini-batch\\
 \hline
\end{tabularx}

\label{fig:terminology}
\end{table}

\subsection{m-TOPO: Topological Sort Placer}

\label{topo}

\noindent{\it \textbf{Background: Topological Sort (Not Memory-Aware). }}
Topological sort \cite{topo} is a linear ordering of vertices in a DAG, such that for each directed edge $u \rightarrow v$, $u$ appears before $v$ in the linear ordering. 

\smallskip\noindent{\it \textbf{New Memory-Constrained Version (m-TOPO). }}
Our modified version, called {\it m-TOPO}, works as follows. It calculates the maximum load-balanced memory that will be used per device, by dividing total required memory by number of devices, and then accounting for outlier operators. Concretely, this per-device cap is  $Cap = (\sum_{i \in [m]} d_i / n + \max_{i \in [m]} d_i)$. Then m-TOPO works iteratively, and assigns operators to devices in increasing order of device ID. 
For the current device, m-TOPO places operators until the device memory usage reaches  $Cap$. 
At that point, m-TOPO moves on to the next device ID, and so on. 
At runtime, m-TOPO executes the operators at a device in the topologically sorted order 
\subsection{m-ETF: Earliest Task First Placer} \label{sec:etf}

\noindent{\it \textbf{Background: ETF (Not Memory-Aware). }}
ETF~\cite{etf} maintains two lists: a sorted task list $T$ containing unscheduled tasks, and a device list $P$. In $T$, tasks are sorted by {\it earliest schedulable time}. The earliest schedulable time of task $i$ is the latest finish time of $i$'s parents in the DAG, plus time for their data to reach $i$. In $P$, each device is associated with its {earliest available time}, i.e., last finish time of its assigned tasks (so far).  

ETF iteratively: i) places the head of the task queue $T$ at that device from $P$ which has the earliest available time, ii) then updates the earliest available time of that device to be the completion time of the placed task, and iii) updates earliest schedulable time of that task's dependencies in queue $T$ (if applicable).

\igr{The earliest schedulable time of task $j$ on device $p$ is the later of two times: (i) device $p$'s earliest available time ($free(p))$, and (ii) all predecessor tasks of $j$ have completed {\it and} have communicated their data to $j$. More formally, let: a) $\Gamma^{-}(j)$ be the set of $j$'s predecessors; b) for $i$: start time is $s_i$, computation time is $k_i$; c) $x_{ip}=0$ when task $i$ is on device $p$, otherwise $x_{ip}=1$; %.
d) commmunication time from task $i$ to $j$ is $c_{ij}$. 
Then, the earliest schedulable time of task $j$ across all devices is:
}
{
\begin{equation}
\label{eq:schedulable_time}
\displaystyle \min_{p \in P} \Big[\max\big(free(p), \max_{i \in \Gamma^{-}(j)} (s_i + k_i + c_{ij} x_{ip})\big)\Big]. 
\end{equation}
}

Under the SCT assumption (Table \ref{fig:terminology}),  
ETF's makespan was shown \cite{etf} to have an approximation ratio of  $(2 + \rho - \frac{1}{m})$ within optimal, 
where $\rho$ is the ratio of the maximum communication time to minimum computation time, and $m$ is the number of devices.
%} 
This approximation ratio approaches 3 when $\rho$ approaches 1 and {$m \gg 1$.} 

\smallskip\noindent{\it \textbf{New Memory-Constrained Version (m-ETF). }}

Our new modified algorithm, called {\it m-ETF}, maintains a queue $Q$ of operator-device pairs $(i,p)$ for all unscheduled operators and all devices. Elements $(i,p)$ in $Q$ are sorted in increasing order of the earliest schedulable time for operator $i$ on device $p$. This earliest schedulable time takes into account dependent parents of $i$ as well as the earliest time that device $p$ is available, given operators already scheduled at $p$. The reader will notice that m-ETF's modified queue can also be used for ETF--the key reason to use $(i,p)$ pairs is for m-ETF to do fast searches, since the earliest available device(s) may have insufficient memory.

m-ETF iteratively looks at the head of the queue. If the head element $(i,p)$ is not schedulable because device $p$'s leftover memory is insufficient, then the head is removed. If the head is schedulable, then operator $i$ is assigned to start on device $p$ at that earliest time, and we: i) update $p$'s earliest available time to be the completion time of $i$, and ii) update $i$'s  child operators' earliest schedulable times in queue $Q$ (if applicable).

\subsection{m-SCT: Small Communication Time Placer} \label{sec:sct_algo}

\noindent{\it \textbf{Background: SCT (Not Memory-Aware). }}
\label{sct:naive} The classical SCT algorithm \cite{sct} first develops a solution assuming an infinite number of available devices, and then specializes for a finite number of devices. 
We elaborate details, as they are relevant to our memory-constrained version.

\smallskip\noindent{\it \textbf{Classical SCT: Infinite Number of  Devices. }} SCT uses  integer linear programming (ILP). The key idea is to find the {\it favorite child} of a given task $i$, and prioritize its scheduling on the same device as task $i$. For a task $i$, denote its {\it favorite child} as $f(i)$.

\igr{The original ILP specification from \cite{sct} solves for variables $x_{ij} \in \{0,1\}$, where $x_{ij}=0$ if and only if $j$ is $i$'s favorite child. 

For completeness, we provide this full ILP specification below~\cite{sct} (Section 3.2 in that paper). %\igc
{Below, the machine learning graph is  $G=(V,E)$; and $i,j$ refer to operators.} 
}

\begin{equation}
\left\{ \begin{array}{ll}
\min{w^{\infty}} & \mbox{Minimize makespan.}\\
% \forall i \rightarrow j \in E(G),\ x_{ij} \in \{0,1\} &  x_{ij} = 0 \mbox{ when $j$ is $i$'s} \\
% &\mbox{favorite child}\\
\forall i \rightarrow j \in E(G),\ x_{ij} \in \{0,1\} &  x_{ij} = 0 \mbox{ when $j$ is $i$'s favorite child.}\\
\forall i \in V(G),\ s_i \geq 0 & \mbox{All tasks start after time=0.}\\
% & \mbox{time=0}\\
\forall i \in V(G),\ s_i + k_i \leq w^{\infty} & \mbox{All tasks should complete before makespan.}\\
% \forall i \in V(G),\ s_i + k_i \leq w^{\infty} & \mbox{All tasks should comp-}\\
% & \mbox{lete before makespan}\\

\forall i \rightarrow j \in E(G),\ s_i + k_i + c_{ij}x_{ij} \leq s_j & \mbox{Given edge }i \rightarrow j \mbox{, then $j $must start after $i$ completes.}\\
& \mbox{If on different devices, communication cost should be  added.}\\

\forall i \in V(G), \displaystyle \sum_{j \in \Gamma^{+}(i)} x_{ij} \geq |\Gamma^{+}(i)|-1 & \begin{array}{@{}ll}\mbox{Every task has at most one favorite child.}\end{array}\\
\forall i \in V(G), \displaystyle \sum_{j \in \Gamma^{-}(i)} %x_{ij}
x_{ji}
\geq |\Gamma^{-}(i)|-1 & \begin{array}{@{}ll}
\mbox{Every task is the favorite child of at most one predecessor.}\\  
\end{array}
\end{array} \right.
\label{main_eqn}
\end{equation}

We modify the above as follows. 
We allow $x_{ij}$ to take any real value between $0$ and $1$, thus making the ILP a relaxed LP. 
This can be solved in polynomial time using the interior point method \cite{interior_point}. Then the SCT algorithm simply rounds the LP solution $x_{ij}$ to be $1$ if $x_{ij} \geq \sctThreshold$, setting it to $0$ otherwise. 
$x_{ij}$ can be used to determine the favorite child of each task: $j$ is $i$'s favorite child if and only if after rounding, {$x_{ij} = 0$.}

This infinite device algorithm's makespan was shown~\cite{sct} to achieve an approximation ratio $\frac{2 + 2\rho}{2+\rho}$ 
{
within optimal, where $\rho$ is the ratio of the maximum communication time to the minimum computation time. 
}

We note that the ILP has a meaningful LP relaxation if and only if: (i) infinite number of devices are available, and (ii) the SCT assumption is satisfied, i.e.,  the ratio of the maximum inter-task communication time to the minimum task computation time is $\leq 1$. %\igc{
Nevertheless, even if this assumption were not true for an ML graph and devices, we show later that SCT can still be attractive.%} 

\medskip\noindent{\it \textbf{Classical SCT: Extension to Finite Number of Devices. }}
\label{algo:sct:finite}

For a finite number of devices, SCT schedules tasks similar to ETF \cite{etf}, but: a) prefers placing the favorite child of a task $i$ on the same devices as $i$ (each task has at most one favorite child, and at most one favorite parent), and b) prioritizes urgent tasks, i.e., a task that can begin right away on any device.

It was proved that SCT's makespan has an approximation ratio of $(\frac{4 + 3\rho}{2 + \rho} - \frac{2 + 2 \rho}{m (2 + \rho)})$ within optimal {\cite{sct}}, which is strictly better than ETF's (Section \ref{sec:etf}).
For instance, when $\rho$ approaches 1 and $m \gg 1$, then SCT is within $\frac{7}{3}$ of optimal while  ETF is 3 times worse than optimal.  

\medskip\noindent{\it \textbf{New Memory-Constrained Version (m-SCT). }}
\label{sct:mem_constraint}
Our proposed memory-constrained algorithm, called {\it m-SCT}, works as follows. First, m-SCT determines scheduling priority and selects devices in the same way as the  finite case SCT algorithm just described.  %(Section~\ref{algo:sct:finite}). 
Second, 
{when a device $p$ runs out of available memory, }
m-SCT excludes $p$ from future operator placements. 

In spite of the seemingly small difference, Figure~\ref{fig:msct_vs_sct} shows that m-SCT can succeed where SCT fails. 
SCT achieves a makespan of 8 time units with infinite memory but OOMs for finite memory. With finite memory, m-SCT succeeds and increases  makespan to only 9 time units.

\begin{figure}[t]
\centering
\includegraphics[width=\linewidth]{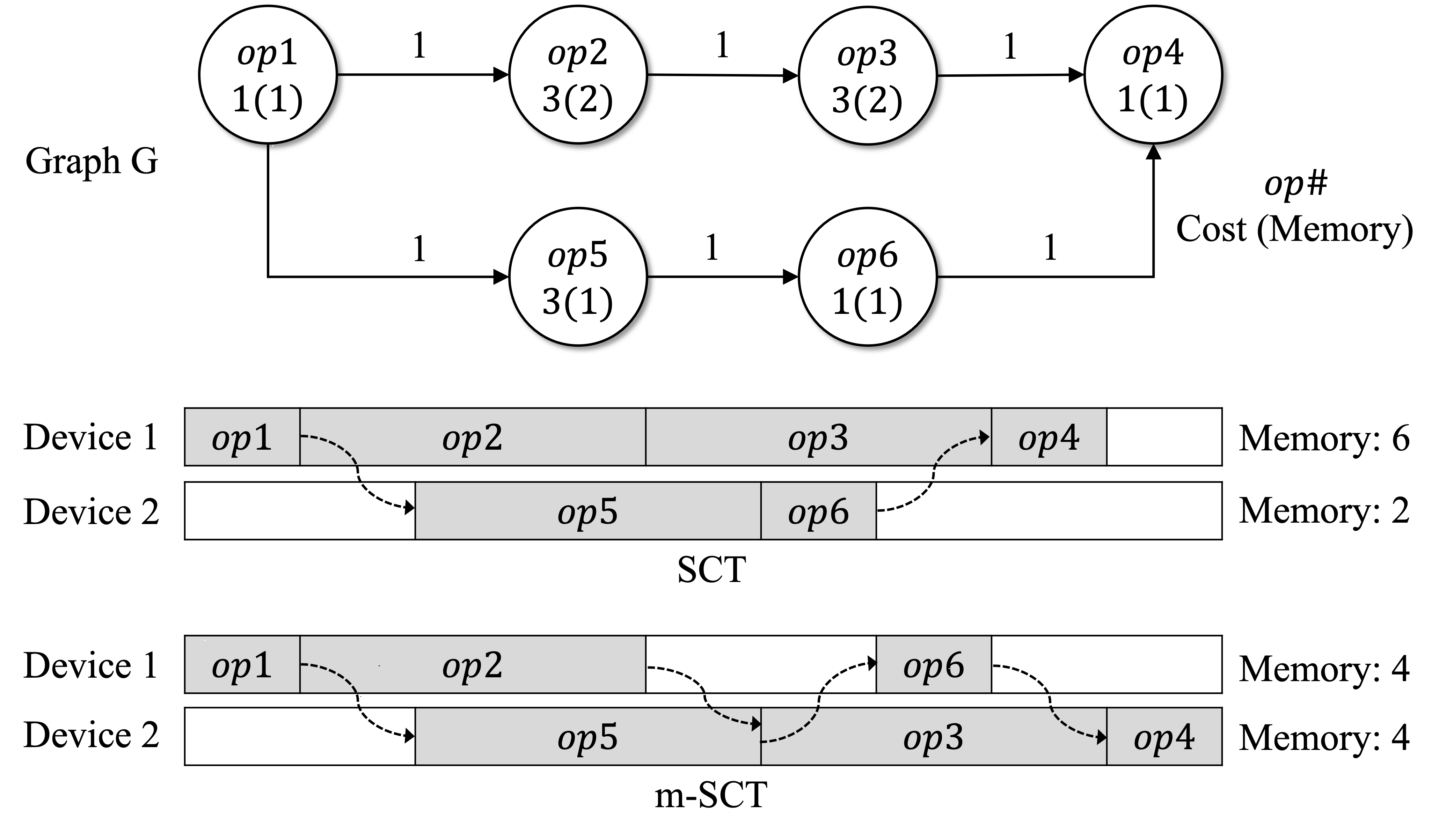}
 \caption{ {\it {\bf Classical SCT vs. m-SCT.} When per-device memory is limited to 4 memory units, SCT OOMs but m-SCT succeeds. m-SCT's training time (makespan) is only slightly higher (9) than SCT with infinite memory (8). Dashed arrows show data transfers.}}
  \label{fig:msct_vs_sct}

 \vspace*{-0.4cm}
\end{figure}

\subsection{Optimality of m-ETF and m-SCT}

Classical ETF and SCT were originally proposed to schedule a DAG of tasks on $n$ processors. The original work~\cite{sct,etf} derived upper bounds on the makespan achieved by them. Processor memory was {\it not} considered in that original formulation, and infinite memory was assumed. 

On the contrary, in ML model training, each device has a memory constraint. Impact of memory on scheduling is further pronounced due to the persistent memory that each task requires. Consequently, the schedules obtained by m-SCT/m-ETF can differ significantly from the SCT/ETF schedules (i.e.,  without memory constraint)---Fig.~\ref{fig:msct_vs_sct} shows an example. It is not clear how much worse m-SCT/m-ETF makespan is, compared to the optimal. Thus we derive upper bound on makespan of both m-ETF and m-SCT by extending the proofs in \cite{etf} and \cite{sct} to the memory constrained case. We show that m-SCT/m-ETF makespans are within a constant factor of the optimal makespan.

Because the proofs for optimality of m-ETF and m-SCT are involved, we show them in Appendix~\ref{appendix:etf} and Appendix~\ref{appendix:sct} respectively. We summarize our results and approach here: \\
{\noindent{\bf Result 1:} The completion time of m-ETF under realistic communication cost and limited memory, is within a known factor of the optimal schedule possible under zero communication cost but with infinite memory.\\
\noindent{\bf Result 2:} The completion time of m-SCT under realistic communication cost and limited memory, is within a known factor of the optimal schedule under infinite memory.\\
}
 Intuitively, our derived upper bounds are proportional to the ratio $\frac{n}{r}$,where $n$ is the total number of devices available, and $r$ is the number of devices out of $n$ that still have spare memory after all the tasks have been placed in a way that ``fills up'' memory device by device. 
 Note that $\frac{n}{r}>1$, otherwise the problem is unsolvable. Intuitively, a smaller $\frac{n}{r}$ (i.e., larger $r$) indicates looser memory constraint and thus better makespan. As $\frac{n}{r}$ approaches 1, m-SCT's solution (respectively m-ETF's) starts to approach that of SCT (respectively ETF). 

\medskip\section{\sysname Design}

\label{sec:design_main}

{This section describes how we implement \sysname{} in a way that works modularly with \tensorflow \cite{tensorflow} as well as \pytorch~\cite{pytorch}, two popular open-source learning platforms originally developed by Alphabet and Meta respectively. 

}

At a high level---for both target systems, \sysname first creates a computation graph of the input model, where each node is annotated with its memory requirements and time to complete. This graph is then fed to the chosen algorithm (Section~\ref{Methodology}'s m-SCT, m-ETF, or m-TOPO) to generate the placement. Finally, training is automatically executed with the given placement and without requiring 
the developer
to modify the code for the model.

{
However, because of two key differences in  abstractions and architectures between \tensorflow and \pytorch, \sysname's design for each is slightly different. 
} First, the ``nodes'' in the computation graph are {\it operators} in \tensorflow while in \pytorch they are {\it modules}. The former are fine-grained mathematical operations on tensors, while the latter are coarser structures similar to  classes in object-oriented languages. 
Second, in \tensorflow a model is a static graph of operators, while in  \pytorch, the  computation graph is constructed only during the forward run. Because of foreknowledge of the graph, \tensorflow can automatically insert rendezvous operators \cite{rnode} for cross-device communication. However,  \pytorch does not automatically insert these essential   cross-device communication primitives. { It requires \pytorch developers to write explicit code that moves  tensors across devices during execution.}  
A side benefit of our work is the automatic generation of these communication primitives. 

{In the remainder of this paper, we refer to the integration of \sysname into  \tensorflow as {\it \tfname}, and \sysname's integration into \pytorch as  {\it \pyname}.}

Next, we describe our techniques and optimizations for \tfname  in Section~\ref{sec:tfdesign}, and then  additional changes and differences required for \pyname design in Section~\ref{sec:pytdesign}.

\subsection{Design of \tfname 
}

\label{sec:design}
\label{sec:tfdesign}

To work with \tensorflow, \sysname{} needs to address four challenges: 

1) Satisfying TensorFlow's colocation constraints,
2) Minimizing Data Transfer via Co-Placement, 
3) Optimizations to reduce the number of operators to be placed, and 
4) Accommodating Sequential and Parallel Communications. 
\igr{\sysname solves these using a mix of both new ideas (Sections~\ref{sec:tf_colocation},\ref{sec:operator_fusion},\ref{sec:sequential_data_transfer}) and ideas similar to past work (Sections~\ref{sec:co-placement},\ref{sec:operator_fusion}).}

\medskip\noindent{\it \textbf{Working Example. }}
We use Figure~\ref{fig:example_graph} as a working example throughout this section. It is a simplified TensorFlow graph  for linear regression training with stochastic gradient descent (SGD).

\subsubsection{TensorFlow Colocation Constraints}
\label{sec:tf_colocation}

The first challenge arises from the fact that TensorFlow {(TF)} {\it requires} certain operators to be colocated. 
For instance, TensorFlow offers a  variable {operator}, \texttt{tf.Variable}, 
which is used to store persistent state such as an ML model parameter. The assignment and read operators of a variable are implemented as separate operators in TensorFlow, but need to be placed on the same device as the  variable operator.  TensorFlow represents this placement requirement as a {\it colocation group} involving all these operators. E.g., in Figure~\ref{fig:example_graph} there are two colocation groups: one containing \texttt{Weight} and \texttt{ApplyGrad}, and another containing \texttt{Step} and \texttt{UpdateStep}.

\sysname's initial placement (using the algorithms of Section~\ref{Methodology})  ignores colocation requirements. 
{Our first attempt was to {\it post-adjust  placement}, i.e., to ``adjust'' the  device placement, which was generated ignoring colocation, by ``moving'' operators from one device to another, in order to satisfy TF's colocation constraints. We explored multiple post-adjustment approaches including: i) preferring the device on which the compute-dominant operator in the group is placed,
ii) preferring the device on which the memory-dominant operator in the group is placed, and iii) preferring the device on which a majority of operators in the group are placed.  
We found all these three approaches produced   inconsistent performance gains, some 
giving step times up to  406\% worse than the expert. 
We concluded that post-adjusting was not a feasible design pathway. 
}

\igr{\sysname's novel contribution is to  {\it co-adjust placement}, using colocation constraint-based grouping  {\it while} creating the schedule. (In comparison, e.g., ColocRL~\cite{colocRL} groups {\it before} placement.) 
}
Concretely, whenever \sysname places the {\it first} operator from a given colocation group, all  other operators in that group are immediately placed on that same device. \sysname tracks the available memory on each device given its assigned operators. If the device cannot hold the entire colocation group, then \sysname moves to the algorithm's next device choice. 
We found this approach the most effective in practice, and it is thus the default setting in \sysname.

\begin{figure}
    \centering
    \begin{minipage}[t]{0.45\textwidth}
        \centering
        \caption{{\it {\bf Working Example.}  ML Graph for Linear Regression.}}
        \vspace{10pt}
        \includegraphics[width=1\linewidth]{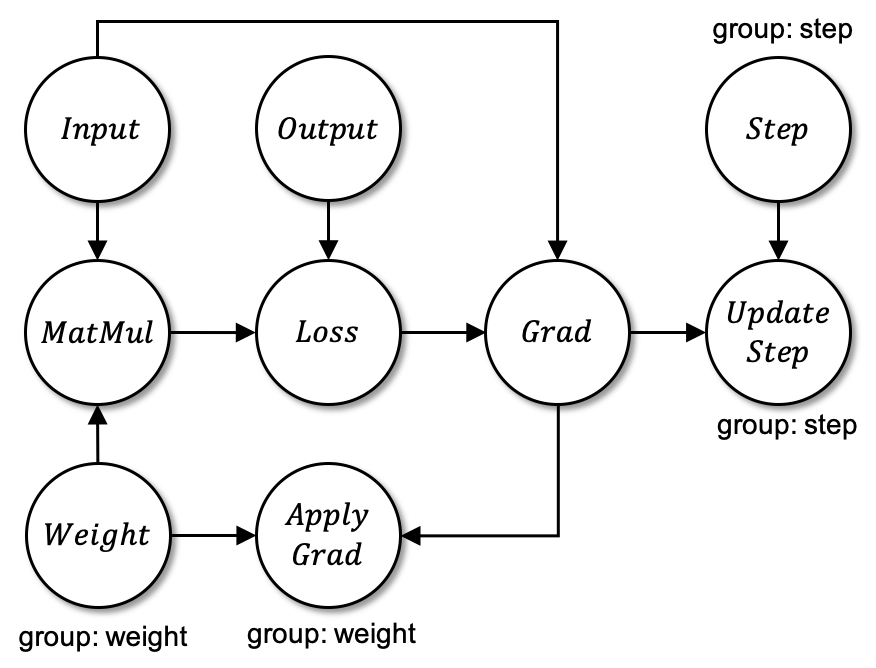}
        \label{fig:example_graph}
    \end{minipage}\hfill
    \begin{minipage}[t]{0.45\textwidth}
        \centering
        \caption{{\it {\bf Co-Placement.} Subgraph of {\tt tf.tensordot} Generating Data Transfers by m-ETF.}}
        \vspace{35pt}
        \includegraphics[width=1\linewidth]{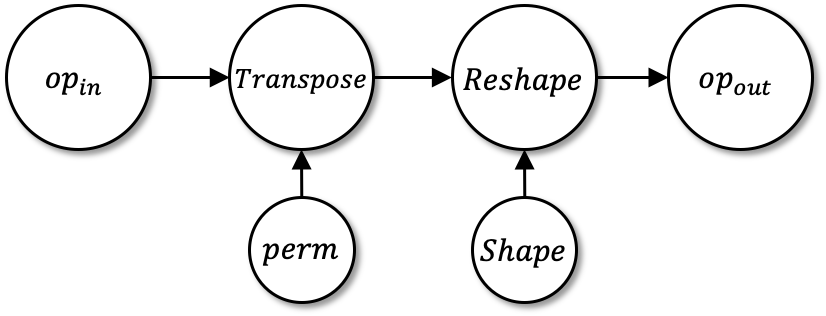}
        \label{fig:min_data_transfer}
    \end{minipage}
\end{figure}

\subsubsection{Co-Placement Optimization}
\label{sec:co-placement}

Different from TensorFlow's colocation constraints  (Section~\ref{sec:tf_colocation}), \sysname further prefers to do {\it co-placement} of certain operators. This is aimed at minimizing  data transfer overheads. Common instances include: (i) groups of communicating operators whose computation times are much shorter than their communication times, and  (ii) matched   forward and backward (gradient-calculating)  operators.

Figure~\ref{fig:min_data_transfer} shows an example for case (i). This subgraph generated by \texttt{tf.tensordot} API is a frequent pattern occurring  inside  TensorFlow graphs.
The subgraph permutes the dimensions of $\texttt{op}_{\texttt{in}}$ output according to the \texttt{perm}'s output (\texttt{Transpose}) and then changes the tensor shape by \texttt{Shape}'s output (\texttt{Reshape}). 

When m-ETF {places} this subgraph on a cluster of 3 devices, it {places} $\texttt{op}_{\texttt{in}}$, \texttt{perm}, and \texttt{Shape} on different devices. 
Computation costs for \texttt{perm} and \texttt{Shape} are very short (because they process predefined values), whereas subsequent communication times are much larger. Thus, m-ETF's initial placement results in a high  execution time.

\sysname's co-placement heuristic works as follows. If the output of an operator is only used by its next operator, we place both operators on the same device. This is akin to similar heuristics used in {ColocRL}~\cite{colocRL}. In  Figure~\ref{fig:min_data_transfer},  

{\sysname's co-placement optimization places all of the operators on one device, avoiding any data transfers among the operators.}

For case (ii), to calculate gradients in the ML model, TensorFlow generates a backward operator for each forward operator. 
\sysname co-places each backward operator on the same device as its respectively-matched forward operator.

\begin{figure*}
    \centering

    \begin{subfigure}{.47\linewidth}
        \includegraphics[width=\linewidth]{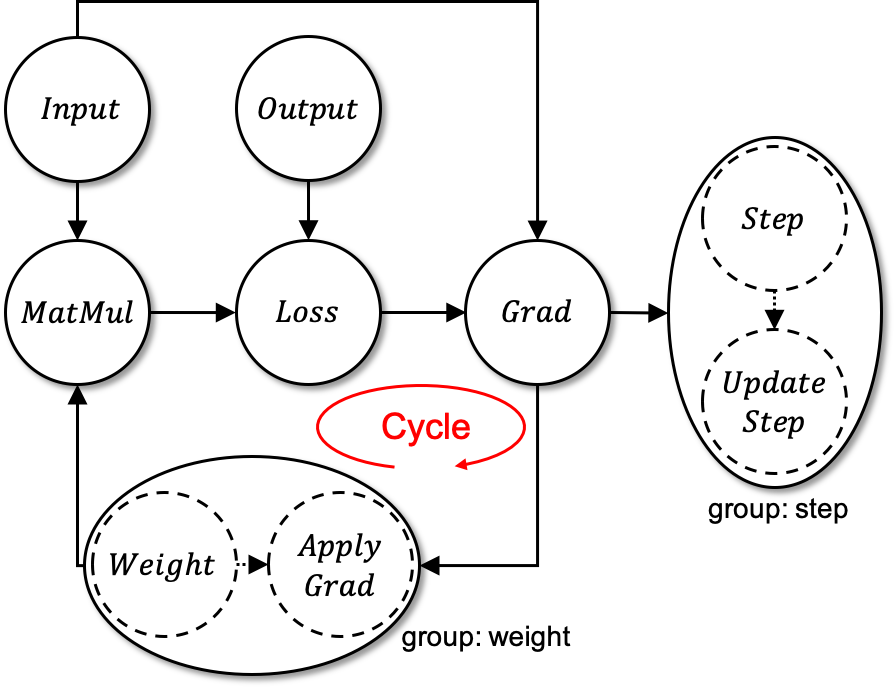}
    \caption{}
    \label{fig:fused_op_graph_cycle}
    \end{subfigure}
    \hfill
    \begin{subfigure}{.23\linewidth}
        \includegraphics[width=\linewidth]{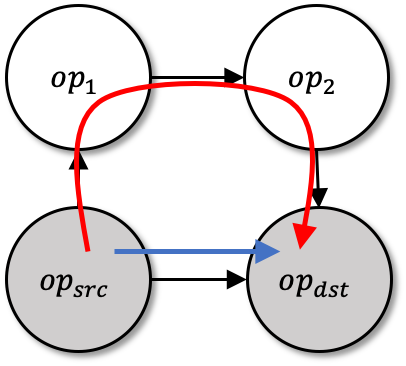}
        \caption{}
        \label{fig:cycle_cycle}
    \end{subfigure}
    \hfill
    \begin{subfigure}{.23\linewidth}
        \includegraphics[width=\linewidth]{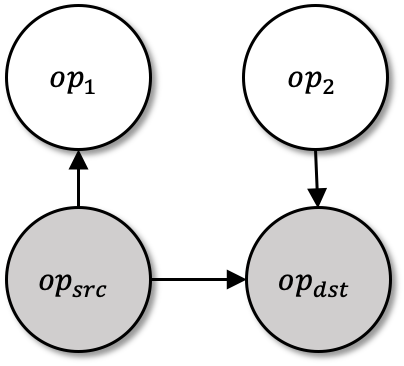}
        \caption{}
        \label{fig:cycle_no_cycle_1}
    \end{subfigure}
    \hfill
    \begin{subfigure}{.23\linewidth}
        \includegraphics[width=\linewidth]{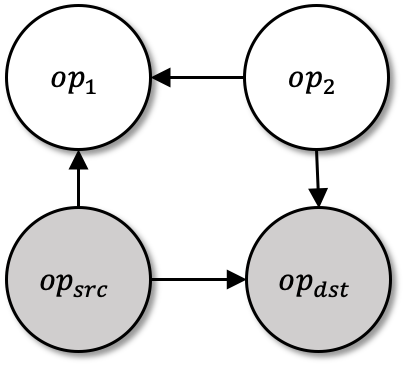}
        \caption{}
        \label{fig:cycle_no_cycle_2}
    \end{subfigure}
    \hfill
    \begin{subfigure}{.23\linewidth}
        \includegraphics[width=\linewidth]{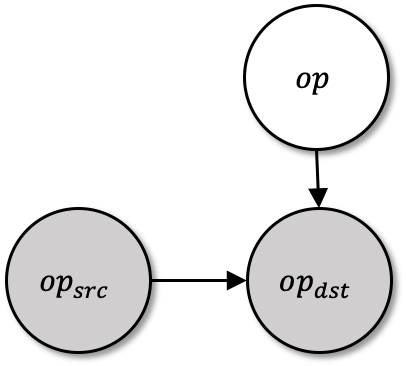}
        \caption{}
        \label{fig:cycle_one_outdegree}
    \end{subfigure}
    \hfill
    \begin{subfigure}{.23\linewidth}
        \includegraphics[width=\linewidth]{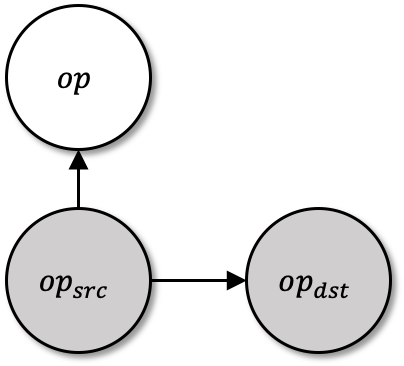}
        \caption{}
        \label{fig:cycle_one_indegree}
    \end{subfigure}
    
    \caption{
    \it {\bf Operator Fusion Without Creating Cycles.} 
    (a) shows a fused ML Graph Example. When $op_{src}$ and $op_{dst}$ are fused, some scenarios create a cycle (b), while others do not (c, d, e, f). \sysname fuses operators in a subset of ``safe'' cases, particularly (e, f).%}
    }
    
    \label{fig:operator_fusion_cycle}
\end{figure*}

Upon placing the first operator in a colocation group, \sysname uses  both the co-placement heuristic and the colocation constraints (Section~\ref{sec:tf_colocation}) to determine which other operators to also place on the same device. Co-placement not only minimizes communication overheads but also speeds up the placement time by 
reducing the overhead of calculating schedulable times on devices.

\subsubsection{Operator Count Minimization}
\label{sec:operator_fusion}
Placement time can be decreased by reducing the number of operators/groups to be placed. We do this via two additional methods: \\
i) {\it Operator Fusion}: Fusing operators that are directly connected and in the same co-placement group; and \\
ii) {\it Forward-Operator-Based Placement:} Placing operators by only considering the forward operators.

\medskip\noindent{\it \textbf{Operator Fusion. }}
\sysname fuses {\it operators}  using either the colocation constraints (Section~\ref{sec:tf_colocation}) or co-placement optimizations (Section~\ref{sec:co-placement}). \igr{This is new and different from TensorFlow's fusion of {\it operations}.} One challenge that appears here is that this may introduce {\it cycles} in the graph, violating the DAG required by our algorithms.  

Figure~\ref{fig:fused_op_graph_cycle} shows an example resulting from Figure~\ref{fig:example_graph}---a cycle is created when {\tt Step} and {\tt UpdateStep} are fused into a new meta-operator, and {\tt Weight} and {\tt ApplyGrad} are fused.

Consider two nodes--source and destination--with an edge from source to destination. 
Merging source and destination creates a cycle if and only if there is 
{\it at least one additional path from source to destination,} %}  
other than the direct edge. Note that there cannot be a reverse destination to source path as this means the original graph would have had a cycle.

In Figure~\ref{fig:cycle_cycle}, fusing $\texttt{op}_\texttt{src}$ and $\texttt{op}_\texttt{dst}$ creates a cycle.
Unfortunately, we found that pre-checking existence of 
{such additional paths} 
before fusing two operators is  unscalable, because the model graph is massive. 

\igr{

Instead, \sysname realizes that a {\it necessary} condition for an additional path to exist is that the source  has an out-degree at least 2 {\it and} the destination has an in-degree at least 2 (otherwise there wouldn't be additional paths). Thus \sysname uses a conservative approach wherein it fuses two operators only if the negation is true, i.e., {\it either} the source has an out-degree of at most 1, {\it or} the destination has an in-degree of at most 1 (Figures \ref{fig:cycle_one_outdegree}, \ref{fig:cycle_one_indegree}). This fusion rule misses a few fusions (Figures \ref{fig:cycle_no_cycle_1},  \ref{fig:cycle_no_cycle_2}) but it catches common patterns we observed, like  Figure~\ref{fig:cycle_one_outdegree}. 

}

%% Combined with image from Implementation.tex
\begin{figure}
        \centering
        \hspace*{\fill}
        \begin{subfigure}[b]{.28\linewidth}
            \includegraphics[width=\linewidth]{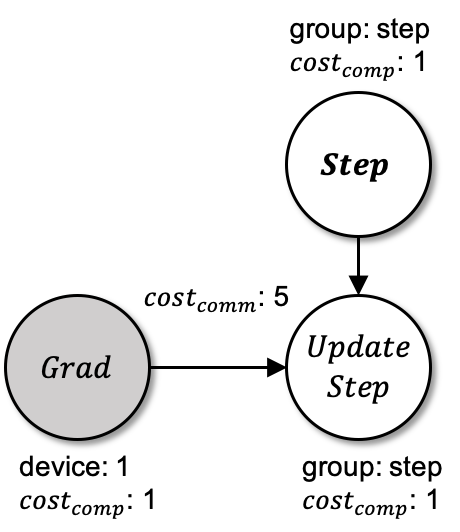}
            \caption{}
            \label{fig:fusion_before}
        \end{subfigure}
        \hfill
        \begin{subfigure}[b]{.3\linewidth}
            \includegraphics[width=\linewidth]{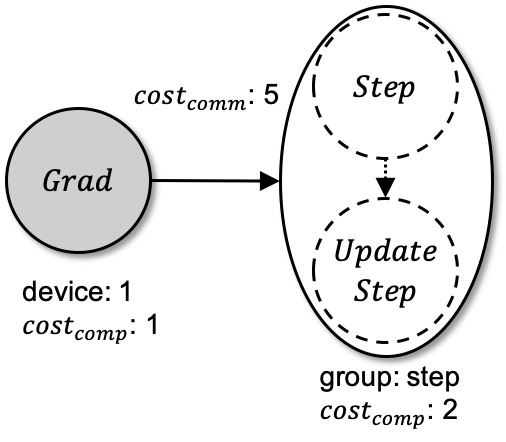}
            \caption{}
            \label{fig:fusion_after}
        \end{subfigure}
        \hspace*{\fill}
        \caption{
        \it {\bf Operator Fusion.} Avoiding Data Transfer Example. (a) Before Fusion. (b) After Fusion. 
        }
        \label{fig:fusion_avoid_data_transfer}
\end{figure}

\medskip\noindent{\it \textbf{Forward-Operator-Based Placement. }}
\label{sec:forward_operator_based_placement} When memory is sufficient (i.e., one device could run the entire model), \sysname  considers only forward operators for placement and thereafter co-places each corresponding backward (gradient) operators on the same respective device as {their forward counterparts.} This is a commonly-used technique~\cite{colocRL, placeto}.
This significantly cuts placement time. 
When device memory is insufficient,  
\sysname runs the placement algorithms using both forward and backward operators, forcing corresponding pairs to be co-placed using the heuristic of Section \ref{sec:co-placement}.

\medskip\noindent{\it \textbf{Example: Benefits of Fusion. }}
Figure~\ref{fig:fusion_before} shows the placement of a subgraph of Figure \ref{fig:example_graph} on two devices. 
\sysname first places \texttt{Grad} on \texttt{device-1}.  \sysname places the next operator, \texttt{Step} on the idle \texttt{device-2}, and colocates (due to TF constraints) \texttt{UpdateStep} on \texttt{device-2}. This creates communication between the devices. Assuming operators' compute costs are 1, and communication cost between \texttt{Grad} and \texttt{UpdateStep} is 5, this results in an execution time of 7 time units.

\bjc{On the other hand, Figure~\ref{fig:fusion_after} shows that \sysname merges \texttt{Step} and \texttt{UpdateStep} with operator fusion.}
Since this meta-operator's schedulable time on {\tt device-1} is earlier than on {\tt device-2} due to  communication overhead, \sysname places it on \texttt{device-1}. Fusion   lowers total execution time to 3 time units.

\medskip\noindent{\it \textbf{Loops in the Original Model Graph. }}
Different from the cycles discussed above, some network graphs consist of loops, e.g., RNNs. We use the unrolled ML graph 
\cite{loop_unrolling} to turn the graph into a DAG, and then apply \sysname's techniques.

\subsubsection{Sequential vs. Parallel Communication}
\label{sec:sequential_data_transfer}
Our algorithms from  Sections~\ref{sec:etf} and \ref{sec:sct_algo} assume that each operator can send data simultaneously to its children. 
\igr{\sysname also proposes a new way to deal with environments involving constrained networks  (including our deployment in Section~\ref{sec:eval}), where data transfer is sequential.
} 
For networks that limit each device to do at most one transfer at a time (out or in), 
\sysname assumes communication queues at  devices. 
Concretely, when a data transfer between two  devices is requested, \sysname assumes the request is put into the respective devices' communication queues and processed sequentially at both ends. 
During placement, \sysname calculates the wait time at the communication queues and adds it to the earliest schedulable time computed for the operator. %}
{Specifically the queue wait time is added to equation (\ref{eq:schedulable_time}) in Section~\ref{sec:etf}.} Otherwise, normal m-SCT/m-ETF apply, as described earlier.  

\subsection{\textbf{Design of \pyname 
}}

\label{sec:design_py}
\label{sec:pytdesign}

{We remind the reader that unlike \tensorflow's fine-grained operators and known communication graph, \pytorch: 
(i) has coarser modules, and (ii) requires the programmer to explicitly program  cross-device communication. 

Concretely---first, \pytorch models are built by composing different modules. The model is not natively available as a graph unlike  \tensorflow. To feed the model to \sysname's algorithms, Section~\ref{sec:py_basic} describes how we construct a graph by using \pytorch's Autograd \cite{pytorch} which tracks the flow of tensors among the modules of the model. Second, the primitive \code{.to()} API provided by \pytorch for developers to program communication is inefficient and manual. To address this, Section \ref{sec:py_comm} presents our communication protocol to handle  cross-device transfers efficiently and automatically.
}

\subsubsection{\sysname-\pytorch Graph}
\label{sec:py_basic}

Developers build \pytorch models by composing modules, i.e, classes inherited from \pytorch's \code{nn.Module}. Each module contains its tensor parameters (e.g., weights of a linear layer) and a \code{forward()} method, which defines how the module modifies the input. By default \sysname{} treats modules as 
the nodes of the graph in \pyname. Placing a module on a device means moving all its parameters to the device before beginning the training. 
During the training, our communication protocol (next Section~\ref{sec:py_comm}) ensures that the input to that module is also moved to the same device. Subsequently the \code{forward()} operation of the module will %can 
be invoked at runtime, and it will be executed on the assigned device. 
A model may also include operations not defined as \pytorch modules. For instance, arithmetic operations like scaling (e.g: \code{x=x/2}). But these operations usually do not have any associated parameters that must be assigned to a device by \pyname. Hence we exclude them from the graph during the placement planning. 
{By default, this associated operation will be executed on the device of the input tensor, i.e., \code{x}'s device in this case}\footnote{Arithmetic operations may still take some time to complete and memory to store their outputs, but we observed that their impact on the overall step time and memory budget is small, thus we  ignore them in generating placements. If desired, such operations may be defined as \code{nn.Modules} and be included in the placement graph.}. 

\pyname constructs the model's graph in two steps. First we obtain all modules that constitute the model and {\it co-place} modules occurring in common design patterns. Second we obtain the edges of the graph by tracking the flow of tensors using \pytorch's Autograd \cite{pytorch}. This approach is similar to PipeDream \cite{pipedream}. 

\noindent {\bf Co-placement:} Treating modules as nodes, we observed it is common for models to contain specific subgraphs (of modules) that occur as common design patterns throughout the model graph. \pyname{} groups such subgraphs into a single node---this is called co-placement (unlike \tfname's co-location constraints in Section \ref{sec:design}, this co-placement is a performance optimization in \pytorch). For instance in Inception models, the subgraph (Conv2d)$\rightarrow$ (Batch Norm) $\rightarrow$  (inplace ReLU) occurs commonly, and \pyname{} groups each occurrence as one node in the computation graph. 
{This co-placement allows \pyname{} to  avoid communication of tensors along the two edges of this subgraph. }
{An additional benefit of co-placement is that it significantly reduces the size of the computation graph, thus making \sysname's algorithms (Section~\ref{Methodology}) run faster. For instance, co-placement reduces the number of nodes in Inception-V3 \pytorch by 60\%, from 325 to 133.} 
{By default, \pyname uses the most atomic modules, i.e., not further divisible, e.g., 2D Convolution module (Conv2d). If the developer wishes, they can programmatically specify which modules  \pyname should be co-placed and treated as individual nodes. }

\noindent{\bf Building the graph:}
Next, Baechi obtains edges between the nodes. To do this, we run a training step of the model with dummy data. \igt{(We  run 20 such dummy training steps, also helping us profile  memory usage and  computation times of all the nodes.) } 
During the forward run, we annotate each intermediate output tensor with the 
node
that generated it. Meanwhile, \pytorch's Autograd automatically fills in the gradient function (\code{grad\_fn}) for each tensor created. Further, each such \code{grad\_fn} has a list of \code{grad\_fn} of tensors used as inputs in creating this tensor. Autograd stores this list to perform back-propagation later. We use this information that traces the tensor gradient functions, along with our tensor to node annotation, to construct a 
dependency graph among modules of the model.

It is possible that some gradient functions may include operations not related to any module, e.g., Autograd-specific operations such as \code{SelectBackward} or gradients of arithmetic operations.
But since these operations do not need a device placement, they are removed to obtain the dependency graph only between nodes of the model \igt{(they are added back in after the model is placed)}.

\subsubsection{Communication Protocol}
\label{sec:py_comm}

\pytorch provides native support for synchronous communication, which can be inefficient. If asynchronous data transfer is used, the developer is required to carefully insert synchronizations to ensure correctness. 

We design a general communication protocol for cross-device communication that is efficient and automated, thus relieving the developer from specifying manual configurations. We do so by leveraging {\it CUDA streams}, an abstraction that allows overlapping multiple sequences of operations in a GPU  \cite{cudastreams}.

To move a tensor \code{T} to $GPU_{0}$, \pytorch provides an API \code{T.to(0)}. To ensure correctness, \code{.to()} conservatively blocks both sending and receiving devices until all the operations submitted to both devices so far are completed. 
This can be avoided by leveraging the abstraction of {\it CUDA streams}~\cite{cudastreams}, in order to overlap communication with computation. 
{A CUDA stream in a GPU is a FIFO queue of operations that will be sequentially executed on the GPU. By default, all operations submitted to a GPU are placed on a single stream. To perform two operations in parallel, they must be placed on two separate streams on the GPU. Accordingly, on each GPU, \pyname defines one {\it compute stream} and multiple {\it communication streams}.  The compute stream queues the computations corresponding to the modules placed on the GPU, while the communication streams concurrently move the relevant tensors across the devices.}

{However, CUDA streams need to be programmed carefully to specify synchronization points that  obey dependencies in the model graph.} We use CUDA Events provided by the runtime API \cite{cudaevent1} to synchronize the independent streams.

{None of the existing ways of using CUDA streams in \pytorch fits our needs. } 
Concretely---first, in \pytorch-Distributed \cite{pydist} and PipeDream~\cite{pipedream}, training steps proceed in stages, each working on a different batch of data. All tensors generated in a given stage are transferred to the next device at the end of the stage. This pattern, where all communication happens synchronously only after all computations are complete requires few, if any, synchronizations. In contrast, Nimble \cite{nimble}, deals with multiple parallel streams working on the same batch of data, like in Model Parallelism. Compute and communication events may asynchronously occur at any time and it requires synchronizations to preserve data dependencies. However, Nimble is a single-GPU system. 

{In \pyname, we use a {\it greedy-wait} strategy.} 
Concretely, first we {\it greedily} push out the output of a node, as soon as it is computed, to the devices of its children nodes. Second, before starting the compute operation, a child node must {\it wait for all} its incoming input stream or if the input has already been transferred it must pull the copy of the inputs on its devices. We implement our greedy-wait communication protocol as a wrapper around the \code{forward()} function of each node.

{\SetAlgoNoLine%
\begin{algorithm}
    
    \SetKwInOut{input}{Input}

    \For{each parent of the \textbf{node} in graph}{
        (node's compute\_stream) {\bf wait for} (rx\_stream from parent's device)\;
    }
    \BlankLine
    {\bf On } node's compute\_stream:\\
    \qquad input\_local = local copies of {\bf input} on node's device\;
    \qquad output = \code{forward\_operation}(input\_local)\;
    \BlankLine
    \For{each child of the node in graph}{
        (tx\_stream to child's device) {\bf wait for} (node\_compute\_stream)\;
     }
    \BlankLine
    \For{each child of the node in graph}{
        {\bf Using } (node's tx-rx\_stream pair to child's device ):\\
                \qquad send output to child\;
        
    }
    {\bf return} output
 \caption{Communication protocol built around each node's \code{forward(\bf{input})  } 
 }
 \label{alg1}
\end{algorithm}}%

Algorithm \ref{alg1} shows the complete communication protocol, and we describe it in detail.  Before we start the training, for a given node and its child node on a different device, we create two streams - a {\it tx-stream} on the node's device and a {\it rx-stream} 
on the child node's device. Only one such stream pair is sufficient for a child device, even if multiple children nodes are on that device. So if a node on $GPU_{0}$ has two children, one on $GPU_{1}$ and another on $GPU_{2}$, then for that node we create one {tx-rx} stream pair each to $GPU_{1}$ and $GPU_{2}$. We do this for every node and for every device any of its child nodes reside in.

Each device's compute stream 
queues the computations of nodes assigned to that device. When a node reaches the head of the compute stream of its device,  the compute stream is made to wait for all the {rx}-streams to that device from all the parent nodes (line 2). 
{The {rx}-streams carry a copy of the output tensors of the parent nodes to the device of the current node. Once all {rx}-streams have completed the transfer, } 
{these tensors are passed as inputs to the actual forward computation of the node (lines 4-6)}.
All the {tx}-streams egressing from this node are made to wait for the compute stream to finish the computation (lines 7-9). 
{The {tx}-stream carries the outputs of the node to the child node's device and the corresponding {rx}-stream on the child node's device receives it.} As soon as the output is ready, it is asynchronously sent to all the child devices through the {tx}-streams and received asynchronously at the child devices through their respective {rx}-streams (lines 10-12). The compute stream can move to the next node assigned to it while the {tx}-streams are transferring out the output. Similarly, the compute streams on child devices are not interrupted by incoming tensors on their rx-streams. Note that the output is sent to a child device only once even if multiple children reside on that child device. The {tx}-{rx} streams serve as synchronization points in addition to overlapping communication and computation. 

{Our m-ETF and m-SCT algorithms from Sections~\ref{sec:etf} and \ref{sec:sct_algo} assume that each node can send data simultaneously to its children. Our protocol mimics this communication with multiple outgoing {tx}-streams transferring the tensors to child devices in parallel. While such overlapping {tx}-streams from a device may slightly increase the communication times in each of the streams, we observed that the associated effect on step time is small.}
Also as long as we facilitate and synchronize the forward run correctly, \pytorch's Autograd \cite{pytorch} ensures that the back-propagation correctly executes in the reverse order of the forward sequence. It automatically manages input-output dependencies and synchronizations in the reversed order.

\section{Implementation}

\begin{figure}
        \centering

        \includegraphics[width=0.8\linewidth]{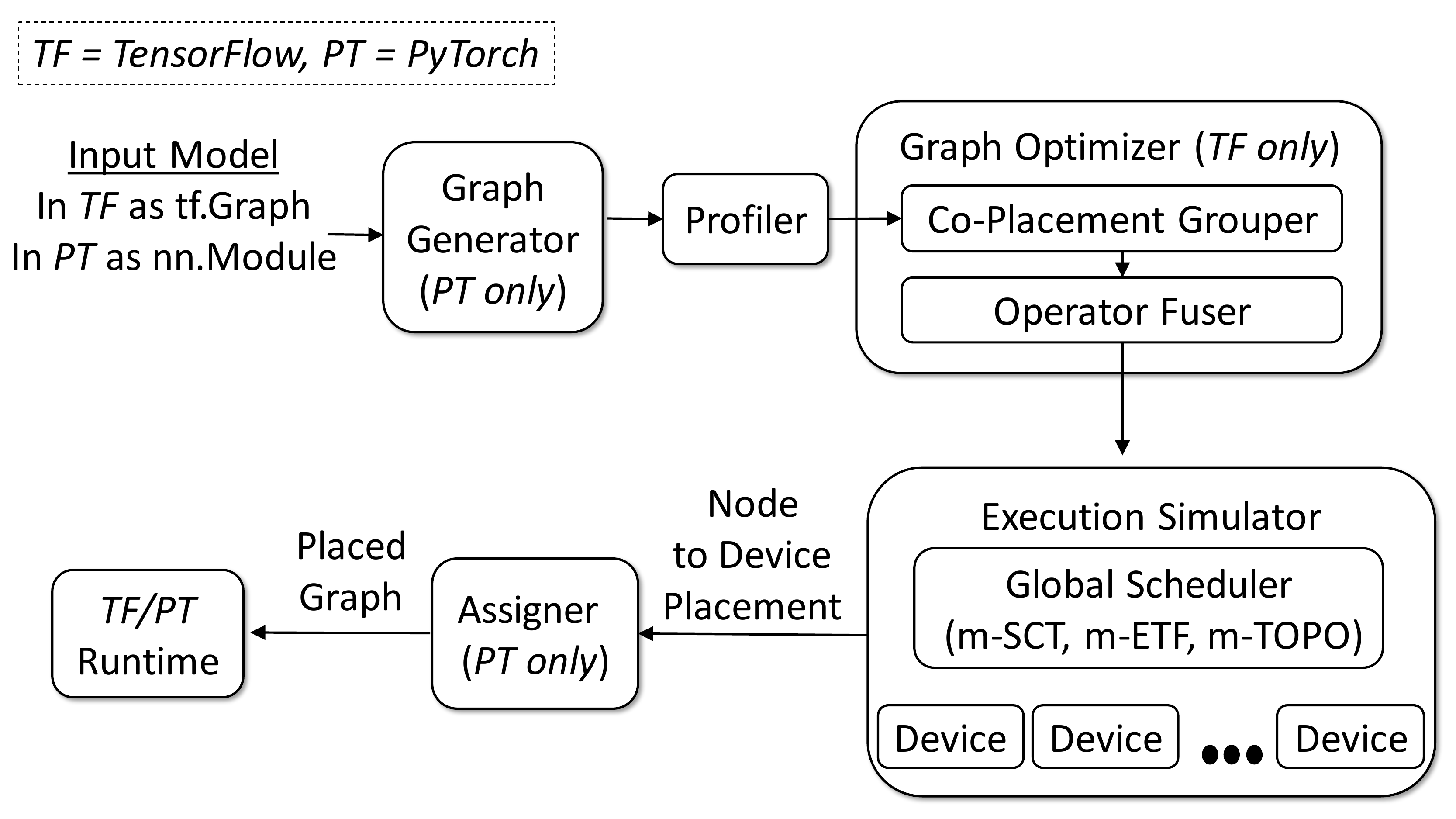}
        \caption{{\it {\bf \sysname System Architecture}} 
        }
        \label{fig:workflow}
\end{figure}

In order to integrate modularly with TensorFlow (TF)~\cite{tensorflow} v1.12 and \pytorch v1.9 \cite{pytorch}, \sysname adopts the architecture shown in Figure~\ref{fig:workflow}. \sysname executes the following steps: 1) its {\it Graph Generator and Profiler} constructs the graph annotated with each node's time and memory requirements,  
2) \sysname's {\it Graph Optimizer for \tensorflow} uses the design of Section~\ref{sec:design} to account for TF colocation constraints, and applies co-placement {and operator fusion},  
3) \sysname's {\it Execution  Simulator (ES)} (Section~\ref{esim}) executes our algorithms (m-TOPO, m-ETF, or m-SCT) and generates the placement (in \tensorflow a ready-to-use placed graph is output),  and 
4) the {\it Assigner for \pytorch} (Section~\ref{assigner}) modifies the model to allow execution according to the generated placement. The placed graph can then be used in the training script as the drop-in replacement for the single-GPU model, in both \tensorflow and \pytorch.  
{Next we describe each of the components in detail. }

%\subsubsection{Profiling}
\subsection{Graph Generator, Profiler and Optimizer}
\label{sec:py_profiler}

In \tfname, the model is already given as a static graph. The Profiler then measures and annotates each node with its time and memory requirements. \sysname parses this annotated graph and generates an equivalent intermediate NetworkX \cite{networkx} graph. The NetworkX format allows \sysname to both store operator execution metadata (computation and communication times, memory needed, etc.), and to easily manipulate the graph (e.g., fuse operators). Then, in case of \tfname, we additionally apply the co-placement and operator fusion optimizations (Sections \ref{sec:co-placement}, \ref{sec:operator_fusion}) to the graph.

In \pyname, the Graph Generator constructs the graph corresponding to the input model as per  Section~\ref{sec:py_basic}.

For  communication time, we use a linear model proportional to data size. Concretely, we implemented a microbenchmark tool to measure communication times for various data sizes, and generated a communication cost function through  linear regression.

\subsubsection{Profiler:}

The Profiler measures 
{computation times} 
and memory requirements of each node in the graph. 
In \tfname, we use the standard \tensorflow profiling tool to obtain computation time and memory allocation for each operator. \tensorflow profiler returns  allocation information for temporary, permanent, and output tensor memory. The temporary memory is allocated at the beginning of an operation and deallocated when the operation finishes. The permanent memory is allocated and used over the entire execution, e.g., to store {persistent states such as} weights. 
%
%\igt
{For \pyname we build a simple profiler, akin to that in \cite{pipedream}, for measuring time using hooks \cite{pyhook}, and for measuring memory~\footnote{While \pytorch has an internal profiler using it would have  required us to  handle dependencies and operation granularity carefully. Our simple approach avoids these. }.
}

\medskip\noindent{\it \textbf{\pyname Profiler Memory Estimation }}

In \pytorch, GPU Memory required to hold all the tensors is reserved once during the first training step and reused in subsequent steps. To model the memory usage pattern, 
 we categorize each node’s memory as consisting of 5 components:

{
\begin{enumerate}[(a)]
\item Parameters memory: Memory occupied by parameters of the node;
\item Output Memory: forward output of the node;
\item Parameter gradient: gradients of parameters of the node;  % (= Parameters memory)
\item Upstream gradient: gradient of output of the node; %  (= Output Memory)
\item Memory temporarily used in computing the
output/gradient.
\end{enumerate}}

\begin{table}
\begin{center}
\begin{tabular}{ | m{5em} | m{2cm}| m{2cm} | } 
  \hline
  Memory & Inference & Training \\ 
  \hline
  Permanent  & (a) & (a) + (b) + (c) \\ 
  \hline
  Temporary  & (b) + (e) & (e) + (d) \\ 
  \hline
\end{tabular}
\end{center}
\caption{Memory Consumption in \pytorch}
\label{tab:pyprofiler}
\end{table}

{
 Table~\ref{tab:pyprofiler} summarizes how these five metrics are used in training and inference, and whether they are used as temporary memory or permanent memory. We describe each term. 
}
{During the training phase, each node requires memory to store: (a) its parameters, (b) the node's forward output tensors, and (c) parameters' gradient information. In \pytorch, memory to store parameter gradient information is acquired once in the beginning of training and permanently held until the end of training. Similarly, output tensors ((b)) are treated as permanent memory since during each forward run, outputs of all nodes must be stored. They are required later during back-propagation. In contrast during the inference phase (forward only runs), output of a node is temporary since it is immediately released after being consumed by the subsequent node. During back-propagation in training, memory is also required to hold the gradient of output of the node ((d)). 
This is a temporary requirement since this memory is released after the output gradient has been used to compute the node's parameter gradients and its input gradient.  Nodes may also  require temporary memory 
while performing these computations ((e)). 
For example, 
in computing the parameter gradients, a temporary matrix is used to store all the gradients and then the node's parameter gradients are updated at once.} %\newAdd
{For in-place nodes like the in-place ReLU, (b) is set to 0 since no new output tensor is generated ((a) and (c) are also 0 incase of ReLU since it has no associated parameters).}

\subsection{Execution Simulator}
\label{esim}
\sysname's Execution Simulator (ES) executes the algorithms on the profiled graph and generates the placement. The ES takes as input the NetworkX operator graph, the number of GPUs and the memory capacity of each of them. The output is the graph in which all operators are assigned to devices. We describe the common parts of the ES across both \tfname and \pyname, explicitly pointing out differences as  necessary.

Our initial attempt was in fact to try re-purposing  \tensorflow’s (existing) simulator, but using that required us to assume operators were already placed, zero communication cost, and no caching---these were inapplicable to \sysname. Motivated by this, we designed \sysname's new ES uniquely for  memory-constrained placement.

The ES consists of: a) a global scheduler, and  b) simulated devices (with specs identical to deployment). 
The global scheduler maintains a single queue with operators that are ready to run. The scheduler extracts operators from its queue and applies our scheduling algorithms (m-TOPO, m-ETF, m-SCT) to place them on devices. 

In ES, each device has two FIFO queues, one for operators and one for data transfer. This allows data transfer to overlap with operator execution. When a device receives a tensor from another device, it caches the tensor to avoid duplicate data transfer.

\medskip\noindent{\it \textbf{Dynamic Memory Allocation. }}
 
Calculating a device's memory usage as the sum total of all its assigned operators (assigned over the entire duration) clearly overestimates memory. %}
For example in \tensorflow,  Inception-V3 with  batch size 32 can execute using 4 GB even though its operators' memory needs add up to 22 GB. %

In generating the placements, the ES calculates memory in a way that parallels how the frameworks manage memory. 
Concretely \sysname's ES tracks an estimate of memory usage during its placement. When an operator executes on a device, the device allocates temporary memory, and separate memory for its output tensors. The temporary memory is deallocated when the operator finishes. There are minor differences in the ES for \tensorflow and \pytorch. %\bjcnew
{\tensorflow uses separate operators for forward and backward computation. The output memory of an operator} is deallocated after all its successors finish. 
In \pytorch, %\bjcnew
{the forward and backward computation runs in a single module. The output memory of a module is held until its backward computation is completed.} The output memory is treated as a part of the permanent memory as explained in the previous Section \ref{sec:py_profiler}. If a device's memory becomes full, \bjc{the device} can be removed--this never happens in practice as usually a device has at least a few bytes left. 

Note that in \tfname, memory is reserved for a colocation group at device $p$ when the first operator is placed on $p$ (Section~\ref{sec:tf_colocation}). The reserved memory is deallocated when all the operators in the group finish.
%%%%%%%%%%%%%%%%%%%%%%%%%%%%%%%%%%%%%%%%%%%%%%%%%%%%%%%%%%%%%

\smallskip\noindent{\it \textbf{Linear Programming Solver. }}
To solve the SCT LP problem, we use the interior point method~\cite{convex_optimization}.
%\igc{we should cite a textbook or wikipedia page or code}
This  is preferable over other solvers such as simplex \cite{bartels1969simplex}  as it guarantees polynomial execution time \cite{tomlin1989note}. Concretely, we use the primal dual interior-point solver via Mosek optimization  \cite{mosek}, which has a run time complexity of $O(n^{3.5}L)$, where $L$ is the maximum number of bits in the LP input, and $n$ is the number of variables.

\subsection{Assigner \textit{in \pyname}}
\label{assigner}
{Once the mapping of graph nodes to devices is decided, the Assigner contains the mechanism to initialize the assignment. 
The Assigner for \tensorflow merely  changes the device attribute of the nodes. The Assigner for \pytorch requires multiple steps---it needs to: i) move nodes' parameters to the assigned devices, ii) send output tensors from a node to devices of child nodes at the device boundaries, and iii) avoid {naive} use of \code{.to()} for cross-device communication as this leads to  inflated step times owing to unnecessary blocking of the devices (see Section~\ref{sec:py_comm}).} 

\pyname's Assigner enables this process
by automatically adding wrappers around the \code{forward()} methods of the nodes. The wrapper transparently handles the communication and caching of the tensors using the communication protocol in
Section~\ref{sec:py_comm}. This way, the developer does not need to make any changes to the input model code written for a single-GPU execution. The output of the Assigner is a 
model assignment that can be used as a drop-in in any  existing training script.

\subsection{Miscellaneous Issues}
\label{secmisc}

We discuss a few key miscellaneous aspects.

\medskip\noindent{\it \textbf{LP Modifications. }}
The ILP solutions (Section~\ref{sec:sct_algo}) resulted in more than one  favorite child (or parent) being selected for certain nodes. {In \sysname we lowered the rounding threshold from 0.5 to below 0.2. This eliminated all violations, and avoided nodes from having multiple favorite children. (We use threshold = 0.1 in practice.)}

\medskip\noindent{\it \textbf{Ignoring Bootstrap Steps in  Profiling. }}
1) In a training run of a model graph, step times are initially high due to \tensorflow bootstrapping. We estimate step times in steady state, after a few iterations have passed. 2) Some \tensorflow operators are implemented with multiple GPU kernels. When profiling these operators, we include multiple kernel executions, in order to avoid underestimation. This is similar to \tensorflow's cost model~\cite{tensorflow}.

\medskip\noindent{\it \textbf{Reordering Layers in \pytorch}.} Dynamic line-by-line execution in \pytorch means that the modules' \code{forward()} functions will be called in the order in which they appear in the code rather than in the topological order followed by \sysname' ES (Sec. \ref{esim}). 
We reorder the actual execution of modules' computations on GPUs by launching a thread when a module's \code{forward()} is called. The thread waits until its topological parent (according to the ES) has submitted the computation task to the GPU and only then submits its task. 
However, 
such reordering does not give a noticeable advantage over just executing the code order\old{given the largely linear structure of the models in \pytorch}since the two orders do not differ significantly in most cases.

\section{Evaluation}
\label{sec:eval}

% questions to answer in the evaluation.
Our evaluation answers the following six  questions:\\
1. How fast is \sysname's {\it placement time}, i.e., how quickly do our algorithms find placements? (Section~\ref{sec:placement})\\
2. How fast are the {\it step times} of the placement generated by \sysname, i.e., training time per step of the placed model? (Section~\ref{eval:sufficient_memory})\\
3. How do the step times for \sysname compare to {\it single GPU} and {\it expert placements}? (Section~\ref{eval:sufficient_memory})\\
4. How do the \sysname's step times change when there is {\it insufficient memory} per GPU?  (Section~\ref{eval:insufficient_memory})\\ 
5. How much is the benefit due to \sysname's {\it optimizations} from Sections~\ref{sec:co-placement}, \ref{sec:operator_fusion} and \ref{sec:py_comm}? (Section~\ref{sec:eval:opt})\\
%\igt
{6. Are algorithmic approaches preferable over RL approaches for model parallelism? (all subsections).}

\subsection{Experimental Settings}
\label{sec:experiment_setup}

{
We use two popular ML benchmarks for each framework: A) for \tensorflow, we use Inception-V3 and Google Neural Machine Translation System (GNMT), and B) for \pytorch, we use Inception-V3 and a Transformer model. The former choice is because: i) Inception-V3 and GNMT are respectively considered the best representatives of vision and Natural Language Processing (NLP) models, and ii) past work~\cite{placeto,colocRL, hierarchicalRL} used Inception-V3 and NMT (GNMT is a more complex version), thus allowing us to compare. 
For \pytorch we replace GNMT with Transformer as the former is implemented using the LSTM module \cite{pylstm1}, making the (latter) Transformer  a more complex and generalized version. We describe the three benchmark configurations in detail below: 
}

\medskip\noindent{\it \textbf{Inception-V3  Benchmark Configuration. }}
Inception-V3~\cite{inceptionv3} is a convolutional neural network architecture that is widely used for image classification. 
This model is composed of multiple blocks called Inception modules. The Inception modules consist of branches of convolutional and pooling operators. 
To train the model, we use RMSProp~\cite{rmsprop} and batch sizes of both 32 and 64.

\medskip\noindent{\it \textbf{GNMT Benchmark Configuration. }}
Google Neural Machine Translation System (GNMT) \cite{gnmt} is a language model for automated translation. GNMT consists of: encoder and decoder modules, each  a stack of recurrent neural networks (RNNs); and the attention module to process long sequences effectively. We use 4 long short-term memory (LSTM) layers of the encoder and the decoder layers with residual connections, and the Bahdanau attention mechanism \cite{BahdanauAttention}. We use the LSTM hidden size of 512, the vocabulary size of 30,000, the unrolled RNNs
with the sequence length of 40 and 50,
and the batch size of 128 and 256. 
{\tfname applies the co-placement optimization to LSTM cell operators and also to attention operators.}

Compared to Inception-V3, GNMT has fewer barriers (sync points) inside its model graph, indicating that GNMT has a higher potential to benefit from \tfname's parallel placements.

\medskip\noindent{\it \textbf{Transformer  Benchmark Configuration. }} Transformers are a versatile family of models used in vision as well as language. 
Like Neural Machine Translation (NMT), Transformers have  an encoder-attention-decoder architecture. But while NMT processes one word at a time, Transformers use 
multi-head attention modules that process the entire sequence at once. In \pytorch, we implement an attention operation in the traditional way \cite{transformerCode}---as one large matrix multiplication and hence as a single module. 
%\canAdd{For very large models (like the GPT3 \cite{gpt3}), it may be necessary to break the attention module into multiple modules. Referred to as  'intra-operator parallelism', this is explored in works like Alpa \cite{alpa} which specifically focus on Transformer models.} 
For concreteness, we use the base Transformer model from \cite{transformer} (without weight sharing) with a vocabulary size of 30,000, sequence length of 50 and batch sizes of 64, and 
128.
%\canAdd{Note that changing the sequence length here does not change the computation graph like in GNMT (the modules just become more compute intensive). Thus we evaluate on only one sequence length}

\medskip\noindent{\it \textbf{Machine Setup. }}
All experiments are run on our local server that has 4 NVIDIA GTX 2080 GPUs, with 8 GB per-GPU memory {(the machine also has an Intel i9-7960X CPU, but this is not used to execute operators).}
GPUs are connected to CPUs via PCIe 3.0 x16  (we do not use NVLink \cite{nvlink}).
All data transfers go through the host memory (no P2P communication among GPUs).  This results in a  slow IO bus, and we believe this high ratio of communication overhead to computation overhead is representative of realistic scenarios like the kinds outlined in Section~\ref{introduction}. We place all GPU-supported operators only on GPUs.

\medskip\noindent{\it \textbf{Approach to Comparison. }} To quantify the benefits of using an algorithmic approach to model parallelism over a Reinforcement Learning (RL)  approach for model parallelism, 
we compare \sysname to the best RL-based model parallelism techniques: \cite{placeto,colocRL, hierarchicalRL}. 
 Directly running these other systems was complicated by lack of uniform availability of working code---Placeto's code~\cite{placeto} missed key optimizations; ColocRL~\cite{colocRL} is proprietary; only HierarchicalRL's code~\cite{hierarchicalRL} was available, but it was slow and generated inefficient placements. E.g., For GNMT, HierarchicalRL took 12 hours+ to run placement (batch size 128, length 50) and the resultant step time was much higher than expert's, contrary to HierarchicalRL paper's claims. Essentially, direct comparison would be unfair to these other papers without knowing the exact hyperparameters they  used to achieve their ``best'' performance. In light of this, our comparison gives the benefit of doubt to, and uses the best performance from, these learning-based placement papers. All the above papers compared step times to experts, and we do too. {We do not compare to other algorithmic techniques for model parallelism (listed in Section \ref{introduction}) because of either: their  standalone nature  \cite{amarpaper} (making a \tensorflow/\pytorch comparison unfair), or their limitation to  Transformers \cite{terapipe,  megatron}, or because they are already shown to be comparable to our performance \cite{pesto}.  This also keeps our evaluation focused on comparing algorithmic approaches to RL approaches for model parallelism. 
 
  }

\begin{table}[tb]

  \centering

  \begin{tabular}{cccc}
    \toprule
    Model  &  HierarchicalRL \cite{hierarchicalRL} &  Placeto \cite{placeto} & \sysname (m-SCT) \\
    % Model  &  ColocRL \cite{colocRL} &  Placeto \cite{placeto} & \sysname (m-SCT) \\
    \midrule
    Inception-V3 & 11 hrs 50 mins & 1 hr 49 mins & \textbf{1-10 seconds} \\
    NMT (GNMT)  & 1 day 21 hrs 14 mins & 2 days 20 hrs 40 mins & \textbf{1.2-48 seconds} \\
    Transformer  & N/A & N/A & \textbf{1-3 seconds} \\
    \bottomrule
  \end{tabular}

  \vspace{0.1in}
  \caption{\it {\bf Placement Time.} Time to Generate a Placement for our target machine with 4 GPUs. 
}
  \label{tab:Performances_of_baseline}
\end{table}

\subsection{Placement Time}
\label{sec:placement}

Table~\ref{tab:Performances_of_baseline} shows both: 1) measured placement times of \sysname, and 2) calculated placement times for two learning-based techniques, namely: 
HierarchicalRL~\cite{hierarchicalRL} and Placeto \cite{placeto}. The numbers for HierarchicalRL and Placeto are normalized quantities, both derived from numbers reported in  {Addanki~et~al.}~\cite{placeto}. 
For these two systems, we multiply the fastest  step time among its reported placements,  by the number of  placement samples\footnote{Even  if one were to parallelize the learning-based placers, their resource usage would be similar to the normalized time metric we show.  }. {For instance, %ColocRL's \cite{colocRL}
HierarchicalRL's~\cite{hierarchicalRL} Inception-V3 placement training time is derived as a product of the reported final step time (1.19 s) and the number of samples (35,800), giving 42,602 s, or 11 hrs 50 mins.
}

Hence, the numbers for these learning-based placers are their best-case performance. In comparison, we use the worst-case placement times from \sysname, specifically from m-SCT which took the longest to generate a placement. 
{Note that all times  in Table~\ref{tab:Performances_of_baseline} exclude time to profile the graph, 
as profiling is a common baseline encountered by all the three approaches shown. } 
{We find the profiling time to be low: about 10--12 s total for Inception-V3 and GNMT {in \tfname, and about 11--14 s for Inception-V3 and Transformer in \pyname.} For instance, in \tfname, this breaks down as 2-4 s for warmup execution, 1--3 s for graph execution for profiles, and less an 1 s for parsing  profile results.  }

Table~\ref{tab:Performances_of_baseline}  shows that \sysname places ML models
orders of magnitude faster than the learning-based approaches. For Inception-V3, \sysname reduces placement time, from 1.8--11.8 hours (using existing techniques~\cite{hierarchicalRL,placeto}),  
to under 10 s in both \tfname and \pyname. 
{Thus \sysname is $654\times$--$42.6$K$\times$ faster at placing Inception-V3. }
For GNMT, \tfname reduces placement time from several days to under 48 s. 
{Thus \sysname is $3392\times$--$206$K$\times$ faster at placing  GNMT.} 
{For Transformer, \pyname places it under 3 s. Because HierarchicalRL and Placeto \cite{placeto, hierarchicalRL} did not include Transformers in their evaluation, the corresponding entries are marked as not available in   Table~\ref{tab:Performances_of_baseline}. }

Overall, \sysname is 
{$654\times$--$206$K$ \times$} faster at placement compared to today's learning-based approaches \cite{placeto,hierarchicalRL}.

\subsection{Placement with Sufficient Memory}
\label{eval:sufficient_memory}

We next evaluate the effectiveness of the generated placement by measuring the step time %(training time)
of the placed model, i.e., its time to execute 1 training step on an input data batch. \igt{Step time is a key metric as completion time on a training set is directly proportional to step time.} We first explore the scenario when each GPU has sufficient memory to run the entire model. We compare against both: i) step time on a single GPU, which might be fast because it avoids the overheads of communication, and ii) an {\it expert}-based placement scheme for placement on multiple GPUs. 

The expert is a manual process and we do it as follows. For GNMT in \tensorflow, we use the technique of Wu~et~al. \cite{gnmt}. Each LSTM layer in the encoder and decoder modules are placed on different GPUs. The embedding layer is placed on the same GPU as the first LSTM layer. The output projection layer is placed on the same GPU as the last decoder LSTM layer. For Inception-V3 {in both \tensorflow and \pytorch}, the expert is the  single GPU placement, 
similar to {HierarchicalRL}~\cite{hierarchicalRL}.

For the Transformer model in \pytorch we use the common practice of putting the encoder on one device and the decoder on another device %as the baseline expert placement 
\cite{transformersplit1}.

\begin{table*}[t]
    \vspace{-0.1in}
    \centering

    \begin{tabular}{cccccccccccc}
        \toprule
        & & & & & & & & \multicolumn{4}{c}{Speedup over} \\
        \cmidrule{8-11}

        & & & & & & & &  \multicolumn{2}{c}{Single GPU} & \multicolumn{2}{c}{Expert (4 GPUs)} \\
        
        \cmidrule(l){8-9} \cmidrule(l){10-11}
        
        & Model & \multirow{-2}{*}{\shortstack{Batch \\ Size}} & \multirow{-2}{*}{\shortstack{Single \\ GPU}} & Expert & m-TOPO & m-ETF  & m-SCT & m-ETF & m-SCT & m-ETF & m-SCT \\
        
        \midrule
        
        & \multirow{2}{*} {Inception-V3}
            & 32 & 0.269 & 0.269 & 0.286 & 0.269 & 0.269 & \multicolumn{4}{c}{0.00\% (1 GPU Expert)} \\
                %\multicolumn{2}{c}{0.00\% (1 GPU Expert)} \\
            & & 64 & 0.491 & 0.491 & 0.521 & 0.491 & 0.491 & \multicolumn{4}{c}{0.00\% (1 GPU Expert)} \\

        &  \multirow{2}{*} {\shortstack{GNMT \\ (length: 40)}}
            & 128 & 0.251 \ignore{& 4} & 0.214 & 0.265 & 0.224 & 0.212 & 12.1\% & 18.4\% & -4.5\% & 0.9\% \\
            & & 256 & 0.474 \ignore{& 4} & 0.376 & 0.481 & 0.354 & 0.369 & 33.9\% & 28.5\% & 6.2\% & 1.9\% \\
            
         &  \multirow{2}{*} {\shortstack{GNMT \\ (length: 50)}}
            & 128 & 0.319 \ignore{& 4} & 0.259 & 0.348 & 0.264 & 0.267 & 20.9\% & 19.5\% & -1.9\% & -3.0\% \\
        {\rotatebox[origin=c]{90}{\rlap{\xspace \tensorflow}}} & & 256 & 0.618 \ignore{& 4} & 0.484 & 0.609 & 0.502 & 0.516 & 23.1\% & 19.8\% & -3.6\% & -6.2\% \\
            
        \midrule
        %-------
        &  \multirow{2}{*} {Inception-V3}
            & 32 & 0.240 & 0.240 & 0.274 & 0.241 & 0.241 & \multicolumn{4}{c}{0.00\% (1 GPU Expert)} \\
                %\multicolumn{2}{c}{0.00\% (1 GPU Expert)} \\
            & & 64 & 0.461 & 0.461 & 0.537 & 0.465 & 0.462 & \multicolumn{4}{c}{0.00\% (1 GPU Expert)} \\
            
        &  \multirow{2}{*} {\shortstack{Transformer \\ (length: 50)}}
            & 64 & 0.249 & 0.257 & 0.262 & 0.242 & 0.244 & 2.9\% & 2.0\% & 6.2\% & 5.3\% \\
        {\rotatebox[origin=c]{90}{\rlap{\pytorch}}} & & 128 & 0.465 & 0.462 & 0.466 & 0.451 & 0.453 & 3.0\% & 2.6\% & 2.4\% & 2.0\% \\
            
        \bottomrule
    \end{tabular}
    %\end{adjustbox}
    \vspace{0.1in}
    \caption{
        \it {\bf \sysname with Sufficient Memory.} Average Step Times (Training) in seconds of Placed Model Graphs, and Speedup over Single GPU and Expert Placements. 4 GPUs (unless otherwise mentioned). }
    \label{tab:step_time}
\end{table*}

\medskip\noindent{\it \textbf{m-ETF, m-SCT --VS.-- Single GPU, Expert. }}
Table~\ref{tab:step_time} shows the step times for the three algorithms in \sysname--namely m-TOPO, m-ETF, and m-SCT---as well as the single GPU and expert. We show numbers for 2 batch sizes in each model, and 2 sequence lengths in GNMT. We observe that for Inception-V3: 1) in \tensorflow, m-ETF and m-SCT find the same device placements as the expert, i.e., place all operators in a single GPU. 2) In \pytorch  m-ETF and m-SCT placements use three and two GPUs respectively, but have the same step time as 1-GPU expert. 3) %\newAdd 
{Compared to the expert, m-TOPO step time's is higher by {6.1--6.3\%} %5.8--5.9\% 
in \tfname and by {14--17\%} in \pyname for Inception-V3. 
}
 This occurs  because m-TOPO  splits the neural network  between the Inception blocks, and hence the next inception block(s) are  unable to run until the previous block(s) finish.

In \tensorflow GNMT, first, compared to  single GPU placement, m-ETF's   placements have step times that are 12.1--33.9\% faster.  The step time speedups for m-SCT over single GPU are between 18.4--28.5\%. These observations show that \sysname's m-ETF and m-SCT are able to %effectively distribute  operators over multiple devices  and are able to 
extract benefits of parallelism in spite of communication overheads. Second, in GNMT, compared to the expert, m-ETF is between 4.5\% slower and 6.2\% faster in step times.
Compared to the expert, m-SCT is between 6.2\%  slower and 1.9\% faster.  %---the expert, being a manual placer, takes several minutes to hours to perform. 

\old{In \pytorch Transformer, m-SCT and m-ETF place only the decoder's embedding and first multi-head attention layer on a separate device. Since this computation is independent of the encoder, m-SCT and m-ETF are able to exploit the parallelism inherent in the model. Rest of the decoder requires encoder's output  and is hence placed by \sysname on the same device as the encoder (in contrast to the expert) to minimize communication. This placement thus is 2.0-6.0\% faster than single-GPU and expert placements.}

In \pytorch Transformer, m-SCT and m-ETF placements are 2.0-6.0\% faster than single-GPU and expert placements. They place only the decoder's embedding and first multi-head attention layer on a separate device. Since this computation is independent of the encoder, m-SCT and m-ETF exploit the parallelism in the model. The rest of the decoder requires output of the encoder and is hence placed on the same device as the encoder (in contrast to the expert) to minimize communication.

These observations show that \sysname's m-ETF and m-SCT are able to generate placements with step times  in the same ballpark as the expert, while taking significantly less time to create a placement than the manual expert which takes minutes to hours.

\medskip\noindent{\it \textbf{m-TOPO. }}
Table \ref{tab:step_time} also shows that, \tfname's m-TOPO is significantly slower than m-ETF and m-SCT. {m-TOPO's step times are 5.8\%--26.4\% slower than m-ETF and 5.8\%--23.3\% slower than m-SCT.} After analysing m-TOPO we found that it places most of the encoder's LSTM layers  at the first two GPUs, and most of the decoder LSTM layers at  the other two GPUs. However, this parallelization is offset negatively by the   
high data transfers between the kernel weight and the LSTM cell operators for LSTM layers. 
{ Similarly, with Transformer in \pyname, m-TOPO's step time is 3.3\%--8.2\% slower than m-ETF and m-SCT. Essentially m-TOPO fails to exploit the parallelism between the encoder and the decoder.}

\medskip\noindent{\it \textbf{m-SCT vs. m-ETF. }}
Theoretical analysis in \cite{sct} shows SCT beating ETF and one would expect the same with m-SCT and m-ETF. In practice, the reverse is true---Table \ref{tab:step_time} shows that m-ETF's step times are faster than m-SCT's for 5 out of 6 
settings in   \tfname (it is slower only under sequence length 40, batch size 128), {and faster or equal in 3 out of all 4 settings in \pyname (it is slower only under Inception-V3, batch size 64).}

This 
behavior of m-SCT is because of two reasons. 
First, SCT's optimality proof 
relies on the assumption that the minimum operator computation time is larger than or equal to the maximum communication time. This does not hold in our experimental machine---a 4 B GPU-GPU transfer takes 50--200 ms while, in \tensorflow, many operators execute within 1 ms, and 67\% of Inception-V3's operators take under 50 ms. 
The m-SCT LP model (Section~\ref{sec:sct_algo}) assumes parallel data transfers from an operator to all its children. Our experimental machine only allows sequential transfers (Section~\ref{sec:sequential_data_transfer})\footnote{Faster data transfers between GPUs, e.g., via NVLink~\cite{nvlink}, have the potential to make m-SCT more competitive than m-ETF, but this is outside our scope.}. Overall, m-SCT and m-ETF are comparable in practice, with m-ETF having a slight edge in both placement time and step time.

\begin{table}
    \centering
    \caption{\it {\bf \sysname with Insufficient Memory.} Average Step Times (Training) in seconds of Placed Model Graphs  (Parentheses show Slowdown compared to Sufficient Memory for the same algorithm).
    }
    \label{tab:step_time_insuff_mem}
    %\begin{adjustbox}{width=\linewidth}
    \begin{tabular}{cccccccccc}
        \toprule
         %& & & & \multicolumn{2}{c}{Speedup} \\
         %\cmidrule{5-6}
         & Model & \shortstack{Batch \\ Size} & \shortstack{Memory \\ Fraction} & \shortstack{Single \\ GPU} & Expert & m-TOPO & m-ETF & m-SCT \\
         \midrule
         & \multirow{2}{*}{Inception-V3} & \multirow{2}{*}{32} & \multirow{2}{*}{0.3} &\multirow{2}{*}{OOM} & \multirow{2}{*}{OOM} & 
            \multirow{2}{*}{\shortstack{0.690 \\ (58.6\%)}} & \multirow{2}{*}{\shortstack{0.312 \\ (13.8\%)}} & \multirow{2}{*}{\shortstack{0.292 \\ (7.9\%)}} \\
         & & & & & \\
         {\rotatebox[origin=c]{90}{\rlap{\tensorflow} \hspace{1cm}}} & \multirow{2}{*}{GNMT} & \multirow{2}{*}{32} & \multirow{2}{*}{0.3} & \multirow{2}{*}{OOM} & \multirow{2}{*}{\shortstack{0.221 \\ (3.2\%)}} & 
            \multirow{2}{*}{\shortstack{0.272 \\ (2.6\%)}} & \multirow{2}{*}{\shortstack{0.230 \\ (2.6\%)}} & \multirow{2}{*}{\shortstack{0.212 \\ (0.0\%)}}\\
         & & & & & \\
         \midrule
         
          & \multirow{2}{*}{Inception-V3} & \multirow{2}{*}{32} & \multirow{2}{*}{0.3} & \multirow{2}{*}{OOM} & \multirow{2}{*}{OOM} & 
            \multirow{2}{*}{\shortstack{0.275 \\ (0.0\%)}} & \multirow{2}{*}{\shortstack{0.250 \\ (3.7\%)}} & \multirow{2}{*}{\shortstack{0.254 \\ (5.4\%)}} \\
            & & & & & \\
         & \multirow{2}{*}{Inception-V3} & \multirow{2}{*}{64} & \multirow{2}{*}{0.4} & \multirow{2}{*}{OOM} & \multirow{2}{*}{OOM} & 
            \multirow{2}{*}{\shortstack{0.537 \\ (0.0\%)}} & \multirow{2}{*}{\shortstack{0.527 \\ (13.3\%)}} & \multirow{2}{*}{\shortstack{0.535 \\ (16.1\%)}} \\
            & & & & & \\
        {\rotatebox[origin=c]{90}{\rlap{\pytorch}}} & \multirow{2}{*}{Transformer} & \multirow{2}{*}{64} & \multirow{2}{*}{0.3} & \multirow{2}{*}{OOM} & \multirow{2}{*}{0.257} & 
            \multirow{2}{*}{\shortstack{0.262 \\ (0.0\%)}} & \multirow{2}{*}{\shortstack{0.240 \\ (0.0\%)}} & \multirow{2}{*}{\shortstack{0.241 \\ (0.0\%)}} \\
            & & & & & \\

         \bottomrule
    \end{tabular}
    %\end{adjustbox}
\end{table}

\subsection{Placement with Insufficient Memory}
\label{eval:insufficient_memory}

Next, we limit per-GPU memory to a fraction of maximum available memory on the GPUs.  Table~\ref{tab:step_time_insuff_mem} shows  results for: 
{1) \sysname \tensorflow   with memory limited to 30\%, i.e., from 8 GB down to 2.4 GB, for:   Inception-V3 with batch size of 32, and GNMT with  batch size of 128 and  sequence length 40; and 2) \sysname \pytorch: memory limited to 30\% (Inception-V3 with  batch size of 32, Transformer with batch size 64) and 40\% memory limit (Inception-V3 with batch size 64).} 

A few notes follow on configuration changes in the experiments {with \tfname}. For GNMT,  co-placement (Section \ref{sec:co-placement}) remains enabled and we use the same configuration as Section~\ref{eval:sufficient_memory}. For Inception-V3 
we  disable  co-placement as otherwise it  generated a massive operator group, causing an Out of Memory error (OOM). 
Disabling  co-placement increases the number of operators to be placed from 2,620 to 7,077, and  placement time from 1 s to 10.3 s. No configuration changes were required for \pyname experiments.

\medskip\noindent{\it \textbf{Effect on Step Time. }}
Table~\ref{tab:step_time_insuff_mem} shows that the single GPU placer always suffers an OOM (Out of Memory) error. The expert placer OOMs for Inception-V3 (in both \tensorflow and \pytorch), but succeeds for \tensorflow GNMT and \pytorch Transformer. In comparison, all three variants of \sysname (m-TOPO, m-ETF, m-SCT) succeed in placing under insufficient memory under all five settings.

For Inception-V3 {in both \tfname and \pyname},  
only \sysname succeeds in placement. 
{Compared to the  sufficient memory cases  (Table~\ref{tab:step_time}), m-ETF and m-SCT provide step times that are only 13.8\% and 7.9\% worse in \tfname and,  
13.3\% and 16.1\% worse in \pyname respectively. }
m-TOPO in \tensorflow degrades by 58.6\% because of its  disabled   co-placement, which ballooned communication along the graph's critical path. In \pytorch there is no change in m-TOPO since the algorithm does not depend on the maximum limit as long as it is more than m-TOPO's per device cap.

For GNMT and Transformer, the overheads of all three \sysname algorithms and the expert are small 
{(shown as \% numbers within parentheses), meaning that with insufficient memory \sysname{} is nearly as fast as when memory is sufficient.} 

\begin{figure}
    \centering
    \begin{subfigure}[t]{0.5\textwidth}
        \centering
        \caption{}
        \label{fig:per_dev_memory_usage_tf}
        \includegraphics[width=1\linewidth]{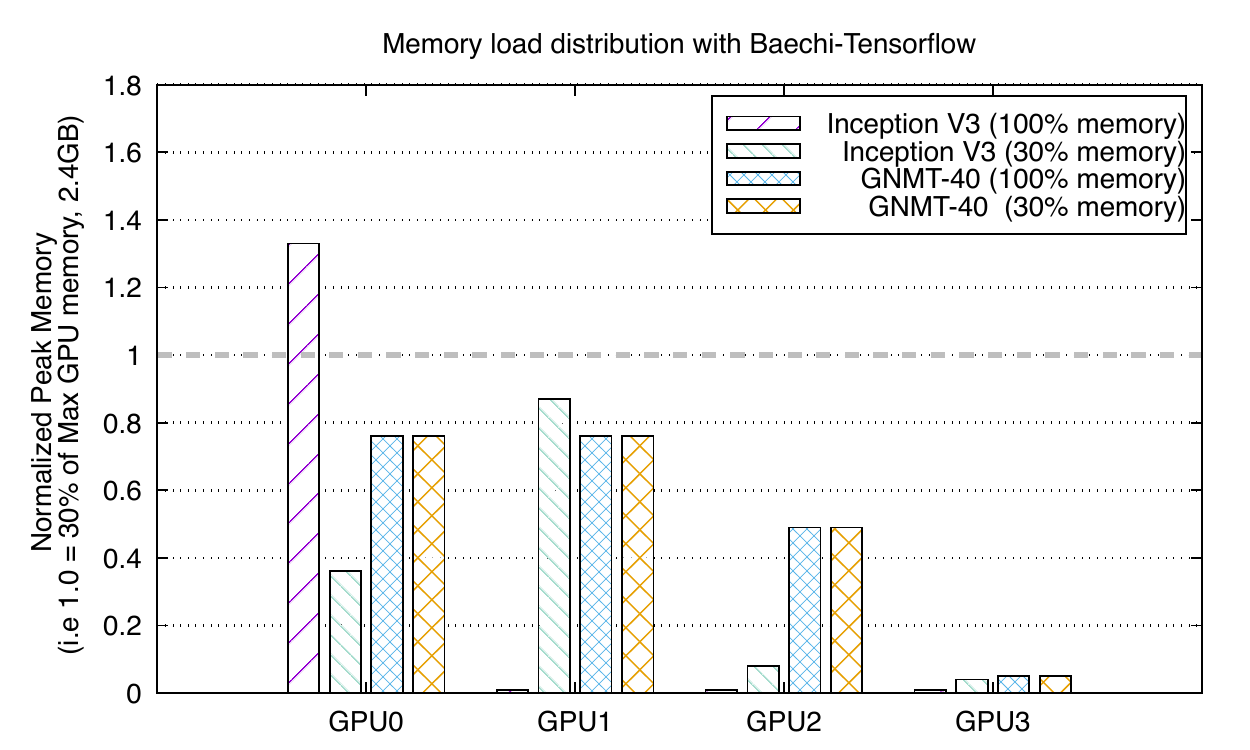}
    \end{subfigure}\hfill
    \begin{subfigure}[t]{0.5\textwidth}
        \centering
        \caption{}
        \label{fig:per_dev_memory_usage_py}
        \includegraphics[width=1\linewidth]{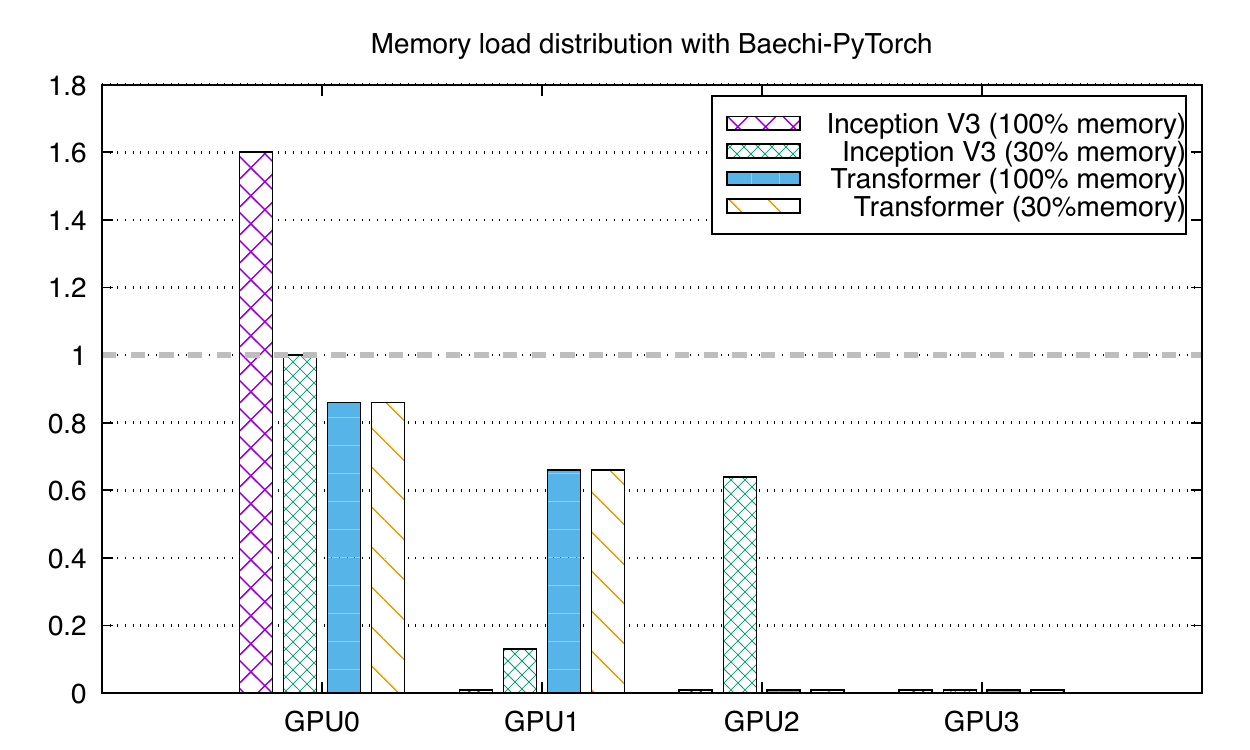}
    \end{subfigure}
    \caption{\it {\bf \sysname Load Balance of Memory Usage using m-SCT.} Dashed line is memory limit for each GPU (normalized). Note that 1.0 on y axis corresponds to 30\% of the max GPU memory (i.e. 2.4 GB in a 8 GB GPU) }
    \label{fig:per_dev_memory_usage}
\end{figure}

\medskip\noindent{\it \textbf{Load Distribution. }}
Figure~\ref{fig:per_dev_memory_usage} shows the peak memory usage, normalized to the memory limit for each GPU (insufficient memory case)
{for \tfname and \pyname.
In both \tfname and \pyname,} for Inception-V3,  

with a 30\% memory cap, a single GPU does not suffice, and that m-SCT relies on a mix of multiple GPUs. In particular, 2 of the 4 GPUs appear to be used more. This is because Inception-V3 has more barriers (sync points) than GNMT {in \tensorflow}, limiting Inception-V3's ability to parallelize effectively.

For \tensorflow GNMT and \pytorch Transformer, \sysname's m-SCT is able to load-balance more evenly (than Inception-V3)  across the GPUs, even when memory is sufficient.
In fact, {for both these cases}  we found that m-SCT generates an identical placement in both cases with sufficient and with insufficient memory. 
{This fact is also true for the expert, m-TOPO, and m-ETF. However specifically in case of GNMT with \tfname, their step times are 2.6--3.2\% {higher} than the sufficient memory cases (Table~\ref{tab:step_time}).}This slowdown is because of TensorFlow runtime memory optimizations.  
Concretely, when the memory usage approaches its limit, the TensorFlow runtime resorts to certain memory optimizations to decrease  peak memory usage. For the expert placement,  peak memory usage for one GPU device decreases from 2 GB (83\% of the memory limit) to 1.45 GB and thus the number of memory operations increases 6\% under insufficient memory. These memory optimizations do not kick in for m-SCT, making it faster than the expert.

\begin{figure}
    \centering
    \begin{subfigure}[t]{0.5\textwidth}
        \centering
        \caption{}
        \label{fig:sensitivity_tf}
        \includegraphics[width=1\linewidth]{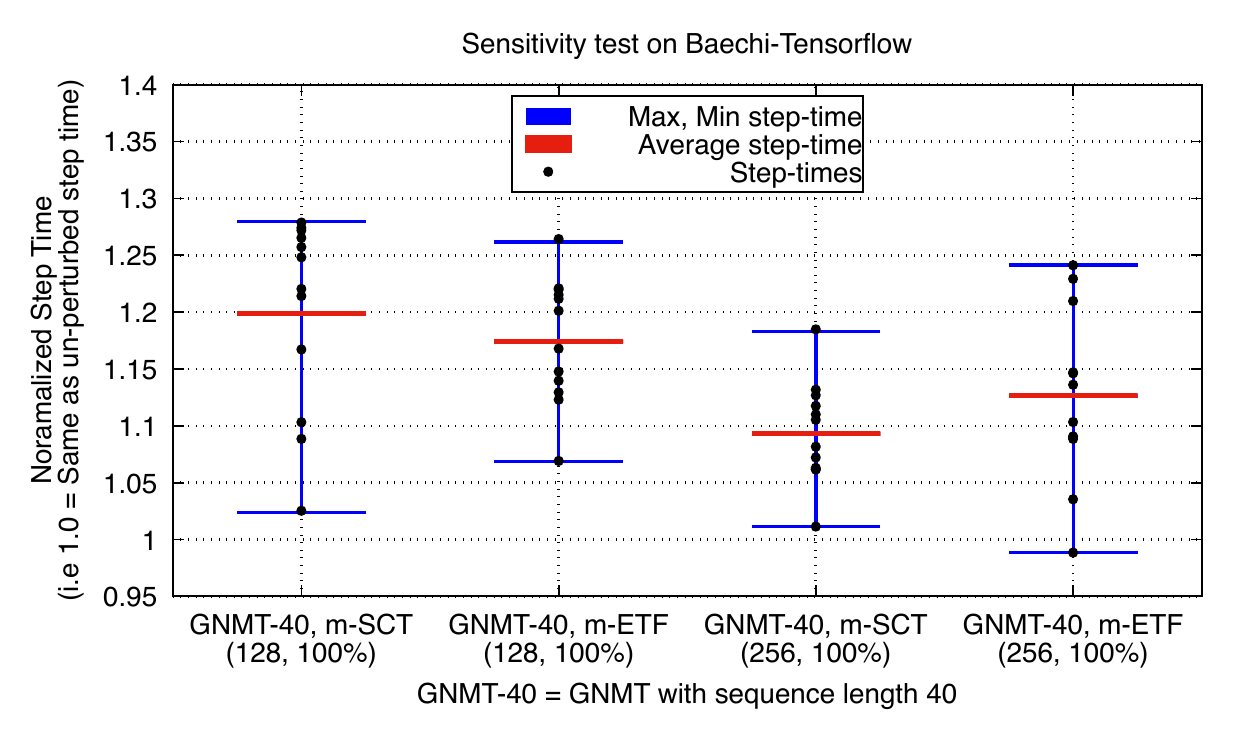}
    \end{subfigure}\hfill
    \begin{subfigure}[t]{0.5\textwidth}
        \centering
        \caption{}
        \label{fig:sensitivity_py}
        \includegraphics[width=1\linewidth]{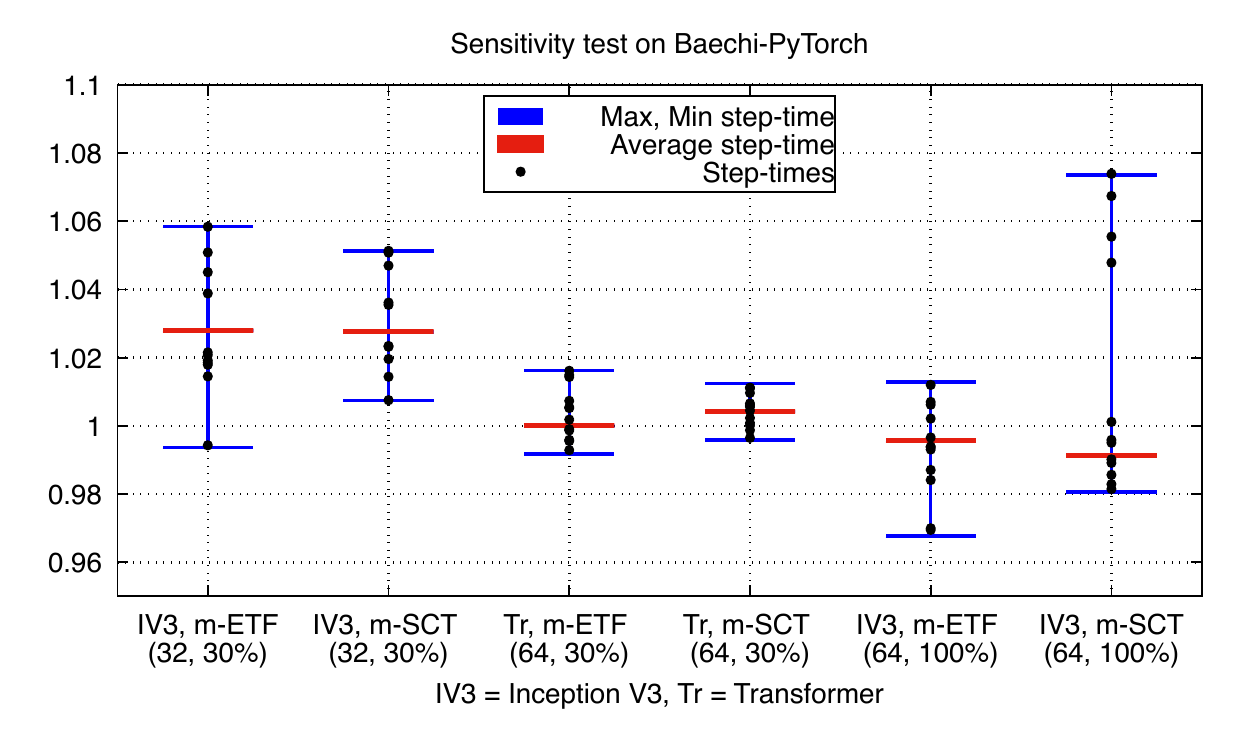}
    \end{subfigure}
    \caption{\it {\bf \sysname Sensitivity to Profiling Errors.} All computation and communication times are perturbed randomly by up to 20\% and the step time for placement generated by \sysname is measured. Values in X-axis parentheses are (Batch Size, Memory Fraction Available) 
    }
    \label{fig:sensitivity}
\end{figure}

\medskip\noindent{\it \textbf{Profile Sensitivity. }}
To measure \sysname's sensitivity to profiling errors, we perform runs where in each run all computation and communication profiles are randomly and independently perturbed by up to $\pm$20\%. This should account for errors in our time measurements as well as small speed differences between the device used to profile the model and devices where the model actually runs. Figure~\ref{fig:sensitivity} shows the perturbation in step-times of the resulting placements, w.r.t. to step-time without any perturbation of the profiles. %\newAdd
{Compared to the unperturbed step times, the step times with perturbed profiles remain within  a fraction of 0.99$\times$ to 1.3$\times$ in \tfname, and between 0.97$\times$ and 1.08$\times$ times in \pyname} 
Thus we conclude that m-SCT and m-ETF are resilient to reasonable levels of errors in profiled values.

\subsection{Benefit of \sysname Optimizations}
\label{sec:eval:opt}

\subsubsection{Benefit of \tfname Optimizations. }

Table~\ref{tab:optimization_results} shows the benefit from the combined optimizations of Section \ref{sec:co-placement} and \ref{sec:operator_fusion} in \tfname. We use Inception-V3 with batch size 32 and GNMT with batch size of 128. 
We use the m-SCT variant of \sysname. The experimental setup has 4 GPUs with sufficient memory.

\begin{table}[tb]
  \caption{
    \it {\bf Benefits of \tfname Optimizations (Section~\ref{sec:tf_colocation}).}  Number of Operators to be Placed, Placement Times in seconds, and Average Step Times in seconds. m-SCT.
  }
  \label{tab:optimization_results}
  \centering
  \begin{tabular}{ccccccc}
    \toprule
    \multirow{3}{*}{Model} & \multicolumn{3}{c}{Un-Optimized} & \multicolumn{3}{c}{Optimized} \\
    \cmidrule(r){2-4}
    \cmidrule(r){5-7}
     & \shortstack{Num. \\ Ops} & \shortstack{Placement \\ (seconds)} & \shortstack{Step \\ (seconds)} & \shortstack{Num. \\ Ops} & \shortstack{Placement \\ (seconds)} & \shortstack{Step \\ (seconds)} \\
    \midrule
    Inception-V3 & 6884 & 68.0 & 0.302 & 17 & 0.9 & 0.269 \\
    \multirow{2}{*}{\shortstack{GNMT \\ (length: 40)}} & \multirow{2}{*}{18050} & \multirow{2}{*}{275.1} & \multirow{2}{*}{0.580} & \multirow{2}{*}{542} & \multirow{2}{*}{1.2} & \multirow{2}{*}{0.212} \\
    & & & & & & \\
    \multirow{2}{*}{\shortstack{GNMT \\ (length: 50)}} & \multirow{2}{*}{22340} & \multirow{2}{*}{ 406.1} & \multirow{2}{*}{0.793} & \multirow{2}{*}{706} & \multirow{2}{*}{2.4} & \multirow{2}{*}{0.267} \\
    & & & & & & \\
    \bottomrule
  \end{tabular}
  %\end{adjustbox}
\end{table}

\begin{table*}[t]
    \caption{
        \it {\bf Benefits of communication protocol in \pyname{} (Section \ref{sec:py_comm}). } Step times in seconds without and with the protocol  }
    \label{tab:optimization_results_py}
    \centering
    \begin{tabular}{ccccc}
        \toprule
        
        Model & Algorithm & Without Protocol %Plain
        & With Protocol & \% Change\\
        %\shortstack{Algorithm \\ Size}
        \midrule
        
        \multirow{2}{*} {\shortstack{Inception V3 (32, 30\% memory)}}
            & m-ETF & 0.252 & 0.250 & 0.0\% \\
            & m-SCT & 0.268 & 0.254 & 5.5\% \\
        \multirow{2}{*} {\shortstack{Inception V3 (64, 40\% memory)}}
            & m-ETF & 0.551 & 0.528 &  4.3\%\\
            & m-SCT & 0.550 & 0.535 &  2.8\%\\
        \multirow{2}{*} {\shortstack{Transformer (64, 100\% memory)}}
            & m-ETF & 0.246 & 0.242 &  0.0\%\\
            & m-SCT & 0.246 & 0.244 &  0.0\%\\

        \bottomrule
    \end{tabular}
    %\end{adjustbox}
    \vspace{-0.1in}
\end{table*}

Overall, we observe that \tfname's combined optimizations achieve   75.6$\times$--229.3$\times$ speedup in  placement times, and 1.1$\times$--3.0$\times$ speedup in  step times. We discuss a few interesting aspects.  
Operator fusion (Section~\ref{sec:operator_fusion}) reduces both number of operators to be placed and thus also placement time. Forward-operator-based placement (Section \ref{sec:forward_operator_based_placement}) significantly speeds up placement. Concretely the latter optimization reduces the number of operators to be placed 2.7$\times$ for Inception-V3 and 6.5$\times$--7.0$\times$ for GNMT. This accelerates the placement times 13.7$\times$ for Inception-V3 and 20.2$\times$--31.4$\times$ for GNMT.

Co-placement (Section~\ref{sec:co-placement}) is efficient because it clusters operators. This  reduces step times. While  co-placement does not change the  operator count to be placed, it decreases placement time by reducing the overhead of calculating schedulable times.

\subsubsection{Benefit of \pyname  Communication Protocol.}
To evaluate the communication protocol in \pyname (Section~\ref{sec:py_comm}), we create a baseline plain wrapper. In it, each node transfers the inputs from devices of its parent (if different from module's device) by simply using blocking calls to \code{.to()} instead of using streams. Table \ref{tab:optimization_results_py} shows the comparison of step times. %\newAdd
{\pyname's communication protocol gives up to 5.5\% benefit with Inception-V3, and very little benefit under Transformer. This is because, in \pytorch, both these models have a strong linear spine, which creates fewer opportunities for parallelism.}

\section{Discussion and Limitations} \label{discussion}
\noindent{\textbf{Algorithmic Approaches vs Learning-based Approaches. }}

When we first implemented m-ETF and m-SCT, the placed models had very high step times because communication-intensive operators violated the SCT assumption (Table~\ref{fig:terminology}). We whittled away at this with a persistent effort at systems design and optimizations (outlined in Section~\ref{sec:design}), which played a major role in bringing the step times down. Although our exploration was efficient 
and we cycled new techniques and optimizations on a weekly basis, it took 1 person-year of effort to converge to what now appears in this paper for \tfname, and an additional 1 person-year for \pyname . This is indicative of the difficulties associated with implementing scheduling algorithms on today's open-source ML systems (and in a sense shows why existing learning-based approaches are so attractive!).  Nevertheless, our results show that the benefits   of algorithmic design were worth the exploratory pain.

{Our  experience also indicates reasons why developers (and companies!) often choose to ``jump'' so quickly towards  using learning-based (including RL-based) solutions for  scheduling problems: fast time to design (optimization of parameters and hyperparameters can be often be done as a rote task, rather than a creative task), and hence fast time to production. However, this comes at the expense of latter pain points in generalizing learning-based approaches to different architectures and models (in comparison, \sysname runs as-is, given an arbitrary model and a machine profile), as well as the high times to generate a placement using  learning-based approaches, which becomes a bottleneck in the exploratory design phase when the developer is iteratively  building and revising their model. We conclude that learning-based techniques (for any problem) should not be built in isolation from, or in lieu of, algorithmic-based approaches---but rather hand-in-hand with them.
}

\medskip \noindent{\textbf{
Limitations of \tfname
}.}
The peak memory usage of \tensorflow is highly dependent on the execution order of operators~\cite{ordering_chaos}. So \sysname would  benefit most if the framework (\tensorflow or \pytorch)  faithfully executed  operators in the same order as specified by \sysname's ES. For \pyname we enforce this via the ``Reordering Problem'' (Section~\ref{secmisc}). For \tensorflow while we do not enforce this ordering, and we observed in several runs of \tfname that \tensorflow deviated from this order,  yet  memory caps were not violated for m-SCT and m-ETF in \tfname runs. It is possible that if memory caps were tightened further (than 30\%, compared to our experiments),  engineering may be required for \tfname to force \tensorflow to follow the ES execution order.

\medskip \noindent{\textbf{Limitations of \pyname}.}

{\it \textbf{(i) Correctness issues with in-place operations}}: Inplace operations may lead to race conditions and incorrectness.
Concretely, select modules in \pytorch can be made to modify the input tensors in-place, e.g., ReLU with in-place flag set. \pyname's communication protocol (in Section \ref{sec:py_comm}) uses an independent \code{tx} stream to move out the tensors from a device. If the subsequent module in the compute stream is in-place, it may modify the tensor while it is being transferred out. This may  lead to an incorrect tensor. While we did not encounter such a cases in evaluation, a simple fix is to turn off the in-place feature for the module in question. This may however increase the memory consumption.

{\it \textbf{(ii) Weight sharing in Transformers}}: Currently the Assigner (Section~\ref{assigner}) does not support cases where weights are shared across multiple modules in the model (e.g., Transformers with embedder weight sharing \cite{transformer}).

{\it \textbf{(iii) Model code modifications}}: Some operations like concatenate and add, which take multiple inputs, must be defined as \pytorch modules. Only then can the Assigner ensure that tensors being concatenated or added are on the same device. In most cases, this is a few lines of code change. For instance, in Inception-V3, concatenate is used 7 times and add is used only once.

\section{Related work}

\noindent{\it \textbf{Data Parallelism (DP). }}
{Data parallelism (a.k.a. DP) refers to training the same model replicas with multiple partitioned data in parallel.} 
This is motivated by increasing sizes of datasets. 
MALT \cite{data_parallel_1} is a fault-tolerant, network-cost effective solution for data parallel ML. Another common data parallelism framework is NESL \cite{data_parallel_2}, a first-order functional language that enables developers to put irregular-parallel program in parallel devices. OptiML \cite{data_parallel_3} is a domain-specific language (DSL). Most major ML frameworks offer support for data parallelism~\cite{distributed_ml_3,tensorflow,pytorch}. 
{While DP typically replicates the model on each device, ZeRO \cite{zero} eliminates this redundancy and reduces memory consumption in DP.} {DP is orthogonal to model parallelism, and therefore DP techniques can be integrated into \sysname.}

\medskip\noindent{\it \textbf{Model Parallelism. }}
Compared 
to data parallelism, relatively fewer solutions exist for model parallelism. 
\igr{DistBelief~\cite{dean2012large} and STRADS~\cite{strads}} require the user to manually specify device placement, while the systems in \cite{krizhevsky2014one, le2013building} do not generalize to arbitrary ML models.

As discussed in Section~\ref{introduction}, reinforcement-learning based approaches have been popular lately to perform placement for model parallelism, including work from Google~\cite{colocRL, hierarchicalRL} and the Placeto system~\cite{placeto}. {ColocRL}~\cite{colocRL} trains a sequence-to-sequence model by RL to generate  placements of manually grouped subsets of TensorFlow operators. 
{HierarchicalRL}~\cite{hierarchicalRL} 
substitutes the human intervention for grouping operators with an ML model and jointly trains the ML models for operator grouping and device placements.
Placeto~\cite{placeto} proposes an approach that transfers learned device placement models to new ML models in order to minimize training times for the new model placements.

The original version of our paper \cite{baechiconf} both inspired  follow-up work \cite{pesto}, and also had parallel work \cite{amarpaper}, on algorithmic approaches to model parallelism. However \cite{amarpaper} is standalone, meaning that it is not integrated with TensorFlow or PyTorch, making a fair comparison with \sysname difficult.  
Pesto~\cite{pesto} presents direct comparisons with \sysname (\tfname)---the most important metric of placement times are similar for Pesto and \sysname, with small improvements for step time. Thus for practical purposes we consider Pesto to be comparable in performance to \sysname.

Works like PipeDream \cite{pipedream}, GPipe \cite{gpipe}, DAPPLE \cite{dapple}, PipeMare \cite{pipemare} introduce and optimize various aspects of Pipeline Parallelism. In Pipeline Parallelism, the model is usually vertically split into contiguous stages. 
{Amazon Sagemaker recently introduced automating Model and Pipeline Parallelism on their platform recently, though their code is proprietary\cite{sagemakerMP, sagemakerpaper}.} {Techniques for pipeline parallelism can be integrated orthogonally into \sysname.}

{
\medskip\noindent{\it \textbf{Model Parallelism for large language models. }}
Recent progress on very large language models like GPT \cite{gpt, gpt3} have given rise to works that focus specifically on Model Parallelism for such models like Megatron-LM \cite{megatron} and TeraPipe \cite{terapipe}. {However, these systems focus narrowly on Transformer models.}   
Alpa \cite{alpa1} combines intra and inter-operator parallelism. The inter-operator parallelism is limited to vertical splits and will not, for instance, place multiple parallel branches of a model on different devices, thus making it different from model parallelism.
}

\medskip\noindent{\it \textbf{Classical Parallel Scheduling. }}
Classical parallel scheduling, e.g.,  ETF~\cite{etf} and SCT~\cite{sct}, has been widely used in task scheduling on multiple computers. ETF  and  SCT are used as baselines by many subsequent works \cite{eyraud2015parallel,5715067,mohring1996scheduling,algo_lit_review,yang1994dsc}. %\igc
{None of these address memory constraints and a finite number of devices. 
For instance, \bjc{Eyraud-Dubois~et~al.~\cite{eyraud2015parallel}} investigate the execution of tree-shaped task graphs using multiple processors,
but without always obeying memory restrictions. }

\medskip\noindent{\it \textbf{TensorFlow Graph Optimizations. }}
Existing techniques~\cite{tf_graph_optmization, tf_xla} work only {\it after} the graph has been placed---e.g., to improve operations' performance---and thus are inapplicable. E.g., Running Grappler (TensorFlow's graph optimizer) generates {an optimized graph protobuf}, but it is unusable as it lacks certain metadata. \sysname's targeted problem is harder as we have to both optimize the graph %as well as 
and do placement.

\section{Conclusions} \label{conclusion}

We have proposed algorithmic  solutions to model parallelism, useful in scenarios where devices are memory-constrained or neural networks are massive. Among our three algorithms (m-ETF, m-TOPO, m-SCT), the m-SCT algorithm is provably within a constant factor of the optimal achievable training time. We have implemented these algorithms into our new \sysname system, as two systems \tfname and \pyname which are respectively usable in a modular manner with \tensorflow and \pytorch.

Experimental results 
showed that, across \tensorflow and \pytorch, our approaches reduce placement time by a factor of between {$654\times$--$206000 \times$} compared to today's state-of-the-art placement approaches which  are learning-based, while  increasing step time (makespan) by only up to {6.2\%} compared to expert placers. 
When memory is constrained further, while single GPU and expert placers suffer OOM errors, \sysname's algorithms, especially m-SCT and m-ETF, were able to place successfully. {Compared to sufficient memory the step times suffered an increase of only up to 13.8\% in \tensorflow and 16.1\% in \pytorch.
Further, \tfname's} optimizations help reduce {placement} time by {$75.6\times$--$229.3\times$}, 
and step time by {$1.1\times$--$3.0\times$}. 
We also conclude that m-SCT and m-ETF perform comparably, with m-ETF having a slight edge for slower networks.

{
The original version of our paper \cite{baechiconf} inspired  follow-up work \cite{pesto} along with  parallel work \cite{amarpaper}, on algorithmic approaches to model parallelism. Together, this new generation of algorithms for model parallelism offers the 
promise of speed, generalizability, predictability, and analyzability. These will be invaluable as learning models, both training and inference, move closer to edge devices and human-facing devices. 
}

\medskip\noindent{\it \textbf{Code. }} \sysname's code is openly available at the following link:\\ 
\url{http://dprg.cs.uiuc.edu/downloads.php}

% \section*{Acknowledgements}
\begin{acks}
\igr{This work was supported in part by the following grants: NSF IIS 1909577, and NSF CNS 1908888; as well as by generous gifts from Capital One, Schlumberger, and Microsoft. We thank Xiaojuan Ma for her invaluable help in reviewing the proofs in the appendix. 

}
\end{acks}

%\balance
\bibliographystyle{ACM-Reference-Format}
\bibliography{paper}


\begin{thebibliography}{82}

%%% the Web address.

\ifx \showCODEN    \undefined \def \showCODEN     #1{\unskip}     \fi
\ifx \showDOI      \undefined \def \showDOI       #1{#1}\fi
\ifx \showISBNx    \undefined \def \showISBNx     #1{\unskip}     \fi
\ifx \showISBNxiii \undefined \def \showISBNxiii  #1{\unskip}     \fi
\ifx \showISSN     \undefined \def \showISSN      #1{\unskip}     \fi
\ifx \showLCCN     \undefined \def \showLCCN      #1{\unskip}     \fi
\ifx \shownote     \undefined \def \shownote      #1{#1}          \fi
\ifx \showarticletitle \undefined \def \showarticletitle #1{#1}   \fi
\ifx \showURL      \undefined \def \showURL       {\relax}        \fi
% invisible to TeX
\providecommand\bibfield[2]{#2}
\providecommand\bibinfo[2]{#2}
\providecommand\natexlab[1]{#1}
\providecommand\showeprint[2][]{arXiv:#2}

\bibitem[Abadi et~al\mbox{.}(2016)]%
        {tensorflow}
\bibfield{author}{\bibinfo{person}{Mart{\'\i}n Abadi}, \bibinfo{person}{Paul
  Barham}, \bibinfo{person}{Jianmin Chen}, \bibinfo{person}{Zhifeng Chen},
  \bibinfo{person}{Andy Davis}, \bibinfo{person}{Jeffrey Dean},
  \bibinfo{person}{Matthieu Devin}, \bibinfo{person}{Sanjay Ghemawat},
  \bibinfo{person}{Geoffrey Irving}, \bibinfo{person}{Michael Isard},
  \bibinfo{person}{Manjunath Kudlur}, \bibinfo{person}{Josh Levenberg},
  \bibinfo{person}{Rajat Monga}, \bibinfo{person}{Sherry Moore},
  \bibinfo{person}{Derek~G. Murray}, \bibinfo{person}{Benoit Steiner},
  \bibinfo{person}{Paul Tucker}, \bibinfo{person}{Vijay Vasudevan},
  \bibinfo{person}{Pete Warden}, \bibinfo{person}{Martin Wicke},
  \bibinfo{person}{Yuan Yu}, {and} \bibinfo{person}{Xiaoqiang Zheng}.}
  \bibinfo{year}{2016}\natexlab{}.
\newblock \showarticletitle{{TensorFlow}: A System for Large-Scale Machine
  Learning}. In \bibinfo{booktitle}{\emph{12th {USENIX} Symposium on Operating
  Systems Design and Implementation ({OSDI} \textquotesingle16)}}.
  \bibinfo{publisher}{{USENIX} Association}, \bibinfo{pages}{265--283}.
\newblock
\showISBNx{978-1-931971-33-1}


\bibitem[Addanki et~al\mbox{.}(2019)]%
        {placeto}
\bibfield{author}{\bibinfo{person}{Ravichandra Addanki},
  \bibinfo{person}{Shaileshh Bojja~Venkatakrishnan}, \bibinfo{person}{Shreyan
  Gupta}, \bibinfo{person}{Hongzi Mao}, {and} \bibinfo{person}{Mohammad
  Alizadeh}.} \bibinfo{year}{2019}\natexlab{}.
\newblock \showarticletitle{Learning Generalizable Device Placement Algorithms
  for Distributed Machine Learning}.
\newblock In \bibinfo{booktitle}{\emph{Advances in Neural Information
  Processing Systems 32 (NeurIPS \textquotesingle19)}}.
  \bibinfo{publisher}{Curran Associates, Inc.}, \bibinfo{pages}{3981--3991}.
\newblock


\bibitem[Ahn et~al\mbox{.}(2020)]%
        {ordering_chaos}
\bibfield{author}{\bibinfo{person}{Byung~Hoon Ahn}, \bibinfo{person}{Jinwon
  Lee}, \bibinfo{person}{Jamie~Menjay Lin}, \bibinfo{person}{Hsin-Pai Cheng},
  \bibinfo{person}{Jilei Hou}, {and} \bibinfo{person}{Hadi Esmaeilzadeh}.}
  \bibinfo{year}{2020}\natexlab{}.
\newblock \showarticletitle{Ordering Chaos: Memory-Aware Scheduling of
  Irregularly Wired Neural Networks for Edge Devices}. In
  \bibinfo{booktitle}{\emph{Proceedings of Machine Learning and Systems}},
  \bibfield{editor}{\bibinfo{person}{I.~Dhillon},
  \bibinfo{person}{D.~Papailiopoulos}, {and} \bibinfo{person}{V.~Sze}} (Eds.),
  Vol.~\bibinfo{volume}{2}. \bibinfo{pages}{44--57}.
\newblock


\bibitem[Allen and Cocke(1972)]%
        {loop_unrolling}
\bibfield{author}{\bibinfo{person}{Frances.~E. Allen} {and}
  \bibinfo{person}{John Cocke}.} \bibinfo{year}{1972}\natexlab{}.
\newblock \showarticletitle{A Catalogue of Optimizing Transformations}.
\newblock \bibinfo{journal}{\emph{Design and Optimization of Compilers}}
  (\bibinfo{year}{1972}), \bibinfo{pages}{1--30}.
\newblock


\bibitem[Amazon(2022a)]%
        {sagemakerMP}
\bibfield{author}{\bibinfo{person}{Amazon}.} \bibinfo{year}{[Accessed
  2-July-2022]}\natexlab{a}.
\newblock \bibinfo{title}{{Amazon Sagemaker Model Parallelism}}.
\newblock
\newblock
\urldef\tempurl%
\url{{https://docs.aws.amazon.com/sagemaker/latest/dg/model-parallel.html}}
\showURL{%
\tempurl}


\bibitem[Amazon(2022b)]%
        {aws}
\bibfield{author}{\bibinfo{person}{Amazon}.} \bibinfo{year}{[Accessed
  2-July-2022]}\natexlab{b}.
\newblock \bibinfo{title}{{Amazon Web Services (AWS)}}.
\newblock
\newblock
\urldef\tempurl%
\url{https://aws.amazon.com}
\showURL{%
\tempurl}


\bibitem[Andersen and Andersen(2000)]%
        {mosek}
\bibfield{author}{\bibinfo{person}{Erling~D. Andersen} {and}
  \bibinfo{person}{Knud~D. Andersen}.} \bibinfo{year}{2000}\natexlab{}.
\newblock \showarticletitle{The Mosek Interior Point Optimizer for Linear
  Programming: An Implementation of the Homogeneous Algorithm}.
\newblock In \bibinfo{booktitle}{\emph{High Performance Optimization}}.
  \bibinfo{publisher}{Springer}, \bibinfo{pages}{197--232}.
\newblock


\bibitem[Bahdanau et~al\mbox{.}(2015)]%
        {BahdanauAttention}
\bibfield{author}{\bibinfo{person}{Dzmitry Bahdanau},
  \bibinfo{person}{Kyunghyun Cho}, {and} \bibinfo{person}{Yoshua Bengio}.}
  \bibinfo{year}{2015}\natexlab{}.
\newblock \showarticletitle{Neural Machine Translation by Jointly Learning to
  Align and Translate}. In \bibinfo{booktitle}{\emph{3rd International
  Conference on Learning Representations (ICLR \textquotesingle15)}}.
\newblock


\bibitem[Bartels and Golub(1969)]%
        {bartels1969simplex}
\bibfield{author}{\bibinfo{person}{Richard~H. Bartels} {and}
  \bibinfo{person}{Gene~H. Golub}.} \bibinfo{year}{1969}\natexlab{}.
\newblock \showarticletitle{The Simplex Method of Linear Programming Using LU
  Decomposition}.
\newblock \bibinfo{journal}{\emph{Commun. ACM}} \bibinfo{volume}{12},
  \bibinfo{number}{5} (\bibinfo{year}{1969}), \bibinfo{pages}{266--268}.
\newblock


\bibitem[Bergstrom and Reppy(2012)]%
        {data_parallel_2}
\bibfield{author}{\bibinfo{person}{Lars Bergstrom} {and} \bibinfo{person}{John
  Reppy}.} \bibinfo{year}{2012}\natexlab{}.
\newblock \showarticletitle{Nested Data-parallelism on the {GPU}}. In
  \bibinfo{booktitle}{\emph{17th ACM SIGPLAN International Conference on
  Functional Programming (ICFP \textquotesingle12)}}. \bibinfo{publisher}{ACM},
  \bibinfo{pages}{247--258}.
\newblock
\showISBNx{978-1-4503-1054-3}


\bibitem[Bonawitz et~al\mbox{.}(2019)]%
        {bonawitz2019towards}
\bibfield{author}{\bibinfo{person}{Keith Bonawitz}, \bibinfo{person}{Hubert
  Eichner}, \bibinfo{person}{Wolfgang Grieskamp}, \bibinfo{person}{Dzmitry
  Huba}, \bibinfo{person}{Alex Ingerman}, \bibinfo{person}{Vladimir Ivanov},
  \bibinfo{person}{Chlo\'{e} Kiddon}, \bibinfo{person}{Jakub Kone\v{c}n\'{y}},
  \bibinfo{person}{Stefano Mazzocchi}, \bibinfo{person}{Brendan McMahan},
  \bibinfo{person}{Timon Van~Overveldt}, \bibinfo{person}{David Petrou},
  \bibinfo{person}{Daniel Ramage}, {and} \bibinfo{person}{Jason Roselander}.}
  \bibinfo{year}{2019}\natexlab{}.
\newblock \showarticletitle{Towards Federated Learning at Scale: System
  Design}. In \bibinfo{booktitle}{\emph{Proceedings of Machine Learning and
  Systems}}, \bibfield{editor}{\bibinfo{person}{A.~Talwalkar},
  \bibinfo{person}{V.~Smith}, {and} \bibinfo{person}{M.~Zaharia}} (Eds.),
  Vol.~\bibinfo{volume}{1}. \bibinfo{pages}{374--388}.
\newblock
\urldef\tempurl%
\url{https://proceedings.mlsys.org/paper/2019/file/bd686fd640be98efaae0091fa301e613-Paper.pdf}
\showURL{%
\tempurl}


\bibitem[Bonawitz et~al\mbox{.}(2017)]%
        {bonawitz2017practical}
\bibfield{author}{\bibinfo{person}{Keith Bonawitz}, \bibinfo{person}{Vladimir
  Ivanov}, \bibinfo{person}{Ben Kreuter}, \bibinfo{person}{Antonio Marcedone},
  \bibinfo{person}{H~Brendan McMahan}, \bibinfo{person}{Sarvar Patel},
  \bibinfo{person}{Daniel Ramage}, \bibinfo{person}{Aaron Segal}, {and}
  \bibinfo{person}{Karn Seth}.} \bibinfo{year}{2017}\natexlab{}.
\newblock \showarticletitle{Practical Secure Aggregation for Privacy-Preserving
  Machine Learning}. In \bibinfo{booktitle}{\emph{24th ACM SIGSAC Conference on
  Computer and Communications Security (CCS \textquotesingle17)}}. ACM,
  \bibinfo{pages}{1175--1191}.
\newblock


\bibitem[Boyd and Vandenberghe(2004)]%
        {convex_optimization}
\bibfield{author}{\bibinfo{person}{Stephen Boyd} {and} \bibinfo{person}{Lieven
  Vandenberghe}.} \bibinfo{year}{2004}\natexlab{}.
\newblock \bibinfo{booktitle}{\emph{Convex Optimization}}.
\newblock \bibinfo{publisher}{Cambridge University Press}.
\newblock
\urldef\tempurl%
\url{https://doi.org/10.1017/CBO9780511804441}
\showDOI{\tempurl}


\bibitem[Brown et~al\mbox{.}(2020a)]%
        {gpt}
\bibfield{author}{\bibinfo{person}{Tom~B. Brown}, \bibinfo{person}{Benjamin
  Mann}, \bibinfo{person}{Nick Ryder}, \bibinfo{person}{Melanie Subbiah},
  \bibinfo{person}{Jared Kaplan}, \bibinfo{person}{Prafulla Dhariwal},
  \bibinfo{person}{Arvind Neelakantan}, \bibinfo{person}{Pranav Shyam},
  \bibinfo{person}{Girish Sastry}, \bibinfo{person}{Amanda Askell},
  \bibinfo{person}{Sandhini Agarwal}, \bibinfo{person}{Ariel Herbert{-}Voss},
  \bibinfo{person}{Gretchen Krueger}, \bibinfo{person}{Tom Henighan},
  \bibinfo{person}{Rewon Child}, \bibinfo{person}{Aditya Ramesh},
  \bibinfo{person}{Daniel~M. Ziegler}, \bibinfo{person}{Jeffrey Wu},
  \bibinfo{person}{Clemens Winter}, \bibinfo{person}{Christopher Hesse},
  \bibinfo{person}{Mark Chen}, \bibinfo{person}{Eric Sigler},
  \bibinfo{person}{Mateusz Litwin}, \bibinfo{person}{Scott Gray},
  \bibinfo{person}{Benjamin Chess}, \bibinfo{person}{Jack Clark},
  \bibinfo{person}{Christopher Berner}, \bibinfo{person}{Sam McCandlish},
  \bibinfo{person}{Alec Radford}, \bibinfo{person}{Ilya Sutskever}, {and}
  \bibinfo{person}{Dario Amodei}.} \bibinfo{year}{2020}\natexlab{a}.
\newblock \showarticletitle{Language Models are Few-Shot Learners}.
\newblock \bibinfo{journal}{\emph{CoRR}}  \bibinfo{volume}{abs/2005.14165}
  (\bibinfo{year}{2020}).
\newblock
\showeprint[arXiv]{2005.14165}
\urldef\tempurl%
\url{https://arxiv.org/abs/2005.14165}
\showURL{%
\tempurl}


\bibitem[Brown et~al\mbox{.}(2020b)]%
        {gpt3}
\bibfield{author}{\bibinfo{person}{Tom~B. Brown}, \bibinfo{person}{Benjamin
  Mann}, \bibinfo{person}{Nick Ryder}, \bibinfo{person}{Melanie Subbiah},
  \bibinfo{person}{Jared Kaplan}, \bibinfo{person}{Prafulla Dhariwal},
  \bibinfo{person}{Arvind Neelakantan}, \bibinfo{person}{Pranav Shyam},
  \bibinfo{person}{Girish Sastry}, \bibinfo{person}{Amanda Askell},
  \bibinfo{person}{Sandhini Agarwal}, \bibinfo{person}{Ariel Herbert{-}Voss},
  \bibinfo{person}{Gretchen Krueger}, \bibinfo{person}{Tom Henighan},
  \bibinfo{person}{Rewon Child}, \bibinfo{person}{Aditya Ramesh},
  \bibinfo{person}{Daniel~M. Ziegler}, \bibinfo{person}{Jeffrey Wu},
  \bibinfo{person}{Clemens Winter}, \bibinfo{person}{Christopher Hesse},
  \bibinfo{person}{Mark Chen}, \bibinfo{person}{Eric Sigler},
  \bibinfo{person}{Mateusz Litwin}, \bibinfo{person}{Scott Gray},
  \bibinfo{person}{Benjamin Chess}, \bibinfo{person}{Jack Clark},
  \bibinfo{person}{Christopher Berner}, \bibinfo{person}{Sam McCandlish},
  \bibinfo{person}{Alec Radford}, \bibinfo{person}{Ilya Sutskever}, {and}
  \bibinfo{person}{Dario Amodei}.} \bibinfo{year}{2020}\natexlab{b}.
\newblock \showarticletitle{Language Models are Few-Shot Learners}.
\newblock \bibinfo{journal}{\emph{CoRR}}  \bibinfo{volume}{abs/2005.14165}
  (\bibinfo{year}{2020}).
\newblock
\showeprint[arXiv]{2005.14165}
\urldef\tempurl%
\url{https://arxiv.org/abs/2005.14165}
\showURL{%
\tempurl}


\bibitem[Chen et~al\mbox{.}(2015)]%
        {distributed_ml_3}
\bibfield{author}{\bibinfo{person}{Tianqi Chen}, \bibinfo{person}{Mu Li},
  \bibinfo{person}{Yutian Li}, \bibinfo{person}{Min Lin},
  \bibinfo{person}{Naiyan Wang}, \bibinfo{person}{Minjie Wang},
  \bibinfo{person}{Tianjun Xiao}, \bibinfo{person}{Bing Xu},
  \bibinfo{person}{Chiyuan Zhang}, {and} \bibinfo{person}{Zheng Zhang}.}
  \bibinfo{year}{2015}\natexlab{}.
\newblock \showarticletitle{MXNet: {A} Flexible and Efficient Machine Learning
  Library for Heterogeneous Distributed Systems}.
\newblock \bibinfo{journal}{\emph{CoRR}}  \bibinfo{volume}{abs/1512.01274}
  (\bibinfo{year}{2015}).
\newblock
\showeprint[arXiv]{1512.01274}
\urldef\tempurl%
\url{http://arxiv.org/abs/1512.01274}
\showURL{%
\tempurl}


\bibitem[Chen et~al\mbox{.}(2016)]%
        {chen2016training}
\bibfield{author}{\bibinfo{person}{Tianqi Chen}, \bibinfo{person}{Bing Xu},
  \bibinfo{person}{Chiyuan Zhang}, {and} \bibinfo{person}{Carlos Guestrin}.}
  \bibinfo{year}{2016}\natexlab{}.
\newblock \showarticletitle{Training Deep Nets with Sublinear Memory Cost}.
\newblock \bibinfo{journal}{\emph{CoRR}}  \bibinfo{volume}{abs/1604.06174}
  (\bibinfo{year}{2016}).
\newblock
\showeprint[arXiv]{1604.06174}
\urldef\tempurl%
\url{http://arxiv.org/abs/1604.06174}
\showURL{%
\tempurl}


\bibitem[Chen et~al\mbox{.}(2012)]%
        {cloudera_workload}
\bibfield{author}{\bibinfo{person}{Yanpei Chen}, \bibinfo{person}{Sara
  Alspaugh}, {and} \bibinfo{person}{Randy Katz}.}
  \bibinfo{year}{2012}\natexlab{}.
\newblock \showarticletitle{Interactive Analytical Processing in Big Data
  Systems: A Cross-Industry Study of {MapReduce} Workloads}.
\newblock \bibinfo{journal}{\emph{Proceedings of the VLDB Endowment}}
  \bibinfo{volume}{5}, \bibinfo{number}{12} (\bibinfo{year}{2012}).
\newblock


\bibitem[Dean et~al\mbox{.}(2012)]%
        {dean2012large}
\bibfield{author}{\bibinfo{person}{Jeffrey Dean}, \bibinfo{person}{Greg~S.
  Corrado}, \bibinfo{person}{Rajat Monga}, \bibinfo{person}{Kai Chen},
  \bibinfo{person}{Matthieu Devin}, \bibinfo{person}{Quoc~V. Le},
  \bibinfo{person}{Mark~Z. Mao}, \bibinfo{person}{Marc’Aurelio Ranzato},
  \bibinfo{person}{Andrew Senior}, \bibinfo{person}{Paul Tucker},
  \bibinfo{person}{Ke Yang}, {and} \bibinfo{person}{Andrew~Y. Ng}.}
  \bibinfo{year}{2012}\natexlab{}.
\newblock \showarticletitle{Large Scale Distributed Deep Networks}. In
  \bibinfo{booktitle}{\emph{Advances in Neural Information Processing Systems
  25 (NIPS \textquotesingle12)}}. \bibinfo{publisher}{Curran Associates Inc.},
  \bibinfo{pages}{1223–1231}.
\newblock


\bibitem[Eyraud-Dubois et~al\mbox{.}(2015)]%
        {eyraud2015parallel}
\bibfield{author}{\bibinfo{person}{Lionel Eyraud-Dubois},
  \bibinfo{person}{Loris Marchal}, \bibinfo{person}{Oliver Sinnen}, {and}
  \bibinfo{person}{Fr\'{e}d\'{e}ric Vivien}.} \bibinfo{year}{2015}\natexlab{}.
\newblock \showarticletitle{Parallel Scheduling of Task Trees with Limited
  Memory}.
\newblock \bibinfo{journal}{\emph{ACM Transactions on Parallel Computing}}
  \bibinfo{volume}{2}, \bibinfo{number}{2} (\bibinfo{year}{2015}),
  \bibinfo{pages}{1--37}.
\newblock


\bibitem[Face(2022)]%
        {transformersplit1}
\bibfield{author}{\bibinfo{person}{Hugging Face}.} \bibinfo{year}{[Accessed
  2-July-2022]}\natexlab{}.
\newblock \bibinfo{title}{{Splitting a Transformer in PyTorch (see flat device
  map)}}.
\newblock
\newblock
\urldef\tempurl%
\url{{https://github.com/huggingface/transformers/pull/9384}}
\showURL{%
\tempurl}


\bibitem[Fan et~al\mbox{.}(2021)]%
        {dapple}
\bibfield{author}{\bibinfo{person}{Shiqing Fan}, \bibinfo{person}{Yi Rong},
  \bibinfo{person}{Chen Meng}, \bibinfo{person}{Zongyan Cao},
  \bibinfo{person}{Siyu Wang}, \bibinfo{person}{Zhen Zheng},
  \bibinfo{person}{Chuan Wu}, \bibinfo{person}{Guoping Long},
  \bibinfo{person}{Jun Yang}, \bibinfo{person}{Lixue Xia},
  \bibinfo{person}{Lansong Diao}, \bibinfo{person}{Xiaoyong Liu}, {and}
  \bibinfo{person}{Wei Lin}.} \bibinfo{year}{2021}\natexlab{}.
\newblock \showarticletitle{DAPPLE: A Pipelined Data Parallel Approach for
  Training Large Models}. In \bibinfo{booktitle}{\emph{Proceedings of the 26th
  ACM SIGPLAN Symposium on Principles and Practice of Parallel Programming}}
  (Virtual Event, Republic of Korea) \emph{(\bibinfo{series}{PPoPP '21})}.
  \bibinfo{publisher}{Association for Computing Machinery},
  \bibinfo{address}{New York, NY, USA}, \bibinfo{pages}{431–445}.
\newblock
\showISBNx{9781450382946}
\urldef\tempurl%
\url{https://doi.org/10.1145/3437801.3441593}
\showDOI{\tempurl}


\bibitem[GitHub(2022)]%
        {transformerCode}
\bibfield{author}{\bibinfo{person}{GitHub}.} \bibinfo{year}{[Accessed
  2-July-2022]}\natexlab{}.
\newblock \bibinfo{title}{{Transformer implementation in \pytorch}}.
\newblock
\newblock
\urldef\tempurl%
\url{{https://github.com/jadore801120/attention-is-all-you-need-pytorch}}
\showURL{%
\tempurl}


\bibitem[Google(2022)]%
        {google_cloud}
\bibfield{author}{\bibinfo{person}{Google}.} \bibinfo{year}{[Accessed
  2-July-2022]}\natexlab{}.
\newblock \bibinfo{title}{{Google Cloud Platform}}.
\newblock
\newblock
\urldef\tempurl%
\url{https://cloud.google.com/gcp/}
\showURL{%
\tempurl}


\bibitem[Hafeez et~al\mbox{.}(2021)]%
        {pesto}
\bibfield{author}{\bibinfo{person}{Ubaid~Ullah Hafeez}, \bibinfo{person}{Xiao
  Sun}, \bibinfo{person}{Anshul Gandhi}, {and} \bibinfo{person}{Zhenhua Liu}.}
  \bibinfo{year}{2021}\natexlab{}.
\newblock \showarticletitle{Towards Optimal Placement and Scheduling of DNN
  Operations with Pesto}. In \bibinfo{booktitle}{\emph{Proceedings of the 22nd
  International Middleware Conference}} (Qu\'{e}bec city, Canada)
  \emph{(\bibinfo{series}{Middleware '21})}. \bibinfo{publisher}{Association
  for Computing Machinery}, \bibinfo{address}{New York, NY, USA},
  \bibinfo{pages}{39–51}.
\newblock
\showISBNx{9781450385343}
\urldef\tempurl%
\url{https://doi.org/10.1145/3464298.3476132}
\showDOI{\tempurl}


\bibitem[Hanen and Munier(1995)]%
        {sct}
\bibfield{author}{\bibinfo{person}{Claire Hanen} {and} \bibinfo{person}{Alix
  Munier}.} \bibinfo{year}{1995}\natexlab{}.
\newblock \showarticletitle{An Approximation Algorithm for Scheduling Dependent
  Tasks on m Processors with Small Communication Delays}. In
  \bibinfo{booktitle}{\emph{4th INRIA/IEEE Symposium on Emerging Technologies
  and Factory Automation (ETFA \textquotesingle95)}}, Vol.~\bibinfo{volume}{1}.
  \bibinfo{publisher}{IEEE}, \bibinfo{pages}{167--189}.
\newblock


\bibitem[Harlap et~al\mbox{.}(2018)]%
        {pipedream}
\bibfield{author}{\bibinfo{person}{Aaron Harlap}, \bibinfo{person}{Deepak
  Narayanan}, \bibinfo{person}{Amar Phanishayee}, \bibinfo{person}{Vivek
  Seshadri}, \bibinfo{person}{Nikhil~R. Devanur}, \bibinfo{person}{Gregory~R.
  Ganger}, {and} \bibinfo{person}{Phillip~B. Gibbons}.}
  \bibinfo{year}{2018}\natexlab{}.
\newblock \showarticletitle{PipeDream: Fast and Efficient Pipeline Parallel
  {DNN} Training}.
\newblock \bibinfo{journal}{\emph{CoRR}}  \bibinfo{volume}{abs/1806.03377}
  (\bibinfo{year}{2018}).
\newblock
\showeprint[arXiv]{1806.03377}
\urldef\tempurl%
\url{http://arxiv.org/abs/1806.03377}
\showURL{%
\tempurl}


\bibitem[Hinton et~al\mbox{.}(2012)]%
        {rmsprop}
\bibfield{author}{\bibinfo{person}{Geoffrey Hinton}, \bibinfo{person}{Nitish
  Srivastava}, {and} \bibinfo{person}{Kevin Swersky}.}
  \bibinfo{year}{2012}\natexlab{}.
\newblock \bibinfo{title}{Neural Networks for Machine Learning Lecture 6a
  Overview of Mini-Batch Gradient Descent}.
\newblock
\newblock
\urldef\tempurl%
\url{https://www.cs.toronto.edu/~tijmen/csc321/slides/lecture_slides_lec6.pdf}
\showURL{%
\tempurl}


\bibitem[Hoogeveen et~al\mbox{.}(1994)]%
        {complexity}
\bibfield{author}{\bibinfo{person}{J.A. Hoogeveen}, \bibinfo{person}{Jan~K.
  Lenstra}, {and} \bibinfo{person}{Bart Veltman}.}
  \bibinfo{year}{1994}\natexlab{}.
\newblock \showarticletitle{Three, Four, Five, Six, or the Complexity of
  Scheduling with Communication Delays}.
\newblock \bibinfo{journal}{\emph{Operations Research Letters}}
  \bibinfo{volume}{16}, \bibinfo{number}{3} (\bibinfo{year}{1994}),
  \bibinfo{pages}{129--137}.
\newblock
\showISSN{0167-6377}


\bibitem[{Hu} et~al\mbox{.}(2010)]%
        {5715067}
\bibfield{author}{\bibinfo{person}{Jinhua {Hu}}, \bibinfo{person}{Jianhua
  {Gu}}, \bibinfo{person}{Guofei {Sun}}, {and} \bibinfo{person}{Tianhai
  {Zhao}}.} \bibinfo{year}{2010}\natexlab{}.
\newblock \showarticletitle{A Scheduling Strategy on Load Balancing of Virtual
  Machine Resources in Cloud Computing Environment}. In
  \bibinfo{booktitle}{\emph{3rd International Symposium on Parallel
  Architectures, Algorithms and Programming (PAAP \textquotesingle10)}}.
  \bibinfo{publisher}{IEEE}, \bibinfo{pages}{89--96}.
\newblock


\bibitem[Huang et~al\mbox{.}(2019)]%
        {gpipe}
\bibfield{author}{\bibinfo{person}{Yanping Huang}, \bibinfo{person}{Youlong
  Cheng}, \bibinfo{person}{Ankur Bapna}, \bibinfo{person}{Orhan Firat},
  \bibinfo{person}{Mia~Xu Chen}, \bibinfo{person}{Dehao Chen},
  \bibinfo{person}{HyoukJoong Lee}, \bibinfo{person}{Jiquan Ngiam},
  \bibinfo{person}{Quoc~V. Le}, \bibinfo{person}{Yonghui Wu}, {and}
  \bibinfo{person}{Zhifeng Chen}.} \bibinfo{year}{2019}\natexlab{}.
\newblock \bibinfo{booktitle}{\emph{GPipe: Efficient Training of Giant Neural
  Networks Using Pipeline Parallelism}}.
\newblock \bibinfo{publisher}{Curran Associates Inc.}, \bibinfo{address}{Red
  Hook, NY, USA}.
\newblock


\bibitem[Hwang et~al\mbox{.}(1989)]%
        {etf}
\bibfield{author}{\bibinfo{person}{Jing-Jang Hwang},
  \bibinfo{person}{Yuan-Chieh Chow}, \bibinfo{person}{Frank~D. Anger}, {and}
  \bibinfo{person}{Chung-Yee Lee}.} \bibinfo{year}{1989}\natexlab{}.
\newblock \showarticletitle{Scheduling Precedence Graphs in Systems with
  Interprocessor Communication Times}.
\newblock \bibinfo{journal}{\emph{SIAM Jornal on Computing}}
  \bibinfo{volume}{18}, \bibinfo{number}{2} (\bibinfo{year}{1989}),
  \bibinfo{pages}{244--257}.
\newblock
\showISSN{0097-5397}


\bibitem[Jeon et~al\mbox{.}(2020)]%
        {baechiconf}
\bibfield{author}{\bibinfo{person}{Beomyeol Jeon}, \bibinfo{person}{Linda Cai},
  \bibinfo{person}{Pallavi Srivastava}, \bibinfo{person}{Jintao Jiang},
  \bibinfo{person}{Xiaolan Ke}, \bibinfo{person}{Yitao Meng},
  \bibinfo{person}{Cong Xie}, {and} \bibinfo{person}{Indranil Gupta}.}
  \bibinfo{year}{2020}\natexlab{}.
\newblock \showarticletitle{Baechi: Fast Device Placement of Machine Learning
  Graphs}. In \bibinfo{booktitle}{\emph{Proceedings of the 11th ACM Symposium
  on Cloud Computing}} (Virtual Event, USA) \emph{(\bibinfo{series}{SoCC
  '20})}. \bibinfo{publisher}{Association for Computing Machinery},
  \bibinfo{address}{New York, NY, USA}, \bibinfo{pages}{416–430}.
\newblock
\showISBNx{9781450381376}
\urldef\tempurl%
\url{https://doi.org/10.1145/3419111.3421302}
\showDOI{\tempurl}


\bibitem[Jia et~al\mbox{.}(2014)]%
        {jia2014caffe}
\bibfield{author}{\bibinfo{person}{Yangqing Jia}, \bibinfo{person}{Evan
  Shelhamer}, \bibinfo{person}{Jeff Donahue}, \bibinfo{person}{Sergey Karayev},
  \bibinfo{person}{Jonathan Long}, \bibinfo{person}{Ross Girshick},
  \bibinfo{person}{Sergio Guadarrama}, {and} \bibinfo{person}{Trevor Darrell}.}
  \bibinfo{year}{2014}\natexlab{}.
\newblock \showarticletitle{{Caffe}: Convolutional Architecture for Fast
  Feature Embedding}. In \bibinfo{booktitle}{\emph{22nd ACM International
  Conference on Multimedia (MM \textquotesingle14)}}. ACM,
  \bibinfo{pages}{675--678}.
\newblock


\bibitem[Jia et~al\mbox{.}(2018)]%
        {jia2018exploring}
\bibfield{author}{\bibinfo{person}{Zhihao Jia}, \bibinfo{person}{Sina Lin},
  \bibinfo{person}{Charles~R. Qi}, {and} \bibinfo{person}{Alex Aiken}.}
  \bibinfo{year}{2018}\natexlab{}.
\newblock \showarticletitle{Exploring Hidden Dimensions in Accelerating
  Convolutional Neural Networks}. In \bibinfo{booktitle}{\emph{35th
  International Conference on Machine Learning (ICML \textquotesingle18)}}.
  \bibinfo{publisher}{PMLR}, \bibinfo{pages}{2274--2283}.
\newblock


\bibitem[Jia et~al\mbox{.}(2019)]%
        {jia2018data}
\bibfield{author}{\bibinfo{person}{Zhihao Jia}, \bibinfo{person}{Matei
  Zaharia}, {and} \bibinfo{person}{Alex Aiken}.}
  \bibinfo{year}{2019}\natexlab{}.
\newblock \showarticletitle{Beyond Data and Model Parallelism for Deep Neural
  Networks}.
\newblock In \bibinfo{booktitle}{\emph{2nd Conference on Machine Learning and
  Systems (MLSys \textquotesingle19)}}. \bibinfo{pages}{1--13}.
\newblock


\bibitem[Kahn(1962)]%
        {topo}
\bibfield{author}{\bibinfo{person}{Arthur~B. Kahn}.}
  \bibinfo{year}{1962}\natexlab{}.
\newblock \showarticletitle{Topological Sorting of Large Networks}.
\newblock \bibinfo{journal}{\emph{Commun. ACM}} \bibinfo{volume}{5},
  \bibinfo{number}{11} (\bibinfo{year}{1962}), \bibinfo{pages}{558--562}.
\newblock
\showISSN{0001-0782}


\bibitem[Karakus et~al\mbox{.}(2021)]%
        {sagemakerpaper}
\bibfield{author}{\bibinfo{person}{Can Karakus}, \bibinfo{person}{Rahul
  Huilgol}, \bibinfo{person}{Fei Wu}, \bibinfo{person}{Anirudh Subramanian},
  \bibinfo{person}{Cade Daniel}, \bibinfo{person}{Derya {\c{C}}avdar},
  \bibinfo{person}{Teng Xu}, \bibinfo{person}{Haohan Chen},
  \bibinfo{person}{Arash Rahnama}, {and} \bibinfo{person}{Luis Quintela}.}
  \bibinfo{year}{2021}\natexlab{}.
\newblock \showarticletitle{Amazon SageMaker Model Parallelism: {A} General and
  Flexible Framework for Large Model Training}.
\newblock \bibinfo{journal}{\emph{CoRR}}  \bibinfo{volume}{abs/2111.05972}
  (\bibinfo{year}{2021}).
\newblock
\showeprint[arXiv]{2111.05972}
\urldef\tempurl%
\url{https://arxiv.org/abs/2111.05972}
\showURL{%
\tempurl}


\bibitem[Karmarkar(1984)]%
        {interior_point}
\bibfield{author}{\bibinfo{person}{Narendra Karmarkar}.}
  \bibinfo{year}{1984}\natexlab{}.
\newblock \showarticletitle{A new Polynomial-Time Algorithm for Linear
  Programming}. In \bibinfo{booktitle}{\emph{16th Annual ACM Symposium on
  Theory of Computing (STOC \textquotesingle84)}}. ACM,
  \bibinfo{pages}{302--311}.
\newblock


\bibitem[Kim et~al\mbox{.}(2016)]%
        {strads}
\bibfield{author}{\bibinfo{person}{Jin~Kyu Kim}, \bibinfo{person}{Qirong Ho},
  \bibinfo{person}{Seunghak Lee}, \bibinfo{person}{Xun Zheng},
  \bibinfo{person}{Wei Dai}, \bibinfo{person}{Garth~A. Gibson}, {and}
  \bibinfo{person}{Eric~P. Xing}.} \bibinfo{year}{2016}\natexlab{}.
\newblock \showarticletitle{{STRADS}: A Distributed Framework for Scheduled
  Model Parallel Machine Learning}. In \bibinfo{booktitle}{\emph{11th European
  Conference on Computer Systems (EuroSys \textquotesingle 16)}}.
  \bibinfo{publisher}{ACM}, Article \bibinfo{articleno}{5},
  \bibinfo{numpages}{16}~pages.
\newblock
\showISBNx{9781450342407}


\bibitem[Kraska et~al\mbox{.}(2013)]%
        {distributed_ml_1}
\bibfield{author}{\bibinfo{person}{Tim Kraska}, \bibinfo{person}{Ameet
  Talwalkar}, {and} \bibinfo{person}{John Duchi}.}
  \bibinfo{year}{2013}\natexlab{}.
\newblock \showarticletitle{{MLbase}: A Distributed Machine-Learning System}.
  In \bibinfo{booktitle}{\emph{6th Biennial Conference on Innovative Data
  Systems Research (CIDR \textquotesingle13)}}.
\newblock


\bibitem[Krizhevsky(2014)]%
        {krizhevsky2014one}
\bibfield{author}{\bibinfo{person}{Alex Krizhevsky}.}
  \bibinfo{year}{2014}\natexlab{}.
\newblock \showarticletitle{One weird trick for parallelizing convolutional
  neural networks}.
\newblock \bibinfo{journal}{\emph{CoRR}}  \bibinfo{volume}{abs/1404.5997}
  (\bibinfo{year}{2014}).
\newblock
\showeprint[arXiv]{1404.5997}
\urldef\tempurl%
\url{http://arxiv.org/abs/1404.5997}
\showURL{%
\tempurl}


\bibitem[Kwon et~al\mbox{.}(2020)]%
        {nimble}
\bibfield{author}{\bibinfo{person}{Woosuk Kwon}, \bibinfo{person}{Gyeong-In
  Yu}, \bibinfo{person}{Eunji Jeong}, {and} \bibinfo{person}{Byung-Gon Chun}.}
  \bibinfo{year}{2020}\natexlab{}.
\newblock \showarticletitle{Nimble: Lightweight and Parallel GPU Task
  Scheduling for Deep Learning}. In \bibinfo{booktitle}{\emph{NeurIPS}}.
\newblock


\bibitem[Le(2013)]%
        {le2013building}
\bibfield{author}{\bibinfo{person}{Quoc~V Le}.}
  \bibinfo{year}{2013}\natexlab{}.
\newblock \showarticletitle{Building High-Level Features Using Large Scale
  Unsupervised Learning}. In \bibinfo{booktitle}{\emph{38th IEEE International
  Conference on Acoustics, Speech, and Signal Processing (ICASSP
  \textquotesingle13)}}. IEEE, \bibinfo{pages}{8595--8598}.
\newblock


\bibitem[Li et~al\mbox{.}(2015)]%
        {data_parallel_1}
\bibfield{author}{\bibinfo{person}{Hao Li}, \bibinfo{person}{Asim Kadav},
  \bibinfo{person}{Erik Kruus}, {and} \bibinfo{person}{Cristian Ungureanu}.}
  \bibinfo{year}{2015}\natexlab{}.
\newblock \showarticletitle{{MALT}: Distributed Data-Parallelism for Existing
  ML Applications}. In \bibinfo{booktitle}{\emph{10th European Conference on
  Computer Systems (EuroSys \textquotesingle15)}}. \bibinfo{publisher}{ACM},
  Article \bibinfo{articleno}{3}, \bibinfo{numpages}{16}~pages.
\newblock


\bibitem[Li et~al\mbox{.}(2020)]%
        {pydist}
\bibfield{author}{\bibinfo{person}{Shen Li}, \bibinfo{person}{Yanli Zhao},
  \bibinfo{person}{Rohan Varma}, \bibinfo{person}{Omkar Salpekar},
  \bibinfo{person}{Pieter Noordhuis}, \bibinfo{person}{Teng Li},
  \bibinfo{person}{Adam Paszke}, \bibinfo{person}{Jeff Smith},
  \bibinfo{person}{Brian Vaughan}, \bibinfo{person}{Pritam Damania}, {and}
  \bibinfo{person}{Soumith Chintala}.} \bibinfo{year}{2020}\natexlab{}.
\newblock \showarticletitle{PyTorch Distributed: Experiences on Accelerating
  Data Parallel Training}.
\newblock \bibinfo{journal}{\emph{CoRR}}  \bibinfo{volume}{abs/2006.15704}
  (\bibinfo{year}{2020}).
\newblock
\showeprint[arXiv]{2006.15704}
\urldef\tempurl%
\url{https://arxiv.org/abs/2006.15704}
\showURL{%
\tempurl}


\bibitem[Li et~al\mbox{.}(2021)]%
        {terapipe}
\bibfield{author}{\bibinfo{person}{Zhuohan Li}, \bibinfo{person}{Siyuan
  Zhuang}, \bibinfo{person}{Shiyuan Guo}, \bibinfo{person}{Danyang Zhuo},
  \bibinfo{person}{Hao Zhang}, \bibinfo{person}{Dawn Song}, {and}
  \bibinfo{person}{Ion Stoica}.} \bibinfo{year}{2021}\natexlab{}.
\newblock \showarticletitle{TeraPipe: Token-Level Pipeline Parallelism for
  Training Large-Scale Language Models}.
\newblock \bibinfo{journal}{\emph{CoRR}}  \bibinfo{volume}{abs/2102.07988}
  (\bibinfo{year}{2021}).
\newblock
\showeprint[arXiv]{2102.07988}
\urldef\tempurl%
\url{https://arxiv.org/abs/2102.07988}
\showURL{%
\tempurl}


\bibitem[Mahdavinejad et~al\mbox{.}(2018)]%
        {mahdavinejad2018machine}
\bibfield{author}{\bibinfo{person}{Mohammad~Saeid Mahdavinejad},
  \bibinfo{person}{Mohammadreza Rezvan}, \bibinfo{person}{Mohammadamin
  Barekatain}, \bibinfo{person}{Peyman Adibi}, \bibinfo{person}{Payam
  Barnaghi}, {and} \bibinfo{person}{Amit~P Sheth}.}
  \bibinfo{year}{2018}\natexlab{}.
\newblock \showarticletitle{Machine Learning for Internet of Things Data
  Analysis: A Survey}.
\newblock \bibinfo{journal}{\emph{Digital Communications and Networks}}
  \bibinfo{volume}{4}, \bibinfo{number}{3} (\bibinfo{year}{2018}),
  \bibinfo{pages}{161--175}.
\newblock


\bibitem[Microsoft(2022)]%
        {azure}
\bibfield{author}{\bibinfo{person}{Microsoft}.} \bibinfo{year}{[Accessed
  2-July-2022]}\natexlab{}.
\newblock \bibinfo{title}{{Microsoft Azure}}.
\newblock
\newblock
\urldef\tempurl%
\url{https://azure.microsoft.com/}
\showURL{%
\tempurl}


\bibitem[Mirhoseini et~al\mbox{.}(2018)]%
        {hierarchicalRL}
\bibfield{author}{\bibinfo{person}{Azalia Mirhoseini}, \bibinfo{person}{Anna
  Goldie}, \bibinfo{person}{Hieu Pham}, \bibinfo{person}{Benoit Steiner},
  \bibinfo{person}{Quoc~V Le}, {and} \bibinfo{person}{Jeff Dean}.}
  \bibinfo{year}{2018}\natexlab{}.
\newblock \showarticletitle{Hierarchical Planning for Device Placement}. In
  \bibinfo{booktitle}{\emph{6th International Conference on Learning
  Representations (ICLR \textquotesingle18)}}.
\newblock


\bibitem[Mirhoseini et~al\mbox{.}(2017)]%
        {colocRL}
\bibfield{author}{\bibinfo{person}{Azalia Mirhoseini}, \bibinfo{person}{Hieu
  Pham}, \bibinfo{person}{Quoc~V. Le}, \bibinfo{person}{Benoit Steiner},
  \bibinfo{person}{Rasmus Larsen}, \bibinfo{person}{Yuefeng Zhou},
  \bibinfo{person}{Naveen Kumar}, \bibinfo{person}{Mohammad Norouzi},
  \bibinfo{person}{Samy Bengio}, {and} \bibinfo{person}{Jeff Dean}.}
  \bibinfo{year}{2017}\natexlab{}.
\newblock \showarticletitle{Device Placement Optimization with Reinforcement
  Learning}. In \bibinfo{booktitle}{\emph{34th International Conference on
  Machine Learning (ICML \textquotesingle17)}}. \bibinfo{publisher}{PMLR},
  \bibinfo{pages}{2430--2439}.
\newblock


\bibitem[M{\"o}hring et~al\mbox{.}(1996)]%
        {mohring1996scheduling}
\bibfield{author}{\bibinfo{person}{Rolf~H. M{\"o}hring},
  \bibinfo{person}{Markus~W. Sch{\"a}ffter}, {and} \bibinfo{person}{Andreas~S.
  Schulz}.} \bibinfo{year}{1996}\natexlab{}.
\newblock \showarticletitle{Scheduling Jobs with Communication Delays: Using
  Infeasible Solutions for Approximation}. In \bibinfo{booktitle}{\emph{4th
  Annual European Symposium on Algorithms (ESA \textquotesingle96)}}. Springer,
  \bibinfo{pages}{76--90}.
\newblock


\bibitem[Munier and König(1997)]%
        {uet_uct_lp}
\bibfield{author}{\bibinfo{person}{Alix Munier} {and}
  \bibinfo{person}{Jean-Claude König}.} \bibinfo{year}{1997}\natexlab{}.
\newblock \showarticletitle{A Heuristic for a Scheduling Problem with
  Communication Delays}.
\newblock \bibinfo{journal}{\emph{Operations Research}} \bibinfo{volume}{45},
  \bibinfo{number}{1} (\bibinfo{year}{1997}), \bibinfo{pages}{145--147}.
\newblock


\bibitem[Nvidia(2022)]%
        {cudastreams}
\bibfield{author}{\bibinfo{person}{Nvidia}.} \bibinfo{year}{[Accessed
  2-July-2022]}\natexlab{}.
\newblock \bibinfo{title}{{CUDA Streams}}.
\newblock
\newblock
\urldef\tempurl%
\url{{https://docs.nvidia.com/cuda/cuda-c-programming-guide/}}
\showURL{%
\tempurl}


\bibitem[NVIDIA(2022)]%
        {nvlink}
\bibfield{author}{\bibinfo{person}{NVIDIA}.} \bibinfo{year}{[Accessed
  2-July-2022]}\natexlab{}.
\newblock \bibinfo{title}{{NVLink} and {NVSwitch}}.
\newblock
\newblock
\urldef\tempurl%
\url{https://www.nvidia.com/en-us/data-center/nvlink/}
\showURL{%
\tempurl}


\bibitem[Paszke et~al\mbox{.}(2019)]%
        {pytorch}
\bibfield{author}{\bibinfo{person}{Adam Paszke}, \bibinfo{person}{Sam Gross},
  \bibinfo{person}{Francisco Massa}, \bibinfo{person}{Adam Lerer},
  \bibinfo{person}{James Bradbury}, \bibinfo{person}{Gregory Chanan},
  \bibinfo{person}{Trevor Killeen}, \bibinfo{person}{Zeming Lin},
  \bibinfo{person}{Natalia Gimelshein}, \bibinfo{person}{Luca Antiga},
  \bibinfo{person}{Alban Desmaison}, \bibinfo{person}{Andreas Kopf},
  \bibinfo{person}{Edward Yang}, \bibinfo{person}{Zachary DeVito},
  \bibinfo{person}{Martin Raison}, \bibinfo{person}{Alykhan Tejani},
  \bibinfo{person}{Sasank Chilamkurthy}, \bibinfo{person}{Benoit Steiner},
  \bibinfo{person}{Lu Fang}, \bibinfo{person}{Junjie Bai}, {and}
  \bibinfo{person}{Soumith Chintala}.} \bibinfo{year}{2019}\natexlab{}.
\newblock \showarticletitle{{PyTorch}: An Imperative Style, High-Performance
  Deep Learning Library}.
\newblock In \bibinfo{booktitle}{\emph{33rd Conference on Neural Information
  Processing Systems (NeurIPS \textquotesingle19)}}. \bibinfo{publisher}{Curran
  Associates, Inc.}, \bibinfo{pages}{8024--8035}.
\newblock


\bibitem[PyTorch(2022a)]%
        {cudaevent1}
\bibfield{author}{\bibinfo{person}{PyTorch}.} \bibinfo{year}{[Accessed
  2-July-2022]}\natexlab{a}.
\newblock \bibinfo{title}{{CUDA Events}}.
\newblock
\newblock
\urldef\tempurl%
\url{{https://pytorch.org/docs/stable/generated/torch.cuda.Event.html}}
\showURL{%
\tempurl}


\bibitem[PyTorch(2022b)]%
        {pyhook}
\bibfield{author}{\bibinfo{person}{PyTorch}.} \bibinfo{year}{[Accessed
  2-July-2022]}\natexlab{b}.
\newblock \bibinfo{title}{{\pytorch Hooks}}.
\newblock
\newblock
\urldef\tempurl%
\url{{https://pytorch.org/tutorials/beginner/former_torchies/nnft_tutorial.html?}}
\showURL{%
\tempurl}


\bibitem[PyTorch(2022c)]%
        {pylstm1}
\bibfield{author}{\bibinfo{person}{PyTorch}.} \bibinfo{year}{[Accessed
  2-July-2022]}\natexlab{c}.
\newblock \bibinfo{title}{{\pytorch LSTM}}.
\newblock
\newblock
\urldef\tempurl%
\url{{https://pytorch.org/docs/stable/generated/torch.nn.LSTM.html}}
\showURL{%
\tempurl}


\bibitem[Rajbhandari et~al\mbox{.}(2020)]%
        {zero}
\bibfield{author}{\bibinfo{person}{Samyam Rajbhandari}, \bibinfo{person}{Jeff
  Rasley}, \bibinfo{person}{Olatunji Ruwase}, {and} \bibinfo{person}{Yuxiong
  He}.} \bibinfo{year}{2020}\natexlab{}.
\newblock \showarticletitle{ZeRO: Memory Optimizations toward Training Trillion
  Parameter Models}. In \bibinfo{booktitle}{\emph{Proceedings of the
  International Conference for High Performance Computing, Networking, Storage
  and Analysis}} (Atlanta, Georgia) \emph{(\bibinfo{series}{SC '20})}.
  \bibinfo{publisher}{IEEE Press}, Article \bibinfo{articleno}{20},
  \bibinfo{numpages}{16}~pages.
\newblock
\showISBNx{9781728199986}


\bibitem[Schult(2008)]%
        {networkx}
\bibfield{author}{\bibinfo{person}{Daniel~A. Schult}.}
  \bibinfo{year}{2008}\natexlab{}.
\newblock \showarticletitle{Exploring Network Structure, Dynamics, and Function
  Using {NetworkX}}. In \bibinfo{booktitle}{\emph{7th Python in Science
  Conference (SciPy \textquotesingle08)}}. \bibinfo{pages}{11--15}.
\newblock


\bibitem[Seide and Agarwal(2016)]%
        {seide2016cntk}
\bibfield{author}{\bibinfo{person}{Frank Seide} {and} \bibinfo{person}{Amit
  Agarwal}.} \bibinfo{year}{2016}\natexlab{}.
\newblock \showarticletitle{{CNTK}: {Microsoft}\textquotesingle{s} Open-Source
  Deep-Learning Toolkit}. In \bibinfo{booktitle}{\emph{22nd ACM SIGKDD
  International Conference on Knowledge Discovery and Data Mining (KDD
  \textquotesingle16)}}. ACM, \bibinfo{pages}{2135--2135}.
\newblock


\bibitem[Shoeybi et~al\mbox{.}(2019)]%
        {megatron}
\bibfield{author}{\bibinfo{person}{Mohammad Shoeybi}, \bibinfo{person}{Mostofa
  Patwary}, \bibinfo{person}{Raul Puri}, \bibinfo{person}{Patrick LeGresley},
  \bibinfo{person}{Jared Casper}, {and} \bibinfo{person}{Bryan Catanzaro}.}
  \bibinfo{year}{2019}\natexlab{}.
\newblock \showarticletitle{Megatron-LM: Training Multi-Billion Parameter
  Language Models Using Model Parallelism}.
\newblock \bibinfo{journal}{\emph{CoRR}}  \bibinfo{volume}{abs/1909.08053}
  (\bibinfo{year}{2019}).
\newblock
\showeprint[arXiv]{1909.08053}
\urldef\tempurl%
\url{http://arxiv.org/abs/1909.08053}
\showURL{%
\tempurl}


\bibitem[Sparks et~al\mbox{.}(2013)]%
        {distributed_ml_2}
\bibfield{author}{\bibinfo{person}{E.~R. Sparks}, \bibinfo{person}{A.
  Talwalkar}, \bibinfo{person}{V. Smith}, \bibinfo{person}{J. Kottalam},
  \bibinfo{person}{X. Pan}, \bibinfo{person}{J. Gonzalez},
  \bibinfo{person}{M.~J. Franklin}, \bibinfo{person}{M.~I. Jordan}, {and}
  \bibinfo{person}{T. Kraska}.} \bibinfo{year}{2013}\natexlab{}.
\newblock \showarticletitle{{MLI}: An {API} for Distributed Machine Learning}.
  In \bibinfo{booktitle}{\emph{13th IEEE International Conference on Data
  Mining (ICDM \textquotesingle13)}}. \bibinfo{publisher}{IEEE},
  \bibinfo{pages}{1187--1192}.
\newblock


\bibitem[Sujeeth et~al\mbox{.}(2011)]%
        {data_parallel_3}
\bibfield{author}{\bibinfo{person}{Arvind Sujeeth}, \bibinfo{person}{HyoukJoong
  Lee}, \bibinfo{person}{Kevin Brown}, \bibinfo{person}{Tiark Rompf},
  \bibinfo{person}{Hassan Chafi}, \bibinfo{person}{Michael Wu},
  \bibinfo{person}{Anand Atreya}, \bibinfo{person}{Martin Odersky}, {and}
  \bibinfo{person}{Kunle Olukotun}.} \bibinfo{year}{2011}\natexlab{}.
\newblock \showarticletitle{{OptiML}: An Implicitly Parallel Domain-Specific
  Language for Machine Learning}. In \bibinfo{booktitle}{\emph{28th
  International Conference on Machine Learning (ICML \textquotesingle11)}}.
  \bibinfo{publisher}{PMLR}, \bibinfo{pages}{609--616}.
\newblock


\bibitem[Szegedy et~al\mbox{.}(2016)]%
        {inceptionv3}
\bibfield{author}{\bibinfo{person}{Christian Szegedy}, \bibinfo{person}{Vincent
  Vanhoucke}, \bibinfo{person}{Sergey Ioffe}, \bibinfo{person}{Jonathon
  Shlens}, {and} \bibinfo{person}{Zbigniew Wojna}.}
  \bibinfo{year}{2016}\natexlab{}.
\newblock \showarticletitle{Rethinking the {Inception} Architecture for
  Computer Vision}. In \bibinfo{booktitle}{\emph{29th IEEE Conference on
  Computer Vision and Pattern Recognition (CVPR \textquotesingle16)}}.
  \bibinfo{publisher}{IEEE}, \bibinfo{pages}{2818--2826}.
\newblock


\bibitem[Tarnawski et~al\mbox{.}(2020)]%
        {amarpaper}
\bibfield{author}{\bibinfo{person}{Jakub Tarnawski}, \bibinfo{person}{Amar
  Phanishayee}, \bibinfo{person}{Nikhil Devanur}, \bibinfo{person}{Divya
  Mahajan}, {and} \bibinfo{person}{Fanny~Nina Paravecino}.}
  \bibinfo{year}{2020}\natexlab{}.
\newblock \showarticletitle{Efficient Algorithms for Device Placement of DNN
  Graph Operators}. In \bibinfo{booktitle}{\emph{Proceedings of the 34th
  International Conference on Neural Information Processing Systems}}
  (Vancouver, BC, Canada) \emph{(\bibinfo{series}{NIPS'20})}.
  \bibinfo{publisher}{Curran Associates Inc.}, \bibinfo{address}{Red Hook, NY,
  USA}, Article \bibinfo{articleno}{1296}, \bibinfo{numpages}{13}~pages.
\newblock
\showISBNx{9781713829546}


\bibitem[Tensorflow(2022)]%
        {rnode}
\bibfield{author}{\bibinfo{person}{Tensorflow}.} \bibinfo{year}{[Accessed
  2-July-2022]}\natexlab{}.
\newblock \bibinfo{title}{{Tensorflow Rendezvous Node}}.
\newblock
\newblock
\urldef\tempurl%
\url{{https://github.com/tensorflow/tensorflow/blob/41285cf7a11fa3a2c2ead6b6e9adcec4232b18ad/tensorflow/core/framework/rendezvous.h}}
\showURL{%
\tempurl}


\bibitem[{TensorFlow Community}(2022a)]%
        {tf_graph_optmization}
\bibfield{author}{\bibinfo{person}{{TensorFlow Community}}.}
  \bibinfo{year}{[Accessed 2-July-2022]}\natexlab{a}.
\newblock \bibinfo{title}{{TensorFlow} Graph Optimization with {Grappler}}.
\newblock
\newblock
\urldef\tempurl%
\url{https://www.tensorflow.org/guide/graph_optimization}
\showURL{%
\tempurl}


\bibitem[{TensorFlow Community}(2022b)]%
        {tf_xla}
\bibfield{author}{\bibinfo{person}{{TensorFlow Community}}.}
  \bibinfo{year}{[Accessed 2-July-2022]}\natexlab{b}.
\newblock \bibinfo{title}{{XLA}: Optimizing Compiler for Machine Learning}.
\newblock
\newblock
\urldef\tempurl%
\url{https://www.tensorflow.org/xla}
\showURL{%
\tempurl}


\bibitem[{Theano Development Team}(2016)]%
        {theano2016}
\bibfield{author}{\bibinfo{person}{{Theano Development Team}}.}
  \bibinfo{year}{2016}\natexlab{}.
\newblock \showarticletitle{Theano: {A} Python framework for fast computation
  of mathematical expressions}.
\newblock \bibinfo{journal}{\emph{CoRR}}  \bibinfo{volume}{abs/1605.02688}
  (\bibinfo{year}{2016}).
\newblock
\showeprint[arXiv]{1605.02688}
\urldef\tempurl%
\url{http://arxiv.org/abs/1605.02688}
\showURL{%
\tempurl}


\bibitem[Tomlin(1989)]%
        {tomlin1989note}
\bibfield{author}{\bibinfo{person}{John~A. Tomlin}.}
  \bibinfo{year}{1989}\natexlab{}.
\newblock \showarticletitle{A Note on Comparing Simplex and Interior Methods
  for Linear Programming}.
\newblock In \bibinfo{booktitle}{\emph{Progress in Mathematical Programming}}.
  \bibinfo{publisher}{Springer}, \bibinfo{pages}{91--103}.
\newblock


\bibitem[Vaswani et~al\mbox{.}(2017)]%
        {transformer}
\bibfield{author}{\bibinfo{person}{Ashish Vaswani}, \bibinfo{person}{Noam
  Shazeer}, \bibinfo{person}{Niki Parmar}, \bibinfo{person}{Jakob Uszkoreit},
  \bibinfo{person}{Llion Jones}, \bibinfo{person}{Aidan~N. Gomez},
  \bibinfo{person}{Lukasz Kaiser}, {and} \bibinfo{person}{Illia Polosukhin}.}
  \bibinfo{year}{2017}\natexlab{}.
\newblock \showarticletitle{Attention Is All You Need}.
\newblock \bibinfo{journal}{\emph{CoRR}}  \bibinfo{volume}{abs/1706.03762}
  (\bibinfo{year}{2017}).
\newblock
\showeprint[arXiv]{1706.03762}
\urldef\tempurl%
\url{http://arxiv.org/abs/1706.03762}
\showURL{%
\tempurl}


\bibitem[Veltman et~al\mbox{.}(1990)]%
        {algo_lit_review}
\bibfield{author}{\bibinfo{person}{Bart Veltman}, \bibinfo{person}{B.~J.
  Lageweg}, {and} \bibinfo{person}{Jan~K. Lenstra}.}
  \bibinfo{year}{1990}\natexlab{}.
\newblock \showarticletitle{Multiprocessor Scheduling with Communication
  Delays}.
\newblock \bibinfo{journal}{\emph{Parallel computing}} \bibinfo{volume}{16},
  \bibinfo{number}{2-3} (\bibinfo{year}{1990}), \bibinfo{pages}{173--182}.
\newblock


\bibitem[Wang et~al\mbox{.}(2019a)]%
        {wang2019tofu}
\bibfield{author}{\bibinfo{person}{Minjie Wang}, \bibinfo{person}{Chien-chin
  Huang}, {and} \bibinfo{person}{Jinyang Li}.}
  \bibinfo{year}{2019}\natexlab{a}.
\newblock \showarticletitle{Supporting Very Large Models Using Automatic
  Dataflow Graph Partitioning}. In \bibinfo{booktitle}{\emph{14th European
  Conference on Computer Systems (EuroSys \textquotesingle19)}}.
  \bibinfo{publisher}{ACM}, Article \bibinfo{articleno}{26},
  \bibinfo{numpages}{17}~pages.
\newblock
\showISBNx{9781450362818}


\bibitem[Wang et~al\mbox{.}(2019b)]%
        {alibaba2019}
\bibfield{author}{\bibinfo{person}{Mengdi Wang}, \bibinfo{person}{Chen Meng},
  \bibinfo{person}{Guoping Long}, \bibinfo{person}{Chuan Wu},
  \bibinfo{person}{Jun Yang}, \bibinfo{person}{Wei Lin}, {and}
  \bibinfo{person}{Yangqing Jia}.} \bibinfo{year}{2019}\natexlab{b}.
\newblock \showarticletitle{Characterizing Deep Learning Training Workloads on
  {Alibaba-PAI}}. In \bibinfo{booktitle}{\emph{22nd {IEEE} International
  Symposium on Workload Characterization (IISWC \textquotesingle19)}}.
  \bibinfo{publisher}{{IEEE}}, \bibinfo{pages}{189--202}.
\newblock


\bibitem[Wu et~al\mbox{.}(2016)]%
        {gnmt}
\bibfield{author}{\bibinfo{person}{Yonghui Wu}, \bibinfo{person}{Mike
  Schuster}, \bibinfo{person}{Zhifeng Chen}, \bibinfo{person}{Quoc~V. Le},
  \bibinfo{person}{Mohammad Norouzi}, \bibinfo{person}{Wolfgang Macherey},
  \bibinfo{person}{Maxim Krikun}, \bibinfo{person}{Yuan Cao},
  \bibinfo{person}{Qin Gao}, \bibinfo{person}{Klaus Macherey},
  \bibinfo{person}{Jeff Klingner}, \bibinfo{person}{Apurva Shah},
  \bibinfo{person}{Melvin Johnson}, \bibinfo{person}{Xiaobing Liu},
  \bibinfo{person}{Lukasz Kaiser}, \bibinfo{person}{Stephan Gouws},
  \bibinfo{person}{Yoshikiyo Kato}, \bibinfo{person}{Taku Kudo},
  \bibinfo{person}{Hideto Kazawa}, \bibinfo{person}{Keith Stevens},
  \bibinfo{person}{George Kurian}, \bibinfo{person}{Nishant Patil},
  \bibinfo{person}{Wei Wang}, \bibinfo{person}{Cliff Young},
  \bibinfo{person}{Jason Smith}, \bibinfo{person}{Jason Riesa},
  \bibinfo{person}{Alex Rudnick}, \bibinfo{person}{Oriol Vinyals},
  \bibinfo{person}{Greg Corrado}, \bibinfo{person}{Macduff Hughes}, {and}
  \bibinfo{person}{Jeffrey Dean}.} \bibinfo{year}{2016}\natexlab{}.
\newblock \showarticletitle{Google's Neural Machine Translation System:
  Bridging the Gap between Human and Machine Translation}.
\newblock \bibinfo{journal}{\emph{CoRR}}  \bibinfo{volume}{abs/1609.08144}
  (\bibinfo{year}{2016}).
\newblock
\showeprint[arXiv]{1609.08144}
\urldef\tempurl%
\url{http://arxiv.org/abs/1609.08144}
\showURL{%
\tempurl}


\bibitem[Xing et~al\mbox{.}(2015)]%
        {distributed_ml_4}
\bibfield{author}{\bibinfo{person}{Eric~P. Xing}, \bibinfo{person}{Qirong Ho},
  \bibinfo{person}{Wei Dai}, \bibinfo{person}{Jin-Kyu Kim},
  \bibinfo{person}{Jinliang Wei}, \bibinfo{person}{Seunghak Lee},
  \bibinfo{person}{Xun Zheng}, \bibinfo{person}{Pengtao Xie},
  \bibinfo{person}{Abhimanu Kumar}, {and} \bibinfo{person}{Yaoliang Yu}.}
  \bibinfo{year}{2015}\natexlab{}.
\newblock \showarticletitle{{Petuum}: A New Platform for Distributed Machine
  Learning on Big Data}. In \bibinfo{booktitle}{\emph{21th ACM SIGKDD
  International Conference on Knowledge Discovery and Data Mining (KDD
  \textquotesingle15)}}. \bibinfo{publisher}{ACM},
  \bibinfo{pages}{1335–1344}.
\newblock
\showISBNx{9781450336642}


\bibitem[Yang et~al\mbox{.}(2021)]%
        {pipemare}
\bibfield{author}{\bibinfo{person}{Bowen Yang}, \bibinfo{person}{Jian Zhang},
  \bibinfo{person}{Jonathan Li}, \bibinfo{person}{Christopher Re},
  \bibinfo{person}{Christopher Aberger}, {and} \bibinfo{person}{Christopher
  De~Sa}.} \bibinfo{year}{2021}\natexlab{}.
\newblock \showarticletitle{PipeMare: Asynchronous Pipeline Parallel DNN
  Training}. In \bibinfo{booktitle}{\emph{Proceedings of Machine Learning and
  Systems}}, \bibfield{editor}{\bibinfo{person}{A.~Smola},
  \bibinfo{person}{A.~Dimakis}, {and} \bibinfo{person}{I.~Stoica}} (Eds.),
  Vol.~\bibinfo{volume}{3}. \bibinfo{pages}{269--296}.
\newblock
\urldef\tempurl%
\url{https://proceedings.mlsys.org/paper/2021/file/6c8349cc7260ae62e3b1396831a8398f-Paper.pdf}
\showURL{%
\tempurl}


\bibitem[Yang and Gerasoulis(1994)]%
        {yang1994dsc}
\bibfield{author}{\bibinfo{person}{Tao Yang} {and} \bibinfo{person}{Apostolos
  Gerasoulis}.} \bibinfo{year}{1994}\natexlab{}.
\newblock \showarticletitle{{DSC}: Scheduling Parallel Tasks on an Unbounded
  Number of Processors}.
\newblock \bibinfo{journal}{\emph{IEEE Transactions on Parallel and Distributed
  Systems}} \bibinfo{volume}{5}, \bibinfo{number}{9} (\bibinfo{year}{1994}),
  \bibinfo{pages}{951--967}.
\newblock


\bibitem[Zeydan et~al\mbox{.}(2016)]%
        {zeydan2016big}
\bibfield{author}{\bibinfo{person}{Engin Zeydan}, \bibinfo{person}{Ejder
  Bastug}, \bibinfo{person}{Mehdi Bennis}, \bibinfo{person}{Manhal~Abdel
  Kader}, \bibinfo{person}{Ilyas~Alper Karatepe}, \bibinfo{person}{Ahmet~Salih
  Er}, {and} \bibinfo{person}{M{\'e}rouane Debbah}.}
  \bibinfo{year}{2016}\natexlab{}.
\newblock \showarticletitle{Big data Caching for Networking: Moving from Cloud
  to Edge}.
\newblock \bibinfo{journal}{\emph{IEEE Communications Magazine}}
  \bibinfo{volume}{54}, \bibinfo{number}{9} (\bibinfo{year}{2016}),
  \bibinfo{pages}{36--42}.
\newblock


\bibitem[Zheng et~al\mbox{.}(2022)]%
        {alpa1}
\bibfield{author}{\bibinfo{person}{Lianmin Zheng}, \bibinfo{person}{Zhuohan
  Li}, \bibinfo{person}{Hao Zhang}, \bibinfo{person}{Yonghao Zhuang},
  \bibinfo{person}{Zhifeng Chen}, \bibinfo{person}{Yanping Huang},
  \bibinfo{person}{Yida Wang}, \bibinfo{person}{Yuanzhong Xu},
  \bibinfo{person}{Danyang Zhuo}, \bibinfo{person}{Joseph~E. Gonzalez}, {and}
  \bibinfo{person}{Ion Stoica}.} \bibinfo{year}{2022}\natexlab{}.
\newblock \showarticletitle{Alpa: Automating Inter- and Intra-Operator
  Parallelism for Distributed Deep Learning}.
\newblock \bibinfo{journal}{\emph{CoRR}}  \bibinfo{volume}{abs/2201.12023}
  (\bibinfo{year}{2022}).
\newblock
\showeprint[arXiv]{2201.12023}
\urldef\tempurl%
\url{https://arxiv.org/abs/2201.12023}
\showURL{%
\tempurl}


\end{thebibliography}

\appendix
\section{Optimality Analysis for m-ETF}
\label{appendix:etf}

We now derive an upper bound on the makespan of graph $G$ running on $n$ memory-constrained devices according to a placement generated by m-ETF. A bound for ETF (with no memory constraint on the devices) was obtained in \cite{etf}. We extend it to the case where the devices have finite memory. For any given $p$ devices, let $\omega_{\textrm{m-etf}}^{p}$ be the m-ETF makespan and $\omega_{\textrm{opt}}^{p}$ be the optimal makespan achieved using devices with infinte memory and zero device-to-device communication costs. 

\begin{theorem}
The makespan of m-ETF, $\omega_{\textrm{m-etf}}^{n}$ is at most $(2+\rho)\omega_{\textrm{opt}}^{R}$, where $R$ is an integer $<n$ (computed in equation \ref{eqn:Rcompute}).
\end{theorem}
\begin{proof}
Let $K = \displaystyle \frac{n\cdot M}{\sum_{i=1}^{m} d_i}$, where for any {operator (task)} $i$ in $G$, $d_i$ is the size of memory required by $i$. Intuitively, $K$ is the ratio of the total memory available from all devices to the total memory required by the model. Thus $K > 1$, and for practical purposes we can assume that $K$ is sufficiently larger than 1. At each step, m-ETF greedily matches a ready task to an available device. Specifically, a device is said to be {\it available} if there is neither a task currently running nor has been scheduled to run on that device. A task is said to be {\it ready} if all of its predecessors have completed. Let $I$ and $A$ be the set of available devices and ready tasks at a given step respectively. 
When a task completes, $I$ is updated to include the recently free device and all of the task's children that are ready are added to $A$.

At each such step, a device $d$ is said to be {\it memory-sufficient (MS)} if the remaining free memory on $d$ is greater than the memory requirement of {\it each} task in $A$. If $d$ has insufficient memory for even a single task in $A$, we say $d$ is not MS thereafter. It is removed from $I$ and is not considered for any further placement. 

The time $(0, \omega_{\textrm{m-etf}})$ can be partitioned into two distinct sets. Set $\mathcal{A}$ containing the time-periods when all the MS devices (in that time-period) are busy and set $\mathcal{B}$ when at least one MS is idle. Suppose $\mathcal{B}$ is the disjoint union of intervals $(b_{li}, b_{ri})$ i.e,
$$
\mathcal{B}=(b_{l1}, b_{r1}) \cup (b_{l2}, b_{r2}) \cup \dots \cup (b_{lq}, b_{rq})
$$
where $b_{l1}<b_{r1}<b_{l2}<b_{r2}\dots<b_{lq}<b_{rq}$.

\begin{lemma*}[Theorem 3.2 in \cite{etf}]
We can find a chain of tasks,
$$X:T_{l}\rightarrow T_{l-1}\rightarrow \dots \rightarrow T_{1}$$
such that
$$\sum_{i=1}^{q} (b_{r1}-b_{l1}) \leq \sum_{j=1}^{l}t(T_{j}) + \sum_{j=1}^{l-1}c_{j(j+1)}$$
\label{lemma:B_chain}
\end{lemma*}
That is, the total time period of $\mathcal{B}$ will be covered by computation and communication times along the chain $X$. We will denote $\sum_{j=1}^{l-1}c_{j(j+1)}$ by $C_{X}$. For proof, we refer the reader to Theorem 3.2 in \cite{etf}.

Let $r$ be the number of devices that remain MS until the end of m-ETF (i.e at time $\omega_{\textrm{m-etf}}$). With $K>1$, we will have $r>=1$. Let $\hat{\mathcal{B}}$ be the set of all time-periods when atleast one of these $r$ devices is idle. Note that $\hat{\mathcal{B}} \subseteq \mathcal{B}$, thus the chain $X$ from \ref{lemma:B_chain} will cover $\hat{\mathcal{B}}$ as well. Thus we have:
\begin{equation}
    \sum_{r}t(\phi_{i}) \; \leq \; r\times\sum_{ T_{j} \in X }t(T_{j}) \: + \: r\times C_{X}
    \label{eqn:a}
\end{equation}
where $\phi_{i}$ is the set of times when the device $d_{i}$ is idle in $(0, \omega_{\textrm{m-etf}})$ and $t(T)$ is computation time of task $T$.

Let $\omega_{\textrm{opt}}^{p}$ be the optimal makespan on $p$ devices with no memory limits and no communication costs. %\lindac{OPT should still include communication costs right? Just not memory limits}
Since the makespan on any number of devices is at least %\xm{typo} 
as large as a chain in the graph, we have
\begin{equation}
    \omega_{\textrm{opt}}^{r} \; \geq \; \sum_{X}t(T_{j}) \dots (i), \qquad \omega_{\textrm{opt}}^{n} \; \geq \;\sum_{X}t(T_{j}) \dots (ii)
    \label{eqn:b}
\end{equation}

Also, the net computation time of $G$ can be bounded as:
\begin{equation}
    \sum_{T_{j} \in G}t(T_{j}) \; \leq \; r \times  \omega_{\textrm{opt}}^{r} \dots (i), \qquad \sum_{T_{j} \in G}t(T_{j}) \; \leq \; n \times  \omega_{\textrm{opt}}^{n} \dots (ii)
    \label{eqn:e}
\end{equation}

Now we bound the m-ETF makespan. Consider the $r$ devices, their idle time and the jobs running on them:

\begin{align}
    \omega_{\textrm{m-etf}}^{n} \; &= \; \frac{1}{r} \left( \sum_{r}t(T_{j}) +  \sum_{r}t(\phi_{i})  \right) \\
    &\leq \;  \frac{1}{r} \left( \sum_{G}t(T_{j}) +  \sum_{r}t(\phi_{i}) \right)
    \label{eqn:f}
\end{align}

Using \ref{eqn:a}, \ref{eqn:b}(i) and \ref{eqn:e}(i), %\xm{there is no  \ref{eqn:f}(i)}

\begin{align*}
    \omega_{\textrm{m-etf}}^{n} \; &\leq \;  \frac{1}{r} \left( 2r\times \omega_{\textrm{opt}}^{r} + r\times C_{X} \right)
    = \; 2\omega_{\textrm{opt}}^{r} + C_{X}
\end{align*}
    
Similarly, using \ref{eqn:a}, \ref{eqn:b}(ii) and \ref{eqn:e}(ii), %\xm{there is no  \ref{eqn:f}(ii)}

\begin{align*}
     \omega_{\textrm{m-etf}}^{n} \; &\leq \;  \frac{1}{r} \left( (n+r)\times \omega_{\textrm{opt}}^{n} + r\times C_{X} \right)
      = \; \left( \frac{n+r}{r}\right)\omega_{\textrm{opt}}^{n} + C_{X}
\end{align*}

Note that $r$ will vary depending on the exact topological order considered for the m-ETF. So we define $R$ as the minimum $r$ across all possible topological ordering of the graph. Thus we have,

\begin{equation}
\omega_{\textrm{m-etf}}^{n} \; \leq \; \min\left( 2\omega_{\textrm{opt}}^{R} , \;  \left( \frac{n+R}{R}\right)\omega_{\textrm{opt}}^{n} \right) + C_{X}
\label{eqn:final_result}
\end{equation}

Further, with $\rho$ as the ratio between maximum communication ttime and minimum computation time, we have $C_{X} \leq \rho \omega_{\textrm{opt}}^{n} $ (also $C_{X} \leq \rho \omega_{\textrm{opt}}^{R} $ )

Thus we have,
\begin{equation}
\boxed{\omega_{\textrm{m-etf}}^{n} \; \leq \;\left( 1+ \frac{n}{R} + \rho\right)\omega_{\textrm{opt}}^{n} }
\label{eqn:final1}
\end{equation}

Alternatively, using the bound with $\omega_{\textrm{opt}}^{R}$,
\begin{equation}
\boxed{\omega_{\textrm{m-etf}}^{n} \; \leq \; (2+\rho)\omega_{\textrm{opt}}^{R}}
\label{eqn:final2}
\end{equation}

Here we note that bound in Equation \ref{eqn:final2} is of the same form as the original bound on ETF given in \cite{etf}, where $R$ replaces $n$. Thus the makespan of m-ETF, like ETF,  is within a constant factor of the optimal

Finally, $R$ can be computed as follows. Let the largest memory %\xm{typo}
requirement of any task in $G$ be $J \times M$. Since the devices become non-MS %\xm{MS?} 
when they can not place any of the available task in A 
, a memory of only $(1-J)M$ is use-able at each device in the worst case. Thus by greedily filling in the tasks onto devices, we get:

\begin{equation*}
    R \; \geq \; n - \left( \frac{\sum_{i=1}^{m} d_i}{(1-J)M} \right) \;=\; n - \frac{n}{(1-J)K}
\end{equation*}

Rounding it up, we have, 
\begin{equation}
    R \; =  \; \left\lceil n \left( 1 - \frac{1}{(1-J)K} \right) \right\rceil
    \label{eqn:Rcompute}
\end{equation}

\end{proof}

\section{Optimality Analysis of m-SCT} \label{appendix:sct}

We now formally prove that m-SCT's approximation ratio {to optimal} is an additive constant away from SCT's approximation ratio. Since SCT itself was known to be within a constant factor of optimal~\cite{sct}, our result means that m-SCT is also within a constant factor of optimal. Recall that we assume $\rho$ - the ration of maximum communication time to minimum computation time (defined in Table: \ref{fig:terminology}) - is less than 1. 

We will use similar notation to our analysis for m-ETF. Let $K = \displaystyle \frac{n\cdot M}{\sum_{i=1}^{m} d_i}$, where for any {operator (task)} %node
$i$ in $G$, $d_i$ is the size of memory required by $i$.  
We will define $J$ as the ratio between largest memory requirement from a single task and $M$. Formally, $J = \max_{i \in [m]} \frac{d_i}{M}$. 

Let $s_i$ be the start time of task $i$ in m-SCT, and $s_i^\infty$ be the start time of task i in the infinite device SCT algorithm. 
Let $u_j$ be the time where a task $j$ becomes urgent, which is exactly the earliest time when task $j$ can start on any device. Formally, $u_j = max_{i \rightarrow j \in E(G)} s_i + p_i + c_{ij}$. 

Similar to m-ETF, we will say a device $d$ is memory sufficient (abbreviated as MS) at time $T$ if and only if remaining free memory on $d$ is greater than the memory requirement of each task in $A$. Finally, we will use $r$ to denote the number of devices that are MS throughout the scheduling process. 

We will now analyse the approximation ratio of m-SCT with three steps. First we will show that not many tasks are impacted by devices going out of memory. Next we will show that any MS device must be idle only for a limited time. Our proof for this step follows a similar outline to the proof of Theorem 3 in~\cite{sct}, but our proof is significantly shorter due to a condensed case analysis. Finally, we will bound the makespace of m-SCT by summing up the idle and busy time on MS devices.

\begin{lemma} \label{claim:sctScheduleDifficulty}
There are at most $n- r$ task pairs $(i,j)$ such that $j$ is $i$'s favourite child, however when $j$ is scheduled, the device $d$ where $i$ is scheduled on does not have sufficient memory for task $j$. 
\end{lemma}
\begin{proof}
    Since $i$ is scheduled on device $d$, $d$ must be memory sufficient when $i$ is scheduled, but is no longer memory sufficient sometime after $i$ is scheduled (since $d$ does not have sufficient memory for task $j$). Since there are in total $n$ devices and $r$ devices are always memory sufficient throughout m-SCT, there must only be $n-r$ events where a device transition from being memory sufficient to not memory sufficient. 
\end{proof}
\begin{lemma}[Variant of Lemma 6 in~\cite{sct}] \label{lem:freetime}
Given two time units $s' \leq s$ such that $s - s' \leq c_{max}$, let $i$ be a task such that $s_i \leq s' \leq s_i + k_i + c_{max}$, then any busy or awake device at $s'$ is free for at most $\max(s' - s_i, c_{max})$ time during $[s_i, s]$.

\end{lemma}
\begin{proof}
    \begin{enumerate}
    \item 
    If a device $d$ is busy at $s'$. Let task $a$ be the task that is running at time $s'$. There are two possibilities: 
    \begin{itemize}
        \item 
        If task $a$ started before $s_i$, then device $d$ is busy for at least $s' - s_i$ time, thus free for at most $s - s' \leq c_{max}$ time.
        \item 
        On the other hand, if task $a$ started at some time $s^*$ where $s_i \leq s^* \leq s'$, then either task $a$ is still being executed at $s$, or task $a$ has completed at $s$. In the first case device $d$ is free for at most $s^* - s_i \leq s' - s_i$ time. In the second case device $d$ is busy for at least $k_a \geq c_{max} \geq s - s'$ time, and thus is free for at most $s - s_i - (s - s') = s' - s_i$ time.  
    \end{itemize}
    We conclude that device $d$ must be busy for at least $\min\{s' - s_i, s - s'\}$ time during $[s_i, s]$. 
    \item 
    If a device is awake at $s'$, let $a$ be the last task on $d$ before $s'$ and let $s^*$ be when $a$ finishes. Then we know by the nature of our algorithm that some task $b$ will start on $d$ no later than $s^* + c_{max}$. Since $s - s' \leq c_{max}$, we know that $b$ is not yet finished at time $s$. Therefore either task $a$ starts after $s_i$ (which means the device is busy for at least $k_a \geq c_{max}$ time and free for at most $s - s_i - c_{max} \leq s' - s_i$ time), or the device is vacant for at most $c_{max}$ time. 
    \end{enumerate}
\end{proof}
\begin{lemma} \label{claim:sctIdleTime}
Assume that task $j$'s favourite parent $i^*$'s device is MS during time period $[s_i, s_j]$. Then there exists a predecessor $i$ of $j$ such that the total amount of idle time during $[s_i, s_j]$ on any device $d$ that is MS throughout the period is at most $s_j^\infty - s_i^\infty$. 
\end{lemma}
\begin{proof}
Note that since task $j$'s favourite predecessor $i^*$'s device is MS during $[s_i, s_j]$, it is possible to schedule $j$ on the same device as $i^*$. This fact will be used in the case analysis. 

Let $i$ be a predecessor of $j$ such that $s_i + k_i + c_{ij}$ is maximized (namely, $u_j = s_i + k_i + c_{ij}$). Notice also that after $j$ becomes urgent at $u_j$ and before $j$ is scheduled, all memory sufficient devices must be busy (otherwise $j$ would have been scheduled on a device). Hence for any $T < s_j$, the total vacant time for an MS device during $[T, s_j]$ is equal to its total vacant time during $[T, u_j]$. Now we will discuss three different scenarios and prove that in each scenario, an MS device is vacant for at most $s_j^\infty - s_i^\infty$ time during $[s_i, s_j]$. 
\begin{enumerate}
\item 
When $j$ is not the favorite child of $i$, we know that in the infinite device SCT algorithm, $i$ and $j$ are scheduled on different devices. Hence $s_j^\infty \geq s_i^\infty + k_i + c_{ij}$. On the other hand, in m-SCT, after $j$ becomes urgent ($j$ becomes urgent at time $s_i + k_i + c_{ij}$) and before $j$ is scheduled, any MS device must be busy. Therefore the amount of vacant time on each MS device during $[s_i, s_j]$ must be at most $k_i + c_{ij}$. 
% Thus total vacant time on those $R$ devices must be at most $R \cdot (k_i + c_{ij}) \leq R \cdot (s_j^\infty - s_i^\infty)$. 
\item 
When $j$ is the favorite child of $i$, but $i$'s device is not awake when task $i$ ends, we know that the time $j$ can be ready on $i$'s device is the same as $j$'s urgent time, which means there is another task $w$ such that $s_w + k_w + c_{wj} = s_i + k_i + c_{ij}$, but $j$ is not $w$'s favorite child. We can now use exactly the same argument as in the first case to prove that vacant time on any MS device during $[s_i, s_j]$ must be at most $k_i + c_{ij}$. 

\item 
When $j$ is the favorite child of $i$, and $i$'s device is awake when task $i$ ends. Denote the time $j$ becomes ready on $i$'s device as $ready(j)$, there must exist some $j$'s predecessor $y \neq i$ such that $s_y + k_y + c_{yj} = ready(j)$. Since $i$'s device is awake when task $i$ ends, task $j$ will be scheduled on $i$'s device if it is still idle by $ready(j)$. Hence a task $w$ (which is either $j$ or an urgent task) has to be scheduled on the device of $i$ at or before $ready(j)$. We will now consider the predecessor successor pair $(y, j)$, and prove that during $[s_y, u_j]$ the vacant time on any MS machine is at most $s_j^\infty - s_y^\infty$.
\begin{itemize}
    \item 
    If $w = j$, note that $j$ is not $y$'s favorite child. Hence in the infinite device SCT, $j$ and $y$ are not on the same device. We hence conclude that  $s_j - s_y = ready(j) - s_y = k_y + c_{yj} \leq s_j^\infty - s_y^\infty$.
    \item 
    If $w \neq j$, then $w$ must be urgent (the only tasks that are allowed to be scheduled on an awake machine is the favorite child and urgent tasks). Hence at the start time of $w$, it must be the case that all MS devices are either busy or awake (because if there is a free MS device, k would have been scheduled on it). By Lemma~\ref{lem:freetime}, any busy or awake device at the start time of $w$ can only be vacant for at most $max(s_w - s_y, c_{max})$ time during $[s_y, u_j]$. Now we will upper bound $max(s_w - s_y, c_{max})$ using the facts 1) $w$ happens before $ready(j)$, but after task $i$ is completed (namely, after $s_i + k_i$) and 2) $ready(j) - s_y \geq k_y \geq c_{max}$. 
    \begin{align*}
        max(s_w - s_y, c_{max}) \leq max(ready(j) - s_y, c_{max})
        = ready(j) - s_y.
    \end{align*}
    Since in the infinite device SCT, $j$ and $y$ are not on the same device, we now conclude that the total vacant time on an MS device must be at most $ready(j) - s_y = k_y + c_{yj} \leq s_j^\infty - s_y^\infty$. 
\end{itemize}
\end{enumerate}
\end{proof}

\begin{lemma} \label{claim:sctIdleTime2}
Assume that task $j$'s favourite parent $i^*$'s device is not MS during time period $[s_i, s_j]$. Then there exists a predecessor $i$ of $j$ such that the total amount of idle time during $[s_i, s_j]$ on any device $d$ that is MS throughout the period is at most $s_j^\infty - s_i^\infty + c_{ij}$. 
\end{lemma}
\begin{proof}
As argued in Lemma~\ref{claim:sctIdleTime}, any MS device must be busy after $u_j$. Let $i$ be a predecessor of $j$ such that $s_i + k_i + c_{ij}$ is maximized (namely, $u_j = s_i + k_i + c_{ij}$). Then $u_j - s_i = k_i + c_{ij} \leq s_j^\infty - s_i^\infty + c_{ij}$ (because even with infinite number of devices, task $i$ must be fully executed before task $j$ starts). 
\end{proof}

Let $R$ be the minimum $r$ across all possible graph configurations with memory availability parameter $K$. We will use $\wmsct^{p}$ and $\wsct^{p}$ to denote the makespan of m-SCT (with memory limit) and SCT (without memory limit) with $p$ devices respectively. Analogously, we will use $\wmopt^p$ and $\wopt^p$ be the optimal makespan on $p$ devices with memory limit and with no memory limit respectively. We will use $\alpha$ to denote the approximation of the infinite device SCT (against the optimal makespan with infinite devices). 

\begin{theorem}
The makespan of m-SCT is at most $(\frac{p}{R} + \alpha) \cdot \wopt^p + \frac{(n - R)}{R} \cdot c_{max}$ for any $p$.
\end{theorem}
\begin{proof}
    Let $D_{MS}$ be the set of all devices that are MS throughout the m-SCT algorithm. We know that $|D_{MS}| = r$. It is clear that the total amount of computation time spent on devices in $D_{MS}$ is at most the sum of computational time for all tasks, which is at most $\wopt^r \cdot r$. 
    
    Now we will count the amount of time a device $d \in D_{MS}$ is idle. WLOG, let $T_1$ be the task that finishes last in m-SCT and let $T_l \rightarrow T_{l - 1} \rightarrow \cdots \rightarrow T_1$ be the chain of task in $G$ ending at $T_1$ such that $T_{l}$ is a source. Before the start time $s_l$ of $T_{l}$, all MS devices must be busy, because $T_{l}$ is urgent from time $0$ and would have been scheduled on a device as soon as it becomes idle.    
    Let $n_d$ be the number of task pairs $(i, j)$ such that $i$ is $j$'s favourite parent but $i$'s parent is not MS when task $j$ starts. By Lemma~\ref{claim:sctScheduleDifficulty},~\ref{claim:sctIdleTime} and~\ref{claim:sctIdleTime2} we know that during $[s_l, \wmsct^{p}] = [s_l, s_1 + k_1]$ ($s_1$, $k_1$ are the start time and computation time of $T_1$ respectively)\cs{k1 not defined}, the amount of time $d$ is idle is at most  
    \begin{align*}
        \paren*{\sum_{j=1}^{l-1} \paren*{s_{j}^\infty - s_{j+1}^\infty}} + n_d \cdot c_{max} \leq \wsct^\infty + n_d \cdot c_{max}. 
    \end{align*} 
    
    Summing these all up for all devices in $D_{MS}$ we get that the total idle time across all devices in $D_{MS}$ is at most $$r \cdot \wsct^\infty + (n-r) \cdot c_{max}.$$
    
    We now conclude that 
    \begin{align*}
        r \cdot \wmsct^p &\leq r \cdot \wopt^r + r \cdot \wsct^\infty + (n-r) \cdot c_{max}\\
        \Rightarrow \wmsct^p &\leq \wopt^r + \wsct^\infty + \frac{n-r}{r} \cdot c_{max}\\
        &\leq \wopt^r + \alpha \cdot \wopt^\infty + \frac{n-r}{r} \cdot c_{max}\\
        &\leq \wopt^R + \alpha \cdot \wopt^\infty + \frac{n-R}{R} \cdot c_{max} \quad \text{(because $R \leq r$)}. 
    \end{align*}
    Lastly, observe that (without memory limit), the optimal makespan with $R$ devices is at most $\frac{p}{R}$ times the optimal makespan with $p$ devices. Also, $\wopt^\infty \leq \wopt^R$. Hence $\wmsct^n \leq (\frac{p}{R} + \alpha) \cdot \wopt^p + \frac{(n - R)}{R} \cdot c_{max}$. One could minimize the RHS over all $p$ to get the best upper bound.
\end{proof}
Using the same analysis as for m-ETF, we know that 
\begin{equation*}
    R \; =  \; \left\lceil n \left( 1 - \frac{1}{(1-J)K} \right) \right\rceil.
\end{equation*}

\end{document}